\PassOptionsToPackage{table}{xcolor}
\documentclass[11pt,letterpaper]{article}
\usepackage[T1]{fontenc}
\usepackage{microtype}
\usepackage{needspace}
\makeatletter
\providecommand{\MT@suspend@tagging}{}
\makeatother
\usepackage[margin=1in]{geometry}
\usepackage[english]{babel}
\usepackage[ruled,linesnumbered,vlined]{algorithm2e}
\usepackage{algpseudocode}
\usepackage{amsmath}
\usepackage{amssymb}
\usepackage{amsthm}
\usepackage{array} % 提供 p{} 功能
\usepackage{booktabs} % 可选，提供更美观的表格线
\usepackage{framed}
\usepackage{float}
\usepackage{graphicx}
\usepackage{makecell}
\usepackage{mathtools}
\usepackage{multirow}
\usepackage[numbers]{natbib}
\usepackage{subcaption}
\usepackage{tikz}
\usepackage[table]{xcolor} % Required for colored table lines and rows
\usepackage[colorlinks=true, allcolors=blue]{hyperref}

\usetikzlibrary{arrows.meta,patterns,positioning,decorations.pathreplacing}

\usepackage{pgfplots}
\usepackage{pgfplotstable}
\pgfplotsset{
    compat=1.17,
    name nodes near coords/.style={
        every node near coord/.append style={
            name=#1-\coordindex,
            alias=#1-last,
        },
    },
    name nodes near coords/.default=coordnode,
}

\definecolor{shadecolor}{gray}{0.9}
\newtheorem{theorem}{Theorem}[section]
\newtheorem{definition}[theorem]{Definition}

\newtheorem{corollary}[theorem]{Corollary}
\newtheorem{lemma}[theorem]{Lemma}
\newtheorem{invariant}[theorem]{Invariant}
\newtheorem{example}[theorem]{Example}
\newtheorem{observation}[theorem]{Observation}
\newtheorem{proposition}[theorem]{Proposition}
\numberwithin{algocf}{section}
\numberwithin{figure}{section}
\numberwithin{table}{section}
\numberwithin{equation}{section}

\newcommand{\bX}{\mathbf{X}}
\newcommand{\bB}{\mathbf{B}}
\newcommand{\bP}{\mathbf{P}}
\newcommand{\bA}{\mathbf{A}}
\newcommand{\bc}{\mathbf{c}}

\newcommand{\MMS}{\mathsf{MMS}}

\newcommand{\Comp}{\mathsf{Comp}}
\newcommand{\PROP}{\mathsf{PROP}}
\newcommand{\Pre}{\mathsf{Pre}}
\newcommand{\Suf}{\mathsf{Suf}}

\DeclareMathOperator*{\argmax}{argmax}
\DeclareMathOperator*{\argmin}{argmin}

\hypersetup{
  pdftitle={Comparison-Based Fair Division of Indivisible Chores},
  pdfauthor={Zehan Lin, Shengxin Liu, Biaoshuai Tao, Shengwei Zhou}
}

\title{Comparison-Based Fair Division of Indivisible Chores}
\author{%
  \begin{tabular}{@{}c@{\qquad\qquad}c@{}}
    Zehan Lin & Shengxin Liu \\
    University of Macau & Harbin Institute of Technology, Shenzhen \\
    \texttt{yc47490@um.edu.mo} & \texttt{sxliu@hit.edu.cn} \\
    \noalign{\vskip 1.0em}
    Biaoshuai Tao & Shengwei Zhou \\
    Shanghai Jiao Tong University & Nanyang Technological University \\
    \texttt{bstao@sjtu.edu.cn} & \texttt{s.arthur.zhou@gmail.com}
  \end{tabular}%
}
\date{}
\begin{document}
\maketitle

\begin{abstract}
We investigate the query complexity of fairly allocating $m$ indivisible chores among $n$ agents with additive cost functions. We depart from the standard cardinal model and assume only comparison access: an algorithm is only allowed to ask an agent which of two bundles is less costly, and it never observes numerical costs.
Our first results concern \emph{proportionality up to one item} (PROP1). We design comparison-based algorithms that compute PROP1 allocations using $O(n^3\log m)$ comparison queries. When the chores are arranged in a fixed order, and allocations are required to be contiguous, we compute a contiguous PROP1 allocation using $O(n^5\log m)$ comparison queries.
Our main result concerns the \emph{maximin share} (MMS)  guarantee. We show that for any fixed number of agents $n$ and constant $\varepsilon>0$, a $\left(13/11 +\varepsilon\right)$-MMS allocation can be computed with a comparison complexity logarithmic in $m$.
Remarkably, comparison access suffices to match the state-of-the-art $13/11$ cardinal-access guarantee of Huang and Segal-Halevi up to an arbitrarily small loss.
Furthermore, our result implies that the MMS distortion of comparison access (i.e., the worst-case multiplicative loss in MMS fairness incurred by observing only comparisons rather than numerical costs) is at most $13/11$.
Finally, we show that, for three agents, an allocation satisfying \emph{envy-freeness up to one item} (EF1) can be computed using $O(\log m)$ comparison queries.
\end{abstract}

\section{Introduction}\label{sec:introduction}
Fair division, originating from the work of Steinhaus~\cite{steinhaus1948problem}, is a fundamental problem in economics and computer science.
In the context of indivisible resources, it asks how to fairly allocate a set $M$ of $m$ indivisible items among a set $N$ of $n$ agents with heterogeneous preferences.
Items that agents value positively are modeled as goods, whereas items that impose disutility are modeled as chores.
In this paper, we focus on indivisible chores: each agent $i$ has a nonnegative additive cost function $c_i$ and prefers bundles of smaller cost.

Indivisibility makes exact fairness guarantees difficult to satisfy.
One of the most prominent fairness notions is envy-freeness (EF)~\cite{foley1966resource}, which requires every agent to weakly prefer her own bundle to every other agent's bundle.
EF allocations, however, may fail to exist for indivisible items.
This has motivated relaxations such as envy-freeness up to one item (EF1)~\cite{conf/sigecom/LiptonMMS04}: for chores, any envy of agent $i$ toward agent $j$ should disappear after removing a single chore from $i$'s own bundle.
Another central family of criteria is based on share guarantees.
Proportionality (PROP) requires each agent to incur cost at most her proportional share, namely $c_i(M)/n$.
Again, PROP cannot always be guaranteed for indivisible items.
As a relaxation, proportionality up to one item (PROP1)~\cite{conitzer2017fair} uses the same benchmark $c_i(M)/n$, but allows one chore to be removed from the agent's own bundle.
It is not hard to see that EF1 implies PROP1.
In addition, an EF1 allocation always exists~\cite{conf/approx/BhaskarSV21}, and so does a PROP1 allocation.
The maximin share (MMS), introduced by Budish~\cite{conf/bqgt/Budish10}, is another widely studied notion relaxing PROP.
For chores, the MMS value $\mu_i$ of agent $i$ is the minimum, over all $n$-partitions $\bP=(P_1,\dots,P_n)$ of the chores, of the maximum bundle cost $\max_{j\in[n]}\{c_i(P_j)\}$.
MMS relaxes PROP by replacing the benchmark $c_i(M)/n$ with the weaker benchmark $\mu_i$.
An MMS allocation may not exist~\cite{conf/wine/FeigeST21}, so its approximation variants are studied.
An allocation is $\alpha$-MMS, for $\alpha\ge 1$, if each agent $i$ receives a bundle of cost at most $\alpha\mu_i$.
The best known ratio $\alpha$ where an $\alpha$-MMS allocation exists is $\alpha=13/11$ due to Huang and Segal-Halevi~\cite{conf/sigecom/HuangS23}.

% A rich body of literature focuses on computing fair allocations under cardinal information.
% Under this model, a standard assumption is that algorithms have full cardinal access to agent preferences, which means they can always query the exact cost of a bundle.
% However, collecting such precise cardinal information is often unrealistic in the real world.
% For instance, agents may find it difficult to assign reliable numerical costs to complicated tasks.
% This limitation naturally motivates the study of allocation models operating only on ordinal information.
% Furthermore, from a privacy perspective, rather than requiring the entire preference profile as static input, it is preferable to utilize an oracle-based query framework that learns preferences dynamically through a sequence of queries.

Most algorithmic work on fair division assumes cardinal access to preferences~\cite{journals/ai/AmanatidisABFLMVW23,liu2024mixed,Suksompong21}: the algorithm is given, or can query, exact numerical costs. Such access is a strong requirement, especially for chores. The disutility of a workload may depend on subjective and context-dependent factors such as effort, inconvenience, expertise, and time constraints, making it difficult for agents to assign stable numerical costs to bundles. At the same time, agents may still be able to answer comparison queries reliably, for example by deciding which of two proposed workloads is less costly. Moreover, cardinal reports can reveal sensitive information about agents' abilities, constraints, or opportunity costs. These considerations motivate ordinal elicitation models in which agents reveal only pairwise comparisons between bundles.

% For indivisible goods, Oh et al.~\cite{journals/siamdm/OhPS21} studied the value-query model, where each query asks an agent for the numerical value of a bundle.
% They showed that EF1 can be computed with $O(\log m)$ value queries for two agents with monotone utilities and for three agents with additive utilities, while stronger notions such as EF and EFX require linear query complexity even for two agents with identical additive utilities.
% Bu et al.~\cite{conf/wine/BuLLST24} introduced the more restrictive comparison-query model, where each query only asks an agent to compare two bundles.
% For additive goods, they obtained logarithmic comparison-query algorithms for PROP1 together with $1/2$-MMS, for EF1 under identical valuations, and for EF1 with three agents.
% They also proved matching $\Omega(\log m)$ lower bounds for PROP1, EF1, and any positive approximation to MMS when the number of agents is fixed.

We study the comparison-query model, introduced for indivisible goods by Bu et al.~\cite{conf/wine/BuLLST24}. In this model, a query asks an agent to compare two bundles and report which one is preferred. We consider the analogous model for indivisible chores: a query to agent $i$ consists of two bundles, and the answer reveals which bundle has smaller cost according to $c_i$. Thus, the algorithm observes only the relative order of the queried bundles and never sees numerical costs.
This restriction rules out many standard algorithmic tools. In particular, a comparison-based algorithm cannot directly compute proportional shares, MMS values, or individual item costs, all of which are heavily used by cardinal fair allocation algorithms. The resulting challenge is to understand which fairness guarantees can still be achieved, and with how many~comparisons.

% In this paper, we study the comparison-based query model introduced by Bu et al.~\cite{conf/wine/BuLLST24} for indivisible chores.
% In this model, a query simply asks an agent to compare two bundles and reveal which one incurs a lower cost, without disclosing any numerical values.
% Consequently, the algorithm only observes the relative ordering of the queried bundles for each agent.
% As a result, it is blocked from directly computing proportional shares, MMS values, or individual item costs, which are heavily relied upon by traditional fair allocation algorithms.
% This fundamental information restriction naturally leads to our first question:
% \smallskip

\begin{center}\begin{minipage}{0.97\linewidth}\begin{shaded}\noindent\textbf{Question 1. } %\textit{How many ordinal bundle comparisons are necessary and sufficient to compute fair allocations of chores?}
\textit{Which fairness guarantees for indivisible chores can be achieved under comparison access, and with how many queries?}
\end{shaded}
\end{minipage}
\end{center}

% \smallskip
% If we ignore query complexity, we can recover the full ordinal ranking by comparing all pairs of bundles and then use traditional algorithms (like envy-cycle elimination) to find an EF1 allocation.
% However, these methods handle chores one by one and therefore require at least linearly many queries in $m$.
% Such a linear dependence on $m$ becomes undesirable in large instances, where the number of chores may far exceed the number of agents.
% This leads to our second question:

A trivial approach is to elicit enough comparisons to reconstruct large parts of each agent's preference order, or to simulate algorithms that allocate chores one by one. Such approaches require at least linear, and in some cases much larger, dependence on the number $m$ of chores. This is undesirable in large instances, where $m$ may be much larger than the number of agents. %We therefore ask whether meaningful fairness guarantees can be obtained with only a very small number of comparisons.
On the other hand, in most practical applications of fair division, the number of agents $n$ is small.
For example, on the fair division website Spliddit.org~\cite{goldman2015spliddit}, the average number of agents is very close to $3$, and the largest number of agents recorded is $10$~\cite{journals/teco/CaragiannisKMPS19}.
The number of items in practical scenarios, however, is large; for example, the largest instance recorded by Spliddit.org contains roughly $1400$ items and $9$ agents~\cite{journals/teco/CaragiannisKMPS19}.
This motivates the focus on those instances with small numbers of agents and large numbers of items.

Furthermore, in many applications with large $m$ and small $n$, it is desirable to avoid checking all the items one-by-one and instead make the time complexity of the algorithm only (poly-)logarithmic with respect to the number of items.
One typical scenario is the discretizations of \emph{cake-cutting} instances.
The cake-cutting problem considers divisions of divisible resources typically modeled as the interval $[0,1]$, and each agent's valuation function can be an arbitrary real function on $[0,1]$.
As remarked in the previous literature (see, e.g., Bu et al.~\cite{BU2023103904}), to resolve the encoding issues of real functions, a line of research on the cake-cutting problem uses piecewise-constant functions to approximate the actual valuation functions.
The cake is then discretized into a large number of small ``segments'', which are viewed as indivisible items later on.
It is particularly unappealing to design an algorithm that processes these segments one-by-one; instead, an algorithm with a (poly-)logarithmic complexity (on the number of segments) enables much finer discretizations (and thus much better approximations) when requiring a reasonable time complexity.

Motivated by these, we study the following problem.

% \smallskip

\begin{center}\begin{minipage}{0.97\linewidth}\begin{shaded}\noindent\textbf{Question 2. }
%\textit{Can fair allocations of indivisible chores be computed in the comparison model using sublinear (or even logarithmic) query complexity with respect to $m$?}
%\textit{Can fair allocations of indivisible chores be computed using only logarithmically many comparison queries in the number of chores?}
\textit{For constant numbers of agents, can fair allocations of indivisible chores be computed using only (poly-)logarithmically many comparison queries in the number of chores?}
\end{shaded}
\end{minipage}
\end{center}

\subsection{Our Contributions and Techniques}
To answer these questions, we initiate a systematic study of the comparison-query complexity of fair allocation of indivisible chores. A basic benchmark in this model is logarithmic dependence on the number of chores: for fixed $n$, logarithmically many comparisons are unavoidable for several standard fairness guarantees, including PROP1, MMS, and EF1 (see Appendix~\ref{app:common-lower-bound} for details). Despite the absence of cardinal cost information, our results show that this benchmark is essentially attainable for chores. In the following, we organize our contributions around three notions: PROP1, MMS, and~EF1.

\paragraph{PROP1 via a Top-Trading Envy Graph}
Our first result concerns the computation of PROP1 allocations. The main difficulty is that, under comparison access, an algorithm cannot directly test whether a bundle satisfies PROP or PROP1, since both notions are defined in terms of the agent's proportional share, a cardinal quantity that is never revealed. For indivisible goods, Bu et al.~\cite{conf/wine/BuLLST24} overcame this obstacle by designing an overlap oracle that separates bundles that are non-PROP from bundles that are PROP1. Their algorithm then allocates bundles that are sufficiently good for their recipients while ensuring that these bundles are no longer useful for the remaining agents. This oracle-based strategy, however, is inherently asymmetric and does not carry over to chores: the analogous oracle is too weak to support the same allocation dynamics. We discuss this obstruction in Appendix~\ref{app:prop-oracles-chores}.
We adapt the top-trading envy-graph approach of Bhaskar et al.~\cite{conf/approx/BhaskarSV21}, who used it to compute EF1 allocations of chores. Our approach avoids testing PROP1 bundles directly. Instead, we maintain a partial allocation in which every allocated bundle remains proportional for its current owner, and we gradually absorb the remaining chores into these bundles. In each phase, we partition the remaining chores into balanced blocks and pair these blocks with the current partial bundles. We then construct a \emph{top-trading envy graph}: each agent points to the bundle that she finds least costly among the available paired bundles. Constructing this graph requires only $O(n)$ comparison queries per agent. Since every vertex has out-degree one, the graph contains a directed cycle. Trading along such a cycle reallocates the corresponding paired bundles so that every participating agent receives her least-costly option on the cycle, thereby preserving the proportionality invariant of the partial allocation.
The key progress analysis shows that the number of unallocated chores drops below $n$ after $O(n\log m)$ phases in total. At that point, assigning at most one remaining chore to each agent converts the proportional partial allocation into a PROP1 allocation. Since each phase uses $O(n^2)$ comparison queries, this yields the following result.

% For indivisible goods, Bu et al.~\cite{conf/wine/BuLLST24} addressed this by introducing an (overlap) oracle that distinguishes between non-PROP and PROP1 bundles.
% Based on this oracle, they designed an algorithm ensuring that allocated bundles are no longer considered proportional by the remaining agents, yet remain PROP1 for the recipients.
% However, this weaker overlap oracle does not extend to the chores setting (see Appendix~\ref{app:prop-oracles-chores} for details).
% Instead of relying on a specific oracle, our approach maintains a partial allocation of proportional bundles.
% In each round, we partition the remaining chores into balanced blocks and pair them with the current partial bundles.
% Based on this partition, we construct a top-trading envy graph where each agent points to their most preferred bundle, which requires only $O(n)$ queries per agent.
% We show that a directed cycle always exists in this graph.
% Trading along this cycle ensures that every updated agent receives her least expensive bundle, thus preserving the PROP guarantee for the partial allocation.
% This process eliminates a constant fraction of the remaining chores within $O(n\log m)$ rounds.
% Once fewer than $n$ chores remain, assigning at most one additional chore to each agent successfully upgrades the partial PROP allocation to a PROP1 allocation.

\medskip
\noindent
{\bf Result 1} (Theorem~\ref{theorem: PROP1}){\bf .}
{\em For instances with additive cost functions, PROP1 allocations can be computed using $O(n^3\log m)$ comparisons.}
\medskip

This result shows that PROP1 allocations for chores can be found without cardinal costs and without a PROP/PROP1-testing oracle. For fixed $n$, the logarithmic dependence on $m$ matches the known $\Omega(\log m)$ lower bound. We also study the contiguous variant, where the chores are arranged in a fixed linear order and each agent must receive a connected interval of items. Contiguity constraints have been widely studied in fair division, especially for EF1 under the cardinal model~\cite{conf/aaai/Igarashi23, journals/geb/BiloCFIMPVZ22, conf/aaai/GolzIMS26}\footnote{These existing results are purely existential.}. We compute such an allocation using $O(n^5\log m)$ comparison queries (Theorem~\ref{thm:contiguous_prop1}), giving the optimal logarithmic dependence on $m$ for every fixed $n$. For each agent, we cover the prefix costs by at most $n$ intervals and minimize their maximum length to construct a personal contiguous PROP1 partition. To find such a cover, we repeatedly select prefix positions according to their weights, compute a cover of the selected points, and check whether it covers all prefix points. If the check fails, we increase the weights of positions in the gaps between selected positions. Finally, Rightmost Merge combines the personal partitions without additional comparison queries.

%Our result highlights the limitations of comparison-based oracles in the chores setting and motivates a novel approach from the perspective of envy graphs.
%This framework successfully bridges the gap between the goods and chores settings and matches the $\Omega(\log m)$ lower bound for fixed $n$.
% We also investigate contiguous allocations, where chores are arranged in a linear order and each agent must receive a connected interval of items, which have been widely studied in previous work for EF1 under the cardinal model~\cite{conf/aaai/Igarashi23, journals/geb/BiloCFIMPVZ22, conf/aaai/GolzIMS26}\footnote{Note that these existing results are purely existential.}.
% We prove that comparison queries suffice to guarantee PROP1 with only a polylogarithmic dependence on $m$ (Theorem~\ref{thm:contiguous_prop1}).
% Together, these results demonstrate that PROP1 for chores requires neither the oracle design nor any numerical cost information.

\paragraph{Approximate MMS via Chore Compression}
Our second result concerns the MMS guarantee. The main challenge is again the lack of cardinal information, but the obstruction is more severe for MMS. Existing algorithms for approximate MMS allocations of chores, including the state-of-the-art $13/11$ approximation of Huang and Segal-Halevi~\cite{conf/sigecom/HuangS23}, rely on cardinal costs and on reductions to structured instances such as identical-ordering (IDO) instances. Implementing such reductions directly in the comparison model would require sorting the chores for each agent, which already costs $\Omega(nm\log m)$ comparisons and rules out logarithmic dependence on $m$.
To bypass this barrier, we introduce a compression technique for chores. Given the original instance $\mathcal I$, we construct a compressed instance $\mathcal I'$ consisting of only $f(n,\varepsilon)$ ``compressed'' chores, independent of $m$. Each compressed chore represents a carefully chosen block of original chores and has controlled cost for every agent. The compression preserves the relevant MMS structure up to a small loss: an approximate MMS allocation of the compressed instance can be lifted back to the original instance while losing only an additive $\varepsilon$ in the approximation factor.
Once the instance has been compressed to constant size, we can implement the necessary steps of the Huang--Segal-Halevi $13/11$-MMS algorithm using comparison queries over bundles of compressed chores. The computation on $\mathcal I'$ contributes only a constant depending on $n$ and $\varepsilon$, so the overall query complexity is dominated by the compression procedure. This gives the following result.

% Our second result explores the MMS guarantee under the comparison-based model.
% Most existing works on guaranteeing (approximate) MMS allocations rely on the identical ordering (IDO) reduction\footnote{Informally, this reduction transforms any instance into a canonical one where all agents share the same ranking of items.} to analyze worst-case scenarios.
% However, directly applying this technique requires sorting all chores for every agent, which takes $O(nm\log m)$ comparison queries and is too expensive.
% To overcome this challenge, we introduce a technique called \emph{Chores Compression}.
% The high-level idea is to compress the original instance $\mathcal{I}$ into a smaller instance $\mathcal{I'}$ consisting of only $f(n)$ ``compressed'' chores (and is thus independent of $m$), such that each compressed chore incurs a bounded cost for every agent.
% By adapting the 13/11-MMS algorithm~\cite{conf/sigecom/HuangS23} to our comparison model, we can find a 13/11-MMS allocation for the compressed instance $\mathcal{I'}$.
% Finally, we prove that this allocation can be mapped back to the original instance $\mathcal{I}$, with only a tiny loss in the approximation ratio depending on the compression size.
% Therefore, the main query complexity lies in chores compression and computation of $13/11$-MMS allocation for $\mathcal{I'}$, leading to the following result.

\medskip
\noindent
{\bf Result 2} (Theorem~\ref{theorem: 13_11_mms}){\bf .}
{\em For any fixed $n$ and constant $\varepsilon>0$,
$\left(13/11+\varepsilon\right)$-MMS allocations can be computed with
query \mbox{complexity}~$O(\log m)$.}
\medskip

% Our result achieves the state-of-the-art approximation ratio of MMS allocations established by Huang and Segal-Halevi~\cite{conf/sigecom/HuangS23} in $O(\log m)$ while preventing item-by-item operation without any numerical value.
% We further study the distortion of the comparison model to quantify how the loss of numerical information affects the performance.
% This distortion arises because the algorithm must handle distinct cost instances that are indistinguishable under ordinal comparisons.
% Our result directly implies a $13/11$ upper bound on the distortion, and we complement this with an example establishing a $31/30$ lower bound.
Thus, comparison access suffices to match, up to an arbitrarily small loss, the best known MMS approximation guarantee under full cardinal access. The algorithm avoids both item-by-item processing and direct access to item costs, proportional shares, or MMS values.
We also formulate and study the corresponding distortion question. The MMS distortion of comparison access measures the worst-case multiplicative loss in MMS fairness forced by observing only ordinal comparisons rather than numerical costs. Equivalently, a comparison-based algorithm must perform well simultaneously on all cardinal cost instances that induce the same comparison answers. Result 2 implies a $13/11$ upper bound on this distortion. We complement this with a lower-bound construction: two cost instances are indistinguishable under comparison access but require different MMS-optimal behavior, yielding distortion at least $31/30$.

\paragraph{EF1 via a Monotone Staircase}
%Our third result concerns EF1 for three agents. With full preference information, EF1 allocations can be obtained by classical procedures such as round-robin or envy-cycle elimination. These procedures, however, operate at the level of individual items and therefore lead to at least linear dependence on $m$ when implemented through queries. Even for the more informative value-query model, logarithmic-query EF1 algorithms are known only for a small number of agents. Achieving EF1 with logarithmically many comparison queries for chores is therefore a delicate task.
Our third result establishes the tight comparison-query complexity of EF1 in the canonical three-agent setting.
With full preference information, EF1 allocations can be obtained by classical item-by-item procedures such as round-robin or envy-cycle elimination.
The challenge here is therefore not existence, but elicitation: can one find enough structure to certify EF1 without inspecting the chores individually?
A direct simulation of these classical procedures would require at least linear dependence on $m$, whereas $\Omega(\log m)$ comparisons are unavoidable for EF1.
We show that, for three agents, this logarithmic lower bound is tight in the comparison model.
The difficulty is twofold: comparison queries reveal strictly less information than value queries, and existing goods-based logarithmic methods do not transfer directly to chores.
At a high level, we use agent $1$ as an anchor and compute a contiguous EF1 partition for the identical-cost instance induced by $c_1$.
If agents $2$ and $3$ can be assigned distinct least-costly blocks, the allocation is immediately EF1.
Otherwise, both agents prefer the same block, and we move along a \emph{monotone staircase} of nearby partitions.
This staircase preserves EF1 for agent $1$ and allows us to locate, by binary search, a preference transition that yields an EF1 allocation using only a constant number of \mbox{additional comparisons}.

% Our high-level idea is to use agent $1$ as an anchor.
% We first compute a contiguous EF1 partition for agent $1$ under identical costs.
% If agents $2$ and $3$ can be assigned distinct least-costly blocks, the allocation is immediate.
% Otherwise, they share the same least-costly block, and we gradually move chores from the two residual blocks into this common block along a monotone staircase of safe states.
% Safety preserves EF1 for agent $1$, while monotonicity allows us to locate, by binary search, the first preference transition of agents $2$ and $3$.
% This transition gives two neighboring partitions that can be converted into an EF1 allocation using only constant additional comparisons.

\medskip
\noindent
{\bf Result 3} (Theorem~\ref{theorem: EF1}){\bf .}
{\em For three agents with additive costs, EF1 allocations can be computed using $O(\log m)$ comparison queries.}
\medskip

% This result provides an $O(\log m)$ query EF1 algorithm for chores, matching the $\Omega(\log m)$ lower bound.
% It also shows a structural difference between goods and chores: while Bu et al.~\cite{conf/wine/BuLLST24} achieved a similar three-agent EF1 result for goods using ideas from cake cutting, the nature of costs prevents a direct extension.
% Instead, the monotone staircase provides the additional structure needed to achieve logarithmic complexity in the chores setting.

For fixed $n=3$, this matches the known $\Omega(\log m)$ lower bound. The result also highlights a structural difference between goods and chores. Bu et al.~\cite{conf/wine/BuLLST24} obtained an analogous three-agent EF1 result for goods using ideas inspired by cake cutting, but the direction of envy elimination is different for chores: envy is eliminated by removing a chore from the envious agent's own bundle. This prevents a direct transfer of the goods argument. The monotone staircase supplies the missing structure and enables logarithmic query complexity in the chores~setting.

\subsection{Related Work}
Due to the vast literature on the fair allocation problem, in the following, we only review the results highly related to our work.
For a comprehensive overview of other related works, please refer to the recent surveys by \cite{journals/ai/AmanatidisABFLMVW23,Suksompong21,liu2024mixed}.
\paragraph{Query Complexity for Indivisible Items}
Prior work has extensively studied the value-based query complexity of fair division for indivisible goods.
In this model, the algorithm is allowed to query an agent's value for any item bundle.
Plaut and Roughgarden~\cite{journals/siamdm/PlautR20} showed that computing an envy-freeness up to any item (EFX) allocation\footnote{For the setting of goods, an allocation is EFX if the envy can be eliminated by removing any single item from the envied agent's~bundle.} requires exponential queries in $m$, even for two agents with identical and monotone valuations.
Oh et al.~\cite{journals/siamdm/OhPS21} extended this study to other fairness notions.
For additive valuations, they established a tight $\Theta(m)$ query bound for EFX when $n = 2$, and a $\Theta(\log m)$ bound for EF1 with $n \le 3$, which becomes $O(nm)$ for $n > 3$ via the classic envy-graph procedure.
More recently, Li et al.~\cite{journals/corr/abs-2602-06361} studied a noisy value query model for indivisible goods where values are observed via noisy samples, establishing both upper and lower bounds for finding EF allocations.
For the comparison-based model, Bu et al.~\cite{conf/wine/BuLLST24} showed that an $O(\log m)$ query bound suffices for a fixed number of agents to compute PROP1 (and $1/2$-MMS) allocations, as well as three-agent EF1 allocations.
Subsequently, Yamagata and Sumita~\cite{journals/corr/YamagataS26} improved the dependence on $n$ in the query complexity of computing PROP1 allocations.
Li et al.~\cite{journals/iandc/LiMSS25} characterized the complexity of the round-robin algorithm and provided a worst-case upper bound of $O(nm \log(m/n))$ and a lower bound of $\Omega(nm + m \log m)$.
Subsequent work by Lin et al.~\cite{journals/corr/abs-2611-04753} established an $O(n^8\log m)$ comparison-query bound for EF1 allocations of both goods and chores under additive preferences.

Relevant to our work, Liu et al.~\cite{journals/corr/abs-2606-31457} studied the query complexity of EF1 and PROP1 allocations under a different query model which also uses no cardinal information.
In their query model, an allocation is proposed, and fairness violations are given as feedback.
Under this model, Liu et al. use the ellipsoid method to show that EF1 and PROP1 allocations can be found in query complexity polynomial in $n$ and $m$, even when both goods and chores are present.
Besides the apparently different query model studied, our work puts a particular focus on algorithms with query complexity that is logarithmic in $m$.

Our problem also relates to existing frameworks for allocating items that apply ordinal or partial preference information~\cite{conf/ijcai/0002021, conf/atal/0001LXZ23, journals/ior/BenadeHP25, conf/aaai/Arunachaleswaran19} or operate under the value oracle model (e.g., \cite{conf/esa/BarmanBKS20, conf/wine/BarmanV21, conf/atal/BarmanNV23}).

\paragraph{Communication Complexity}
Communication complexity studies the number of bits exchanged during an allocation protocol, without restricting the interaction to a fixed oracle format.
Since Yao's seminal work~\cite{yao1979some}, it has served as a central paradigm for quantifying the information inherently required to compute a joint outcome. This line of research has yielded powerful modern tools, such as information complexity~\cite{journals/tit/BravermanR14}, query-to-communication lifting~\cite{journals/siamcomp/GoosPW20}, and rank-based lower bounds~\cite{journals/jacm/ChattopadhyayMS20}.
In the field of fair division, this perspective has been studied in both divisible settings, such as cake cutting~\cite{conf/ec/BranzeiN19}, and indivisible settings~\cite{journals/siamcomp/PlautR20, conf/sigecom/Feige25}.
Plaut and Roughgarden~\cite{journals/siamcomp/PlautR20} introduced communication complexity to discrete fair division and established upper and lower bounds for approximate envy-freeness and proportionality beyond additive valuations.
More recently, Feige~\cite{conf/sigecom/Feige25} designed low-communication protocols for additive goods and showed that allocations satisfying both PROP1 and $\frac{n}{2n-1}$-TPS\footnote{Here, TPS stands for truncated proportional share, a proportional-share benchmark that caps the contribution of very high-value individual goods, which is stronger than MMS.} can be found with randomized communication complexity $O(n\log m)$ and deterministic communication complexity $O(n\log m\log n)$.

\paragraph{Query Complexity for Divisible Items}
Query complexity has long been central in fair division of divisible resources, where the standard framework is the Robertson--Webb (RW) model with evaluation and cut queries~\cite{books/daglib/RobertsonW98,conf/nips/BranzeiN22}.
For proportional cake cutting with connected pieces, classical moving-knife protocols use $O(n^2)$ RW queries~\cite{journals/amm/DubinsS61}, Even and Paz~\cite{journals/dam/EvenP84} achieved $O(n\log n)$ queries, and matching $\Omega(n\log n)$ lower bounds are known~\cite{journals/talg/EdmondsP11}.
Envy-freeness is harder: every general RW protocol needs $\Omega(n^2)$ queries~\cite{conf/ijcai/Procaccia09}, connected envy-free allocations cannot always be found by finite protocols for three or more agents~\cite{journals/combinatorics/Stromquist08}, and the bounded protocol of Aziz and Mackenzie~\cite{conf/stoc/AzizM16} allows disconnected pieces but has very large complexity.
The literature on divisible chores, commonly referred to as bad-cake cutting or chore-cutting, presents a closer parallel to our problem.
Notably, prior studies on envy-free and connected chore division underscore that traditional goods-based trimming techniques are inadequate, thereby necessitating the development of distinct structural insights~\cite{journals/mathmag/PetersonS02,journals/tcs/HeydrichS15}.

% \subsection{Organization}
% The remainder of this paper is organized as follows.
% Section~\ref{sec:preliminary} introduces the model, notation, and foundational fairness concepts used throughout our work.
% Section~\ref{sec:prop1} investigates PROP1 allocations across both general and contiguous settings.
% Next, Section~\ref{sec:mms} develops the compression framework for approximate MMS allocations and analyzes the resulting distortion under comparison queries.
% Following this, Section~\ref{sec:ef1} presents our logarithmic-query EF1 algorithm for three agents. Finally, Section~\ref{sec:conclusion} concludes with a discussion of open problems.

\paragraph{(Approximate) MMS Allocations for Chores}
Under the cardinal model, exact MMS allocations can be found for two agents by divide-and-choose.
However, subsequent results established that an exact MMS allocation need not exist for three agents~\cite{conf/wine/FeigeST21, conf/aaai/AzizRSW17}, where the current state-of-the-art lower bound on the approximation ratio is $44/43$ introduced by Feige et al.~\cite{conf/wine/FeigeST21}.
Consequently, a series of literature has shifted its focus toward exploring approximation guarantees~\cite{conf/aaai/AzizRSW17,journals/mp/AzizLW24,conf/sigecom/FeigeH23,journals/teco/BarmanK20,conf/sigecom/HuangL21}.
The current best approximation ratio for chores is $13/11$~\cite{conf/sigecom/HuangS23}, while no algorithm can guarantee a ratio better than $31/30$ even for four agents~\cite{journals/corr/EzraG26}.
When only ordinal information is available, the attainable guarantees are inherently weaker, with the currently known ratios being $2-1/n$~\cite{conf/aaai/AzizRSW17}, $5/3$~\cite{journals/mp/AzizLW24}, and $8/5$~\cite{conf/sigecom/FeigeH23}.

\section{Preliminaries}\label{sec:preliminary}
We study the problem of fairly allocating a set of $m$ indivisible items (chores), denoted by $M$, to a group of $n$ agents $N$.
We refer to a subset of items $B \subseteq M$ as a bundle.
Each agent $i \in N$ is associated with an additive cost function $c_i: 2^M \to \mathbb{R}^+ \cup \{0\}$, which assigns a cost to every bundle of items.
We use $\bc = (c_1, \dots, c_n)$ to denote the cost functions of all agents.
For any integer $q \geq 1$ and any $M' \subseteq M$, a $q$-partition $\bB = (B_1, \dots, B_q)$ is a collection of disjoint subsets of $M'$ such that $B_r \cap B_s = \emptyset$ for all $r \neq s$ and $\bigcup_{r \in [q]} B_r = M'$.
We use $\prod_q(M')$ to denote the set of all $q$-partitions of $M'$.
An allocation $\bX = (X_1, \dots, X_n)$ is an ordered $n$-partition of the item set $M$, where agent $i$ receives bundle $X_i$.
More generally, a partial allocation is an $n$-tuple $\bX = (X_1,\dots,X_n)$ of pairwise disjoint bundles whose union may be a proper subset of $M$.
Given a partial allocation $\bX$, we use $R = M \setminus \bigcup_{i \in N} X_i$ to denote the residual set of unallocated items.
In this paper, we distinguish the concepts of partition and allocation.
A $q$-partition consists of $q$ unordered bundles whose assignments have not been decided.
On the other hand, in an allocation, the bundles are ordered and indexed by the agents receiving the bundles.
When items are arranged in a fixed linear order, we say that a partition or allocation is \emph{contiguous} if every bundle is an interval in this order, where empty bundles are allowed.
For any integer $t \geq 1$, we use $[t]$ to denote the set $\{1, \dots, t\}$.
Moreover, for an ordered sequence $S=(s_1,\ldots,s_t)$ and $x\in\{0,\ldots,t\}$, we define $\Pre(S,x)=(s_1,\ldots,s_x)$ and $\Suf(S,x)=(s_{x+1},\ldots,s_t)$.

% \begin{definition}[Cardinal Information]
%     Cardinal information specifies that an algorithm has numerical access to the agents' cost functions, which return the exact cost $c_i(X)$ for any agent $i \in N$ and bundle $X \subseteq M$.
% \end{definition}

% \begin{definition}[Ordinal Information]
%     Ordinal information restricts access solely to the relative ranking of bundles.
%     For an agent $i \in N$, the induced preference relation $\preceq_i$ satisfies
%     \begin{equation*}
%         X \preceq_i Y \quad \Longleftrightarrow \quad c_i(X) \leq c_i(Y),
%     \end{equation*}
%     which indicates that agent $i$ weakly prefers $X$ to $Y$.
% \end{definition}

In this paper, to systematically compute fair allocations, we operate within a comparison-based framework, which gives query access to ordinal information.
\begin{definition}[Comparison Model]
    We assume algorithms access agent preferences exclusively via the following query mechanism:
    \begin{itemize}
        \item $\Comp_i(X, Y)$: Given two bundles $X,Y \subseteq M$, the query asks agent $i$ to compare their respective costs, $c_i(X)$ and $c_i(Y)$. It returns the bundle with the lower cost, breaking ties in favor of $X$ (i.e., if $c_i(X) \leq c_i(Y)$, the query returns $X$).
    \end{itemize}
\end{definition}
Note that each such comparison constitutes a single query.
Specifically, we assume that $n$ is a constant number.
The primary objective of this work is to compute fair allocations for all agents while minimizing the overall query complexity.
Next, we define the fairness notions.

\begin{definition}[EF]
    An allocation $\bX$ is envy-free (EF) if for any agents $i, j \in N$, we have $c_i(X_i) \leq c_i(X_j)$.
\end{definition}

\begin{definition}[EF1]
    An allocation $\bX$ is envy-free up to one item (EF1) if for any agents $i, j \in N$, either $X_i = \emptyset$, or there exists an item $e \in X_i$ such that $c_i(X_i \setminus \{e\}) \leq c_i(X_j)$.
\end{definition}

\begin{definition}[PROP]
    Given an instance $I = (M, N, c)$, the proportional share of each agent $i \in N$ is defined as $\PROP_i = c_i(M)/n$.
    An allocation $\bX$ is called proportional (PROP) if each agent receives a bundle with~$c_i(X_i) \leq \PROP_i$.
\end{definition}

\begin{definition}[PROP1]
\emergencystretch=2em
    An allocation $\bX$ is proportional up to one item (PROP1)~if for any agent $i \in N$, either $X_i = \emptyset$, or there exists an item $e \in X_i$ such that $c_i(X_i \setminus \{e\}) \leq \PROP_i$.
\end{definition}

\begin{definition}[MMS]
    Given a set of items $M' \subseteq M$, for any $i \in N$, her maximin share (MMS) value of items $M'$ is defined as
    \begin{equation*}
        \MMS_i(M') := \min_{\bX \in \prod_n(M')} \Bigl\{\max_{j \in N} \left\{c_i(X_j)\right\}\Bigr\},
    \end{equation*}
    where $\prod_n(M')$ denotes the set of all $n$-partitions of $M'$.
\end{definition}

For convenience, we use
$\text{MMS}_i$ to denote $\text{MMS}_i(M)$ when $M$ is clear from the~\mbox{context}.

\begin{definition}[$\alpha$-MMS]
    An allocation $\bX$ is said to satisfy the $\alpha$-approximate maximin share guarantee for some $\alpha \geq 1$ if, for every agent $i \in N$, $c_i(X_i) \leq \alpha \cdot \MMS_i$.
    Specifically, when $\alpha = 1$, the allocation is an \emph{MMS}~allocation.
\end{definition}

Note that the query complexity of computing a PROP1/EF1/$\alpha$-MMS allocation, for any fixed nontrivial $1\leq \alpha<n$, in the chores setting is $\Omega(\log (m/n))$, which simplifies to $\Omega(\log m)$ when $n$ is constant (see Appendix~\ref{app:common-lower-bound} for details).

\section{Query Complexity of PROP1}\label{sec:prop1}
In this section, we investigate the query complexity of computing PROP1 allocations.
We first present an algorithm that computes a PROP1 allocation within $O(n^3\log m)$ queries (Theorem~\ref{theorem: PROP1}).
We further explore the computation of contiguous PROP1 allocations under the comparison model and demonstrate that such allocations can be computed using $O(n^5\log m)$ queries~(Theorem~\ref{thm:contiguous_prop1}).

\subsection{Computation of PROP1 Allocations}
We first focus on the computation of PROP1 allocations.
Before detailing our algorithm, we discuss prior work for goods~\cite{conf/wine/BuLLST24} and the inherent challenges of computing PROP1 for chores under the comparison query model.

First, it is impossible to directly determine whether a bundle satisfies PROP or PROP1, due to the uncomputability of the proportional shares under the comparison model.
For the goods setting, Bu et al.~\cite{conf/wine/BuLLST24} addressed this challenge by designing an oracle that returns ``yes'' only if a bundle satisfies PROP1 and ``no'' only if it is strictly not PROP.~\footnote{We note that this is an overlap oracle, that is, a bundle can satisfy PROP1 but violate PROP at the same time. Under such cases, their oracle may return any answer among ``yes'' and ``no''.}
This allows for the construction of a bipartite graph, enabling a classical Hall's matching framework to assign PROP1 bundles to a subset of agents while leaving enough items behind for the remaining agents.
However, such an oracle is also unavailable in the chores setting, due to the counterexample presented in Appendix~\ref{app:prop-oracles-chores}.
Consequently, standard techniques relying on absolute thresholds (e.g., moving-knife procedures) or standard oracles cannot be applied, motivating the development of a fundamentally different algorithmic~framework.

Now we begin to introduce our algorithm for computing PROP1 allocations under the comparison model (see Algorithm~\ref{alg: PROP1} for details).
Generally, the algorithm consists of two phases.

\paragraph{Phase 1: Compute a partial PROP allocation}
We say that a partial allocation is a partial PROP if $c_i(X_i) \leq \PROP_i$ for every agent $i \in N$.\footnote{Here, $\PROP_i$ is defined with respect to the full input set $M$ and remains fixed throughout the~algorithm.}
In Phase 1, we maintain a partial PROP allocation round by round until the number of remaining items is strictly smaller than $n$.
During each round, we partition the remaining chores into $n$ balanced bundles $\bP = (P_1, \dots, P_n)$ and combine these bundles with the current partial allocation to form tentative bundles $\textbf{T}$.
We then construct a variant of the top-trading envy graph of Bhaskar et al.~\cite{conf/approx/BhaskarSV21}, denoted by $G$, where each agent $i$ points to an agent $f(i)$ whose tentative bundle $T_{f(i)}$ is least costly for her.\footnote{We break ties arbitrarily to ensure that the out-degree of each agent is exactly 1.}
Then we select any directed cycle, and simultaneously assign each agent $i$ on the cycle the tentative bundle $T_{f(i)}$ to which she points.
Since every selected agent receives a least-costly tentative bundle, the proportional-share invariant is~preserved.

\paragraph{Phase 2: Allocate the remaining items}
Based on the construction of Phase $1$, the partial allocation $\mathbf{X}$ satisfies PROP with at most $(n-1)$ items left unallocated.
Therefore, allocating at most one remaining item to each agent suffices to achieve the PROP1 guarantee for all agents.

\noindent\textit{Remark.}
Our algorithm builds on the top-trading envy-graph approach of Bhaskar et al.~\cite{conf/approx/BhaskarSV21} for computing EF1 allocations of chores.
Rather than applying cycle exchanges to the currently allocated bundles, we construct the graph on tentative bundles $T_i=X_i\cup P_i$, where $(P_1,\ldots,P_n)$ is a balanced partition of the remaining chores.
This allows us to maintain a partial PROP allocation and compute a PROP1 allocation using $O(n^3\log m)$ \mbox{comparison queries}.

\medskip

\begin{algorithm}[!ht]
\caption{Computation of PROP1 Allocations}
\label{alg: PROP1}
\KwIn{Instance $\mathcal{I}=(N,M,\bc)$}
Initialize $X_i \gets \emptyset$ for $i \in N$; \\
$R \gets M$; \\
\tcp{Phase 1: Compute a partial PROP allocation}
\While{$|R| \geq n$}{
    Compute a partition $\bP = (P_1, \dots, P_n)$ of $R$ where $\big| |P_i| - |P_j| \big| \leq 1$ for all $i, j \in N$; \\
    \For{$i \in N$}{
        $T_i \gets X_i \cup P_i$; \tcp{Construct a tentative allocation}
    }
    \For{$i \in N$}{
        $f(i) \gets \argmin_{j \in N}\{c_i(T_j)\}$; \\
        \tcp{Index of $i$'s favorite bundle, breaking ties arbitrarily}
    }
    Construct a directed graph $G = (N, E)$ where $E = \{(i, f(i)) \mid i \in N\}$; \\
    Pick any directed cycle $C = (i_1, i_2, \dots, i_k)$ in $G$ \tcp{By Lemma~\ref{lemma: cycle}}
    \tcp{For a self-loop at $i$, the update assigns $T_i$ to agent $i$}
    \For{$j \in C$}{
        $X_j \gets T_{f(j)}$; \\
        $R \gets R \setminus P_{f(j)}$;
    }
}
\tcp{Phase 2: Allocate the remaining items}
Let the remaining items be $R = \{e_1, \dots, e_r\}$; \\
\For{$i$ from $1$ to $r$}{
    $X_i \gets X_i \cup \{e_i\}$;
}
\KwOut{A PROP1 allocation $\mathbf{X} = (X_1, \dots, X_n)$}
\end{algorithm}

Next, we move to analyze the correctness of Algorithm~\ref{alg: PROP1}.
Since Phase $2$ is simple, our analysis will mainly focus on Phase $1$.
We first show that the top-trading envy graph $G$ always has a cycle.
\begin{lemma}\label{lemma: cycle}
    In every iteration of Phase $1$, the graph $G$ contains at least one directed~cycle.
\end{lemma}
\begin{proof}
    Since every agent has an out-degree of exactly $1$, starting from any agent and following the unique outgoing edges creates a deterministic path.
    With only $n$ agents in $G$, the Pigeonhole Principle guarantees that an agent will be revisited within $n$ steps, which forms a cycle.
    % This traversal explores at most $n$ edges, yielding an $O(n)$ time complexity.
\end{proof}

We remark that self-loops count as directed cycles of length one in Lemma~\ref{lemma: cycle}.
Next, we construct the following invariants for each round.
\Needspace{7\baselineskip}
\begin{invariant} \label{invariant: prop_invariant}
Throughout Phase~1, the following two properties hold:
\begin{itemize}

\item [(1)] $(X_1,\dots,X_n)$ forms a partition of $M \setminus R$.

\item [(2)] For every agent $i\in N$, we have $c_i(X_i)\leq \PROP_i = c_i(M)/n$.
\end{itemize}
\end{invariant}

\begin{proof}
    We prove both properties by induction over the iterations of Phase $1$.
    Initially, $X_i=\emptyset$ for every $i\in N$ and $R=M$.
    Hence, $(X_1,\dots,X_n)$ forms a partition of $M\setminus R$.
    Moreover, for every agent $i\in N$, we have $c_i(X_i)=c_i(\emptyset)=0\leq \PROP_i$.

    Now, assume that both properties hold at the beginning of a given iteration. Let $P=(P_1,\dots,P_n)$ denote the balanced partition of $R$, and let $T_i = X_i \cup P_i$ be the corresponding tentative bundle for each agent $i\in N$.
    By the induction hypothesis, $(X_1,\dots,X_n)$ is a partition of $M\setminus R$.
    Since $(P_1,\dots,P_n)$ is a partition of $R$, it immediately follows that $(T_1,\dots,T_n)$ constitutes a valid partition of $M$.
    Let $C$ be the directed cycle selected by the algorithm.
    Because $C$ is a cycle in the graph induced by $f$, the mapping $f$ permutes the vertices of $C$.
    During the simultaneous update, every agent $j\in C$ is allocated $T_{f(j)}$, while agents outside $C$ retain their previous bundles.
    Therefore, the updated bundles are:
    \begin{equation*}
        X'_j =
        \begin{cases}
            T_{f(j)}, & \text{if } j\in C,\\
            X_j, & \text{if } j\notin C.
        \end{cases}
    \end{equation*}
    The remaining items are updated to $R' = R \setminus \bigcup_{j\in C} P_{f(j)}$.
    Since $f$ is a permutation on $C$, this can be equivalently written as
    $R' = R \setminus \bigcup_{j\in C} P_j$.
    Consequently, the update assigns exactly the tentative subsets corresponding to the indices in $C$ and removes precisely these subsets from the remaining set $R$. All other previously allocated bundles and the unallocated portions of $R$ remain unchanged. Therefore, $(X'_1,\dots,X'_n)$ remains a valid partition of $M \setminus R'$.

    It remains to verify the proportional-share bound. For any agent $j\notin C$, we have $X'_j=X_j$, so the bound holds trivially by the induction hypothesis, i.e., $c_j(X'_j)=c_j(X_j)\leq \PROP_j$.

    Next, consider an agent $j\in C$. By the definition of $f(j)$, agent $j$ weakly prefers the tentative bundle $T_{f(j)}$ over all other tentative bundles in terms of cost. That is,
    \begin{equation*}
        c_j(T_{f(j)}) \leq c_j(T_i) \qquad \text{for every } i\in N.
    \end{equation*}
    This implies that the cost of $T_{f(j)}$ is bounded by the average cost of all tentative~bundles:
    \begin{equation*}
        c_j(X'_j) = c_j(T_{f(j)}) \leq \frac{1}{n}\sum_{i\in N} c_j(T_i).
    \end{equation*}
    \emergencystretch=2em
    Because $(T_1,\dots,T_n)$ forms a partition of $M$, the additivity of the cost function $c_j$ yields $\sum_{h\in N} c_j(T_h)=c_j(M)$. Substituting this into the inequality, we obtain:
    \begin{equation*}
        c_j(X'_j) \leq \frac{c_j(M)}{n} = \PROP_j.
    \end{equation*}
    Hence, the proportional-share bound is preserved after the update.
    Consequently, both properties hold throughout Phase $1$ by induction.
\end{proof}

\Needspace{6\baselineskip}
Then we turn to analyze the query complexity of Phase $1$.
\begin{lemma} \label{lemma: phase1}
    Phase $1$ terminates after $O(n \log m)$ iterations.
\end{lemma}
\begin{proof}
    Consider an arbitrary iteration of Phase 1, and let $r = |R|$ be the number of remaining items at the beginning of this iteration. The loop condition implies that $r \geq n$. The algorithm partitions $R$ into $n$ balanced bundles $(P_1, \dots, P_n)$, where each bundle contains at least $\lfloor r/n \rfloor$ items.
    Let $C$ be the directed cycle selected in this iteration according to Lemma~\ref{lemma: cycle}. Since $|C| \geq 1$, removing all bundles in $C$ from $R$ eliminates at least $\lfloor r/n \rfloor$ items.
    Utilizing the fact that $\lfloor x \rfloor \geq x/2$ for any $x \geq 1$, the number of removed items is bounded below by $\frac{r}{2n}$. Consequently, the number of remaining items after this iteration, denoted as $r'$, satisfies:
    \begin{equation*}
    r' \leq r - \frac{r}{2n} = \left(1 - \frac{1}{2n}\right) r.
    \end{equation*}
    Starting with an initial size of $|R| = m$, the number of remaining items after $q$ iterations is at most $m \left(1 - \frac{1}{2n}\right)^q$. Given that the loop terminates once $|R| < n$, and applying the standard inequality $1 - x \leq e^{-x}$, we upper bound the remaining items by:
    \begin{equation*}
    |R| \leq m \cdot \exp\left(-\frac{q}{2n}\right).
    \end{equation*}
    To guarantee that $|R| < n$, it suffices to set $q = O(n \log(m/n)) = O(n \log m)$. Therefore, Phase 1 terminates within $O(n \log m)$ iterations.
\end{proof}

By combining all the lemmas, we arrive at the following theorem.
\begin{theorem} \label{theorem: PROP1}
    For instances with additive cost functions, Algorithm~\ref{alg: PROP1} returns PROP1 allocations using $O(n^3 \log m)$ comparison queries.
\end{theorem}
\begin{proof}
    We begin by showing that the returned allocation $\bX$ is PROP1.
    By Invariant~\ref{invariant: prop_invariant}, when Phase 1 terminates, the partial allocation $\bX$ satisfies $c_i(X_i) \leq \PROP_i$ for any $i \in N$.
    Since the while-loop terminates when the number of unallocated items is strictly less than $n$ (i.e., $|R| < n$), Phase 2 assigns these remaining items to distinct agents.
    Consequently, each agent receives at most one additional item.
    By the definition of PROP1, this ensures that the final allocation $\mathbf{X}$ is a valid PROP1~allocation.

    Then, we analyze the overall comparison query complexity in Algorithm~\ref{alg: PROP1}, which appears only in Phase $1$.
    In each iteration, the only comparison queries are used to compute $f(i)$ for each agent $i \in N$.
    Fix an agent $i$.
    To compute $f(i)$, the algorithm must find a minimum-cost bundle among the $n$ tentative bundles $T_1, \dots, T_n$ according to $c_i$.
    This can be done by a standard tournament scan: start with $T_1$ as the current minimum, compare it with $T_2$, keep the lower-cost bundle, compare the current minimum with $T_3$, and so on.
    This requires exactly $n - 1$ comparison queries for agent $i$.
    Since there are $n$ agents, computing all values $f(i)$ requires $n(n-1)$ comparison queries.
    Constructing the graph $G$, finding a directed cycle, updating the bundles, and updating $R$ require no additional comparison queries.
    Therefore, each Phase 1 iteration uses at most $n(n-1)$ queries.
    By Lemma~\ref{lemma: phase1}, Phase 1 has $O(n \log m)$ iterations.
    Hence, the total number of comparison queries in Phase 1 is~$O(n^3 \log m)$.
    % This proves the theorem.
\end{proof}

\Needspace{7\baselineskip}
\subsection{Computation of Contiguous PROP1 Allocations}\label{sec:contiguous_prop1}
In this subsection, we consider the contiguous version of the PROP1 allocation problem.
In this setting, chores are assumed to be arranged in a fixed linear order
$M = (e_1, e_2, \dots, e_m)$, and each agent must receive a contiguous subsequence (possibly empty) of this order.
In the following, we present Algorithm~\ref{alg: contiguousPROP1} and show that contiguous PROP1 allocations can be computed using $O(n^5\log m)$ comparison queries.

To describe our algorithm, we use the prefix and suffix notation from Section~\ref{sec:preliminary}.
Consistent with the partition notation in Section~\ref{sec:preliminary}, let $\prod_k(S)$ denote the set of all partitions of $S$ into $k$ contiguous subsequences, where empty subsequences are allowed.
For agent $i$, we further define
\begin{equation*}
    \mu_i^k(S)=
    \max_{\mathbf{P}\in\prod_k(S)}\bigl\{\min_{j\in[k]}\{c_i(P_j)\}\bigr\}.
\end{equation*}
Equivalently, $\mu_i^k(S)$ is the largest threshold such that $S$ can be divided into $k$ contiguous parts, each having a cost of at least $\mu_i^k(S)$ for agent $i$.
By definition, we have $\mu_i^n(M) \leq \PROP_i$.
Our high-level idea is to compute, for each agent $i$, a contiguous bundle of cost $\mu_i^n(M)$ as a benchmark.
We use this benchmark to construct a personal contiguous partition in which removing the last chore from each block leaves cost at most $\mu_i^n(M)\leq\PROP_i$.

\paragraph{Algorithm Overview}
Algorithm~\ref{alg: contiguousPROP1} proceeds in two main phases.
First, for each agent $i\in N$, the subroutine $\textsc{PersonalPartition}(i)$ constructs a personal contiguous PROP1 partition (Lemma~\ref{lemma: personal_prop1}).
To achieve this, we treat the prefix costs of $M$ as ordered points on a line and show that computing $\mu_i^n(M)$ is equivalent to covering all prefix points with at most $n$ intervals while minimizing their maximum length.
Any valid cover directly induces a personal contiguous partition in which removing the last chore from each nonempty block leaves a cost of at most $\PROP_i$.
To find such a cover efficiently, we maintain weights on the prefix positions and iteratively sample $O(n^2)$ positions based on their weights.
In each round, a two-pointer dynamic program computes an optimal cover for the sampled points, and we verify whether it extends to all prefix points.
Whenever this verification fails, we double the weights of all positions within the potentially uncovered gaps between consecutive sampled points.
We prove that this procedure terminates within $O(n\log m)$ rounds (Lemma~\ref{lemma: cover_rounds}).
Since each round requires $O(n^3)$ comparisons, constructing the personal partitions across all $n$ agents takes $O(n^5\log m)$ comparisons in total.
In the second phase, the Rightmost Merge procedure aggregates these personal partitions into a global contiguous PROP1 allocation without any additional queries (Lemma~\ref{lemma: rightmost_merge}).

For ease of presentation, we first analyze the procedures with strictly positive chore costs and comparisons by cost. The lexicographic comparisons used in the pseudocode are justified by the perturbation argument in the proof of Lemma~\ref{lemma: personal_prop1}, which extends the guarantees to all nonnegative costs.
For each prefix index $j \in \{0, \dots, m\}$, we write $\operatorname{sum}_i(j) = c_i(\Pre(M, j))$ for agent $i$'s cost for the first $j$ chores, with $\operatorname{sum}_i(0) = 0$.
Here, $j$ specifies a position in the chore order, whereas $\operatorname{sum}_i(j)$ is the corresponding cumulative cost.
We emphasize that the numerical values $\operatorname{sum}_i(j)$ are used solely for analysis and are never observed by the algorithm.
By additivity, the distance between prefix points $\operatorname{sum}_i(a)$ and $\operatorname{sum}_i(b)$, where $0 \leq a \leq b \leq m$, equals the cost of the contiguous bundle,
\begin{equation*}
    \operatorname{sum}_i(b) - \operatorname{sum}_i(a) = c_i(\{e_{a+1}, \dots, e_b\}).
\end{equation*}
{\emergencystretch=2em
Consequently, comparing the lengths of two candidate intervals $[\operatorname{sum}_i(a),\allowbreak \operatorname{sum}_i(b)]$ and $[\operatorname{sum}_i(a'),\allowbreak \operatorname{sum}_i(b')]$ reduces to a comparison query between the contiguous bundles $\{e_{a+1}, \dots, e_b\}$ and $\{e_{a'+1}, \dots, e_{b'}\}$.
All bundle comparisons in Algorithms~\ref{alg: sample_cover} and~\ref{alg: personal_contiguous} are evaluated under the lexicographic order on $(c_i(B), |B|)$, which compares costs first (via $\Comp_i(B, B')$ and $\Comp_i(B', B)$) and breaks cost ties by bundle cardinality.\par}

Given a subset of prefix indices $\mathcal{S} \subseteq \{0, \dots, m\}$ containing $0$ and $m$, we define $\lambda_i(\mathcal{S})$ as the smallest $\lambda \geq 0$ such that the point set $\{\operatorname{sum}_i(j) : j \in \mathcal{S}\}$ can be covered by at most $n$ closed intervals, each of length at most $\lambda$.
For a fixed bound $\lambda$, the greedy cover starts each interval at the leftmost uncovered point and extends it $\lambda$ to the right, even if the covered points span a shorter distance.
% For any fixed length bound $\lambda$, the left-to-right greedy procedure starts an interval at the leftmost uncovered point $\operatorname{sum}_i(a)$, covers all points in $[\operatorname{sum}_i(a), \operatorname{sum}_i(a) + \lambda]$, and repeats. This procedure uses the minimum possible number of intervals.
% Indeed, any interval covering the leftmost uncovered point can be shifted to start at that point without uncovering any point to its right, and optimality follows by induction on the remaining uncovered points.

The following lemma shows that the optimal length bound for covering all prefix points coincides with the contiguous maximin benchmark $\mu_i^n(M)$.

\begin{lemma}\label{lemma: prefix_cover}
For every $i \in N$, we have $\lambda_i(\{0,\ldots,m\}) = \mu_i^n(M) \leq \PROP_i$.
\end{lemma}
\begin{proof}
\emergencystretch=2em
To see that $\lambda_i(\{0,\ldots,m\}) = \mu_i^n(M)$, observe first that any optimal contiguous $n$-partition of $M$ attaining $\mu_i^n(M)$ is determined by prefix indices $0=q_0\leq q_1\leq\cdots\leq q_n=m$ such that
\begin{equation*}
    \operatorname{sum}_i(q_r) - \operatorname{sum}_i(q_{r-1}) \geq \mu_i^n(M) \qquad \text{for all } r \in [n].
\end{equation*}
Conversely, any $n+1$ indices $0 \leq q_0 < q_1 < \cdots < q_n \leq m$ with consecutive gaps $\operatorname{sum}_i(q_r) - \operatorname{sum}_i(q_{r-1}) \geq \lambda$ yield a contiguous $n$-partition of $M$ in which every block has cost at least $\lambda$, since extending the outer endpoints to $0$ and $m$ only increases the costs of the first and last blocks.

We now prove that $\lambda_i(\{0,\ldots,m\}) = \mu_i^n(M)$.
We first consider the $n+1$ prefix points $\operatorname{sum}_i(q_0),\allowbreak \dots,\allowbreak \operatorname{sum}_i(q_n)$ attaining $\mu_i^n(M)$.
By the Pigeonhole Principle, any cover of $\{\operatorname{sum}_i(0),\allowbreak \ldots,\allowbreak \operatorname{sum}_i(m)\}$ using at most $n$ intervals must place at least two of these $n+1$ points into the same interval.
Since the distance between any two of these points is at least $\mu_i^n(M)$, the length bound must satisfy $\lambda_i(\{0,\ldots,m\}) \geq \mu_i^n(M)$.
Conversely, suppose for contradiction that the greedy cover of $\{\operatorname{sum}_i(0),\allowbreak \ldots,\allowbreak \operatorname{sum}_i(m)\}$ with length bound $\mu_i^n(M)$ requires more than $n$ intervals.
Then, the starting points of the first $n+1$ greedy intervals provide $n+1$ prefix positions whose consecutive distances are each strictly greater than $\mu_i^n(M)$, which induces a contiguous $n$-partition of $M$ with minimum block cost strictly exceeding $\mu_i^n(M)$, a contradiction.
Therefore, $\lambda_i(\{0,\ldots,m\}) = \mu_i^n(M)$.
Finally, since the minimum block cost of any $n$-partition is bounded above by the average block cost, we have $\mu_i^n(M) \leq c_i(M)/n = \PROP_i$.
\end{proof}

Lemma~\ref{lemma: prefix_cover} allows us to translate a valid interval cover directly into a personal contiguous PROP1 partition.
Specifically, suppose that $\{\operatorname{sum}_i(0), \ldots, \operatorname{sum}_i(m)\}$ is covered by $k \leq n$ intervals with length bound $\lambda \leq \PROP_i$ and starting indices $0 = \ell_1 < \ell_2 < \dots < \ell_k \leq m$, such that for every $j \in [k-1]$, the point $\operatorname{sum}_i(\ell_{j+1}-1)$ is covered by the $j$-th interval $[\operatorname{sum}_i(\ell_j), \operatorname{sum}_i(\ell_j) + \lambda]$, and $\operatorname{sum}_i(m)$ is covered by the $k$-th interval $[\operatorname{sum}_i(\ell_k), \operatorname{sum}_i(\ell_k) + \lambda]$.
Then, the contiguous blocks
\begin{equation*}
    B_j = \{e_{\ell_j+1},\ldots,e_{\ell_{j+1}}\} \quad (j \in [k-1]), \qquad B_k = \{e_{\ell_k+1},\ldots,e_m\},
\end{equation*}
form a personal contiguous PROP1 partition for agent $i$ after we set $B_{k+1} = \dots = B_n = \emptyset$ when $k < n$.
Indeed, for each $j \in [k-1]$, removing the rightmost chore $e_{\ell_{j+1}}$ from $B_j$ leaves a bundle of~cost
\begin{equation*}
    c_i(B_j \setminus \{e_{\ell_{j+1}}\})
    = \operatorname{sum}_i(\ell_{j+1}-1) - \operatorname{sum}_i(\ell_j) \leq \lambda \leq \PROP_i,
\end{equation*}
and the final block satisfies $c_i(B_k) = \operatorname{sum}_i(m) - \operatorname{sum}_i(\ell_k) \leq \lambda \leq \PROP_i$.
In particular, every nonempty block satisfies the PROP1 condition upon removing its rightmost chore.
While an optimal cover can be computed via dynamic programming, there are $O(m)$ prefix points to cover, so running the dynamic program directly over all points would cost $O(nm)$ comparison queries.
In the following, we show that it suffices to sample only $O(n^2)$ points for each round and dynamically update the sampling weights whenever the resulting sample cover fails to cover all prefix points.
This approach finds a valid global cover within $O(n\log m)$ rounds and reduces the total query complexity to $O(n^4\log m)$.
\paragraph{Computing an Optimal Cover on Sampled Prefixes}
To limit the number of comparisons, in each round we compute an optimal cover on a sampled subset $\mathcal{S}$ of $|\mathcal{S}| = O(n^2)$ prefix indices.
Given $\mathcal{S}$ listed in increasing order as $0 = s_1 < s_2 < \dots < s_{|\mathcal{S}|} = m$, the subroutine \textsc{SampleCover}$_i(\mathcal{S}, n)$ (Algorithm~\ref{alg: sample_cover}) computes a contiguous bundle $T$ whose cost equals the optimal length bound $\lambda_i(\mathcal{S})$, together with the first and last sampled indices $(\ell_j, r_j)$ covered by each interval in the resulting greedy cover of $\mathcal{S}$.

For each $k \in [n]$ and $t \in [|\mathcal{S}|]$, $\lambda_i^k(t)$ denotes the smallest $\lambda \geq 0$ such that the first $t$ sampled points $\{\operatorname{sum}_i(s_1), \ldots, \operatorname{sum}_i(s_t)\}$ can be covered by at most $k$ intervals, each of length at most $\lambda$.
Since numerical values are not directly accessible under the comparison model, Algorithm~\ref{alg: sample_cover} represents each $\lambda_i^k(t)$ by a contiguous bundle $T_t^k$ satisfying $c_i(T_t^k) = \lambda_i^k(t)$, and computes the dynamic programming table via a monotone two-pointer scan.

\begin{algorithm}[!ht]
\caption{SampleCover$_i(\mathcal S,n)$}
\label{alg: sample_cover}
\KwIn{Agent $i$, sampled prefix positions $\mathcal S$, and $n$}
\KwOut{A benchmark bundle $T$ and the first and last sampled positions $(\ell_j,r_j)$ of each covering interval}
List the sampled positions in increasing order as $0=s_1<\cdots<s_{|\mathcal{S}|}=m$\;
\tcp{$c_i(T_t^k)$: optimal length bound for the first $t$ points with at~most~$k$~intervals}
Initialize $T_0^k\gets\emptyset$ for every $k\in[n]$\;
Initialize $T_t^1\gets\{e_{s_1+1},\ldots,e_{s_t}\}$ for every $t\in[|\mathcal{S}|]$\;
\tcp{Dynamic programming via a two-pointer scan}
\For{$k=2,\ldots,n$}{
    $j\gets 1$\;
    \tcp{The last interval starts at sampled position $s_j$}
    \For{$t=1,\ldots,|\mathcal{S}|$}{
        \tcp{Advance $j$ while the last interval gives the larger~bound}
        \While{$j<t$ \textbf{and} $(c_i(T_{j-1}^{k-1}),|T_{j-1}^{k-1}|)<(c_i(\{e_{s_j+1},\ldots,e_{s_t}\}),s_t-s_j)$}{
            $j\gets j+1$\;
        }
        \tcp{At $j$, the first $k-1$ intervals determine the bound}
        $T_t^k\gets\max\bigl\{T_{j-1}^{k-1},\{e_{s_j+1},\ldots,e_{s_t}\}\bigr\}$\;
        \If{$j>1$}{
            \tcp{At $j-1$, the last interval determines the bound}
            $T_t^k\gets\min\Bigl\{T_t^k,\max\bigl\{T_{j-2}^{k-1},\{e_{s_{j-1}+1},\ldots,e_{s_t}\}\bigr\}\Bigr\}$\;
        }
    }
}
\tcp{Construct the greedy cover on $\mathcal{S}$ with benchmark bundle $T = T_{|\mathcal{S}|}^n$}
$T\gets T_{|\mathcal{S}|}^n$; $k\gets 0$; $t\gets 1$\;
\While{$t\leq |\mathcal{S}|$}{
    $k\gets k+1$; $\ell_k\gets s_t$\;
    \While{$t<|\mathcal{S}|$ \textbf{and} $(c_i(\{e_{\ell_k+1},\ldots,e_{s_{t+1}}\}),s_{t+1}-\ell_k)\leq(c_i(T),|T|)$}{
        $t\gets t+1$\;
    }
    $r_k\gets s_t$; $t\gets t+1$\;
}
\Return $(T,(\ell_1,r_1),\ldots,(\ell_k,r_k))$\;
\end{algorithm}

We establish the correctness and query complexity of \textsc{SampleCover}$_i$ in the following lemma.

\begin{lemma}\label{lemma: sample_cover}
\emergencystretch=2em
Given sampled prefix positions $\mathcal{S}$, \textsc{SampleCover}$_i(\mathcal{S}, n)$ computes a bundle $T$ with $c_i(T) = \lambda_i(\mathcal{S})$ and the greedy cover of $\mathcal{S}$ using $O(n|\mathcal{S}|)$ \mbox{comparison queries}.
\end{lemma}
\begin{proof}
We first verify the correctness of the dynamic program.
For the base cases, covering $0$ points requires length $\lambda_i^k(0) = 0$ (represented by $T_0^k = \emptyset$), and covering the first $t$ sampled points with a single interval requires length $\lambda_i^1(t) = \operatorname{sum}_i(s_t) - \operatorname{sum}_i(s_1)$ (represented by $T_t^1 = \{e_{s_1+1}, \ldots, e_{s_t}\}$).
For $k \geq 2$ and $t \in [|\mathcal{S}|]$, conditioning on the index $j \in [t]$ of the first sampled point covered by the $k$-th interval yields
\begin{equation*}
    \lambda_i^k(t) = \min_{1 \leq j \leq t}\Bigl\{
    \max\bigl\{\lambda_i^{k-1}(j-1),\, \operatorname{sum}_i(s_t) - \operatorname{sum}_i(s_j)\bigr\}\Bigr\}.
\end{equation*}
For fixed $k$ and $t$, the first term $\lambda_i^{k-1}(j-1)$ is nondecreasing in $j$, whereas the second term $\operatorname{sum}_i(s_t) - \operatorname{sum}_i(s_j)$ is nonincreasing in $j$.
Consequently, the maximum inside the recurrence is minimized either at the smallest index $j \in [t]$ satisfying $\lambda_i^{k-1}(j-1) \geq \operatorname{sum}_i(s_t) - \operatorname{sum}_i(s_j)$, or at $j - 1$ (when $j > 1$).
Furthermore, as $t$ increases to $t+1$, the second term $\operatorname{sum}_i(s_t) - \operatorname{sum}_i(s_j)$ only increases for every $j$, which implies that this optimal crossover index $j$ is nondecreasing in $t$.
Thus, advancing the pointer $j$ monotonically from $1$ to $|\mathcal{S}|$ across each row $k$ evaluates all necessary split~candidates.

Because each difference is the cost of a contiguous bundle of original chores and each minimum or maximum retains one of its input bundles, every table entry $T_t^k$ is a valid contiguous bundle of original chores satisfying $c_i(T_t^k) = \lambda_i^k(t)$.
In particular, $T = T_{|\mathcal{S}|}^n$ satisfies $c_i(T) = \lambda_i^n(|\mathcal{S}|) = \lambda_i(\mathcal{S})$.

For the query complexity, in each row $k \in \{2, \ldots, n\}$, both the outer loop index $t$ and the inner pointer $j$ increase from $1$ to $|\mathcal{S}|$, so the \texttt{while} loop performs at most $2|\mathcal{S}|$ comparisons, and updating $T_t^k$ requires $O(1)$ comparisons per $t$.
Thus, filling each row takes $O(|\mathcal{S}|)$ comparisons, and computing the entire table across all $n$ rows takes $O(n|\mathcal{S}|)$ comparisons.
Finally, the greedy scan in the second half of Algorithm~\ref{alg: sample_cover} advances $t$ from $1$ to $|\mathcal{S}|$ and uses at most $|\mathcal{S}|$ additional comparisons to identify the intervals $(\ell_1, r_1), \ldots, (\ell_k, r_k)$.
Hence, the total query complexity is $O(n|\mathcal{S}|)$.
\end{proof}

\paragraph{Constructing a Personal Contiguous PROP1 Partition}
We now describe how \textsc{PersonalPartition}$(i)$ (Algorithm~\ref{alg: personal_contiguous}) combines adaptive prefix sampling with \textsc{SampleCover}$_i$ to find a cover of all prefix points $\{\operatorname{sum}_i(0), \ldots, \operatorname{sum}_i(m)\}$.
Our construction adapts the iterative reweighting framework of~\cite{journals/jacm/Clarkson95,journals/dcg/BronnimannG95}, selecting sample positions deterministically according to cumulative weights.
Specifically, the algorithm maintains a positive weight $w(j)$ for each prefix position $j \in \{0, \ldots, m\}$, initialized to $w(j) = 1$.
These weights are updated purely combinatorially and do not depend on numerical costs.

In each round, the algorithm forms the sample set $\mathcal{S}$ by selecting prefix positions at $4n^2$ equally spaced cumulative-weight thresholds, as illustrated in Figure~\ref{fig:contiguous-sampling}(a).
Specifically, $\mathcal{S}$ includes the endpoints $0$ and $m$, together with the smallest index $j \in \{0, \ldots, m\}$ satisfying
\begin{equation*}
    \sum_{t=0}^j w(t) \geq \frac{k}{4n^2}\sum_{t=0}^m w(t)
    \qquad \text{for each } k \in [4n^2].
\end{equation*}
After removing duplicates, $\mathcal{S}$ contains at most $4n^2 + 2$ prefix positions.
By construction, any open gap $(s_r, s_{r+1})$ between consecutive sampled positions contains total weight at most $\frac{1}{4n^2}\sum_{t=0}^m w(t)$.

\begin{algorithm}[!ht]
\caption{PersonalPartition$(i)$}
\label{alg: personal_contiguous}
\KwIn{Agent $i$, ordered chores $M=(e_1,\ldots,e_m)$, and $n$}
\KwOut{A personal contiguous PROP1 partition $\bP^i$ into $n$ blocks}
Initialize $w(j)\gets 1$ for $j=0,\ldots,m$\;
\While{\textbf{true}}{
    $\mathcal S\gets\{0,m\}$\;
    \For{$k=1,\ldots,4n^2$}{
        Let $j$ be the smallest index in $\{0,\ldots,m\}$ with $\sum_{t=0}^j w(t)\geq\frac{k}{4n^2}\sum_{t=0}^m w(t)$\;
        $\mathcal S\gets\mathcal S\cup\{j\}$\;
    }
    $(T,(\ell_1,r_1),\ldots,(\ell_k,r_k))\gets\textsc{SampleCover}_i(\mathcal S,n)$\;
    \tcp{$(\ell_j,r_j)$ records the first and last sampled positions of the $j$-th covering interval}
    $B_j\gets\{e_{\ell_j+1},\ldots,e_{\ell_{j+1}}\}$ for $j \in [k-1]$\;
    $B_k\gets\{e_{\ell_k+1},\ldots,e_m\}$\;
    \If{$(c_i(B_j\setminus\{e_{\ell_{j+1}}\}),|B_j|-1)\leq(c_i(T),|T|)$ for every $j \in [k-1]$}{
        $B_j\gets\emptyset$ for $j=k+1,\ldots,n$\;
        \Return $\bP^i = (B_1,\ldots,B_n)$\;
    }
    $U\gets\bigcup_{j=1}^{k-1}\{r_j+1,\ldots,\ell_{j+1}-1\}$; \tcp{Collect the uncover point}
    $w(j)\gets 2w(j)$ for all $j\in U$\;
}
\end{algorithm}

Next, the algorithm applies \textsc{SampleCover}$_i(\mathcal{S}, n)$ to obtain the benchmark bundle $T$ (with $c_i(T) = \lambda_i(\mathcal{S})$) and the sampled interval endpoints $(\ell_1, r_1), \ldots, (\ell_k, r_k)$, where $0 = \ell_1 < \dots < \ell_k$ and $r_k = m$ (since $0, m \in \mathcal{S}$).
Using the starting positions $\ell_1, \ldots, \ell_k$, the algorithm forms the candidate contiguous blocks $B_j = \{e_{\ell_j+1}, \ldots, e_{\ell_{j+1}}\}$ for $j \in [k-1]$ and $B_k = \{e_{\ell_k+1}, \ldots, e_m\}$.
For each $j \in [k]$, every prefix point $\operatorname{sum}_i(t)$ with $\ell_j \leq t \leq r_j$ is guaranteed to be covered by the $j$-th interval $[\operatorname{sum}_i(\ell_j), \operatorname{sum}_i(\ell_j) + c_i(T)]$ by monotonicity of prefix costs. In particular, $c_i(B_k) = \operatorname{sum}_i(m) - \operatorname{sum}_i(\ell_k) \leq c_i(T)$ always holds because $r_k = m$.
Thus, the $k$ intervals cover all prefix points $\{\operatorname{sum}_i(0), \ldots, \operatorname{sum}_i(m)\}$ if and only if the rightmost prefix point $\operatorname{sum}_i(\ell_{j+1}-1)$ before each cut $\ell_{j+1}$ is covered by the $j$-th interval, which is equivalent to the condition
\begin{equation*}
    c_i(B_j \setminus \{e_{\ell_{j+1}}\}) \leq c_i(T) \qquad \text{for every } j \in [k-1].
\end{equation*}
Checking this condition requires at most $k - 1 \leq n - 1$ comparison queries.
If the condition holds, then since $\mathcal{S} \subseteq \{0, \ldots, m\}$ implies $c_i(T) = \lambda_i(\mathcal{S}) \leq \lambda_i(\{0, \ldots, m\}) = \mu_i^n(M) \leq \PROP_i$, the blocks $(B_1, \ldots, B_n)$ immediately form a personal contiguous PROP1 partition and the procedure terminates.

Otherwise, the verification fails, and any uncovered prefix point must lie strictly between the last covered sampled position $r_j$ of some interval $j < k$ and the first sampled position $\ell_{j+1}$ of the next interval.
The algorithm then collects all such inter-interval gaps into the set
\begin{equation*}
    U = \bigcup_{j=1}^{k-1} \{r_j + 1, \ldots, \ell_{j+1} - 1\}
\end{equation*}
and doubles the weight $w(j)$ of every position $j \in U$, as illustrated in Figure~\ref{fig:contiguous-sampling}(b).
Note that the set $U$ contains every uncovered position and may also contain covered positions, such as $r_j+1$ in Figure~\ref{fig:contiguous-sampling}(b).
Because each constituent set $\{r_j + 1, \ldots, \ell_{j+1} - 1\}$ is an open gap between consecutive sampled positions in $\mathcal{S}$, the total weight doubled in a failed round is at most
\begin{equation}\label{eq:gap-weight-bound}
    \sum_{j \in U} w(j) \leq \frac{k - 1}{4n^2}\sum_{j=0}^m w(j) \leq \frac{1}{4n}\sum_{j=0}^m w(j).
\end{equation}

\begin{figure}[!ht]
    \centering
    \definecolor{coverblue}{HTML}{0072B2}
    \definecolor{gaporange}{HTML}{C45100}
    \definecolor{updatepurple}{HTML}{70519B}
    \definecolor{pointgray}{HTML}{555555}
    \begin{subfigure}{\linewidth}
        \centering
        \resizebox{\linewidth}{!}{%
        \begin{tikzpicture}[x=1.15cm,y=1cm,font=\small]
            \path[use as bounding box] (-.45,-.85) rectangle (13.45,2.65);
            % An initial portion of the cumulative-weight axis.
            \draw[pointgray!55,-{Stealth[length=4pt]},line width=.55pt]
                (-.28,0)--(13.17,0);
            \foreach \idx/\blockstart/\blockend in {0/0/.7,2/3.5/4.2,4/7/7.7,5/7.7/8.4}{
                \filldraw[fill=pointgray!5,draw=pointgray!65,line width=.55pt]
                    (\blockstart,0) rectangle (\blockend,.6);
                \node at ({(\blockstart+\blockend)/2},.3) {$w(\idx)$};
            }
            \foreach \idx/\blockstart/\blockend/\labelpos in {1/.7/3.5/2.1,3/4.2/7/5.6,6/8.4/11.2/10.1}{
                \filldraw[fill=coverblue!15,draw=coverblue,line width=.8pt]
                    (\blockstart,0) rectangle (\blockend,.6);
                \node[text=coverblue] at (\labelpos,.3) {$w(\idx)$};
            }
            \node[text=pointgray] at (12,.3) {$\cdots$};
            \node[anchor=north east,text=pointgray] at (13.17,-.19) {cumulative weight};
            \foreach \idx/\pos in {1/3.15,2/6.3,3/9.45}{
                \node[text=coverblue] at (\pos,1.95) {$\displaystyle\frac{\idx}{4n^2}\sum_{t=0}^m w(t)$};
                \draw[coverblue,densely dashed,line width=.7pt]
                    (\pos,1.25)--(\pos,0);
            }
        \end{tikzpicture}%
        }
        \caption{Illustration of sampling by cumulative weight.}
        \label{fig:contiguous-weight-sampling}
    \end{subfigure}
    \par\bigskip
    \begin{subfigure}{\linewidth}
        \centering
        \resizebox{\linewidth}{!}{%
        \begin{tikzpicture}[x=1.15cm,y=1cm,font=\small,
            sampled/.style={circle,fill=black,draw=black,minimum size=4.6pt,inner sep=0pt},
            covered/.style={circle,fill=white,draw=coverblue,line width=.9pt,minimum size=4.8pt,inner sep=0pt},
            missed/.style={circle,fill=white,draw=gaporange,line width=1pt,minimum size=5.2pt,inner sep=0pt}]
            \path[use as bounding box] (-.45,-1.95) rectangle (13.45,3.2);
            % Point labels are prefix indices, not numerical cost coordinates.
            \node[sampled] at (.05,2.8) {};
            \node[anchor=west] at (.23,2.8) {sampled};
            \node[covered] at (3.8,2.8) {};
            \node[anchor=west] at (3.98,2.8) {unsampled, covered};
            \node[missed] at (9,2.8) {};
            \node[anchor=west] at (9.18,2.8) {uncovered};

            \fill[coverblue!6] (0,.15) rectangle (4,1.45);
            \fill[coverblue!6] (8.4,.15) rectangle (12.4,1.45);
            \draw[coverblue,line width=1.2pt] (0,1.45)--(4,1.45);
            \draw[coverblue,line width=1.2pt] (8.4,1.45)--(12.4,1.45);
            \foreach \pos in {0,4,8.4,12.4}{
                \draw[coverblue,line width=1.2pt] (\pos,1.35)--(\pos,1.55);
                \draw[coverblue!65,densely dashed,line width=.55pt] (\pos,.42)--(\pos,1.28);
            }
            \node[text=coverblue] at (2,1.95) {$\leq c_i(T)$};
            \node[text=coverblue] at (10.4,1.95) {$\leq c_i(T)$};
            \draw[pointgray!55,-{Stealth[length=4pt]},line width=.55pt] (-.28,.32)--(13.17,.32);
            \node[anchor=north east,text=pointgray] at (13.17,-.7) {prefix cost};

            \foreach \pos in {0,1.6,2.8,8.4,10,12.4}{\node[sampled] at (\pos,.32) {};}
            \foreach \pos in {.8,2.3,3.7,9.2,11.2}{\node[covered] at (\pos,.32) {};}
            \foreach \pos in {4.9,6.2,7.3}{\node[missed] at (\pos,.32) {};}

            \node[anchor=base] at (0,-.35) {$\ell_j$};
            \node[anchor=base] at (2.8,-.35) {$r_j$};
            \node[anchor=base] at (8.4,-.35) {$\ell_{j+1}$};
            \node[anchor=base] at (12.4,-.35) {$r_{j+1}$};
            \node[anchor=base,text=coverblue] at (3.7,-.35) {$r_j+1$};
            \node[anchor=base,text=gaporange] at (7.3,-.35) {$\ell_{j+1}-1$};
            \draw[updatepurple,decorate,decoration={brace,mirror,amplitude=4pt},line width=.8pt]
                (3.7,-.95)--(7.3,-.95);
            \node[anchor=north,text=updatepurple] at (5.5,-1.32)
                {$\{r_j+1,\ldots,\ell_{j+1}-1\}\subseteq U$};
        \end{tikzpicture}%
        }
        \caption{Coverage and the sampling gap whose weights are doubled.}
        \label{fig:contiguous-gap-update}
    \end{subfigure}
    \caption{Sampling and weight updates in \textsc{PersonalPartition}$(i)$. In (a), blocks have widths $w(t)$, and blue blocks mark sampled positions. The endpoints $0$ and $m$ are always sampled. In (b), point $t$ is placed at $\operatorname{sum}_i(t)$. After a failed verification, all positions in $U$ are reweighted, including the covered position $r_j+1$.}
    \label{fig:contiguous-sampling}
\end{figure}

To prove that \textsc{PersonalPartition}$(i)$ terminates within $O(n\log m)$ rounds, we show that every failed round doubles the weight of at least one position from a fixed reference set $\mathcal{B}^* \subseteq \{0, \ldots, m\}$ of size at most $2n$.
Specifically, we define $\mathcal{B}^* \subseteq \{0,\ldots,m\}$ as the union of the following two fixed sets of prefix positions:
\begin{itemize}
    \item a set of $n+1$ prefix positions $0 = q_0 < q_1 < \cdots < q_n = m$ attaining $\mu_i^n(M)$ in~Lemma~\ref{lemma: prefix_cover};
    \item the starting prefix positions of the greedy cover of $\{\operatorname{sum}_i(0), \ldots, \operatorname{sum}_i(m)\}$ with length bound $\mu_i^n(M)$.
\end{itemize}
Since both sets contain $0$ and the greedy cover uses at most $n$ intervals, it follows that $ |\mathcal{B}^*| \leq 2n$.
We emphasize that $\mathcal{B}^*$ is used purely for the analysis and is never computed by the algorithm.

\Needspace{6\baselineskip}
\begin{lemma}\label{lemma: cover_certificate}
If the verification fails, there exists some point $\operatorname{sum}_i(j)$ with $j \in \mathcal{B}^*$~uncovered.
\end{lemma}
\begin{proof}
\widowpenalty=10000
Since $\mathcal{S} \subseteq \{0, \ldots, m\}$, we have $\lambda_i(\mathcal{S}) \leq \lambda_i(\{0, \ldots, m\}) = \mu_i^n(M)$.
We consider two cases depending on whether this inequality is strict.

First, suppose that $\lambda_i(\mathcal{S}) < \mu_i^n(M)$.
Recall that $\mathcal{B}^*$ contains the $n+1$ prefix positions $0 = q_0 < q_1 < \dots < q_n = m$ attaining $\mu_i^n(M)$, whose consecutive distances satisfy $\operatorname{sum}_i(q_r) - \operatorname{sum}_i(q_{r-1}) \geq \mu_i^n(M) > \lambda_i(\mathcal{S})$.
Consequently, any single interval of length at most $\lambda_i(\mathcal{S})$ can cover at most one of these $n+1$ points.
By the Pigeonhole Principle, a cover consisting of at most $n$ intervals, each of length at most $\lambda_i(\mathcal{S})$, must leave at least one point $\operatorname{sum}_i(q_r)$ (with $q_r \in \mathcal{B}^*$) uncovered.

Second, suppose that $\lambda_i(\mathcal{S}) = \mu_i^n(M)$.
We compare the greedy covers of the sampled prefix points and of all prefix points with length bound $\mu_i^n(M)$, as illustrated in Figure~\ref{fig:greedy-cover-first-difference}.
Both greedy covers start their first interval at $\operatorname{sum}_i(0)$.
If the starting positions of the two covers coincide for all intervals, then the two covers are identical and cover all prefix points, contradicting our assumption.
Otherwise, we consider the first interval index at which their starting positions differ, or where the sampled cover runs out of intervals before the full cover finishes.
We denote the starting point of the full greedy cover at this index by $\operatorname{sum}_i(j^*)$, where $j^* \in \mathcal{B}^*$.
Because the preceding intervals of both covers were identical, $\operatorname{sum}_i(j^*)$ is the leftmost prefix point in $\{\operatorname{sum}_i(0), \ldots, \operatorname{sum}_i(m)\}$ lying strictly to the right of the preceding interval.
In the sampled cover, either no further intervals are placed, or the next interval starts at the leftmost uncovered \emph{sampled} point, which lies strictly to the right of $\operatorname{sum}_i(j^*)$ by the choice of this first differing interval.
Therefore, $\operatorname{sum}_i(j^*)$ is covered neither by the preceding interval nor by any subsequent interval of the sampled cover, so $j^* \in \mathcal{B}^*$ \mbox{remains uncovered}.
\end{proof}

\begin{figure}[!ht]
    \centering
    \definecolor{greedyblue}{HTML}{0072B2}
    \definecolor{greedyorange}{HTML}{C45100}
    \definecolor{greedygray}{HTML}{666666}
    \resizebox{\linewidth}{!}{%
        \begin{tikzpicture}[x=1.05cm,y=1cm,font=\small,
        sampled/.style={circle,fill=black,draw=black,minimum size=4.6pt,inner sep=0pt},
        covered/.style={circle,fill=white,draw=greedyblue,line width=.9pt,minimum size=4.8pt,inner sep=0pt},
        missed/.style={circle,fill=white,draw=greedyorange,line width=1pt,minimum size=5.6pt,inner sep=0pt}]
        \path[use as bounding box] (-4.1,-1.3) rectangle (10.7,2.95);
        \node[sampled] at (0,2.65) {};
        \node[anchor=west] at (.18,2.65) {sampled};
        \node[covered] at (3.2,2.65) {};
        \node[anchor=west] at (3.38,2.65) {unsampled, covered};
        \node[missed] at (7.6,2.65) {};
        \node[anchor=west] at (7.78,2.65) {uncovered};
        \node[anchor=east] at (-.4,1.95) {Full greedy cover};
        \node[anchor=east] at (-.4,.75) {Sampled greedy cover};

        % The last common interval has the same endpoint in both covers.
        \foreach \height in {1.95,.75}{
            \fill[greedygray!8] (0,{\height-.12}) rectangle (3,{\height+.12});
            \draw[greedygray,line width=1.2pt] (0,\height)--(3,\height);
            \foreach \pos in {0,3}{
                \draw[greedygray,line width=1pt]
                    (\pos,{\height-.13})--(\pos,{\height+.13});
            }
        }
        \draw[greedygray!65,densely dashed,line width=.6pt] (3,.12)--(3,2.1);

        % Both new intervals have length mu_i^n(M), but different starts.
        \foreach \start/\finish/\height in {4.2/7.2/1.95,6.9/9.9/.75}{
            \fill[greedyblue!12] (\start,{\height-.12}) rectangle (\finish,{\height+.12});
            \draw[greedyblue,line width=1.2pt] (\start,\height)--(\finish,\height);
            \foreach \pos in {\start,\finish}{
                \draw[greedyblue,line width=1pt]
                    (\pos,{\height-.13})--(\pos,{\height+.13});
            }
        }
        \draw[greedyorange,densely dashed,line width=.7pt] (4.2,.12)--(4.2,1.8);
        \draw[greedyblue!70,densely dashed,line width=.6pt] (6.9,.12)--(6.9,.6);

        % Both rows use the same prefix-cost coordinates.
        \draw[greedygray!60,-{Stealth[length=4pt]},line width=.55pt] (-.25,0)--(10.45,0);
        \foreach \pos in {0,1.4,2.5,6.9,7.8,9.3}{\node[sampled] at (\pos,0) {};}
        \foreach \pos in {.8,2.8,8.5,9.7}{\node[covered] at (\pos,0) {};}
        \foreach \pos in {4.9,5.7}{\node[missed] at (\pos,0) {};}
        \node[missed] at (4.2,0) {};
        \node[text=greedyorange] at (4.2,-.45) {$\operatorname{sum}_i(j^*)$};
        \node[text=greedyorange] at (4.2,-1) {$j^*\in\mathcal{B}^*$};
        % \node[text=greedyblue] at (6.9,-.45) {next sampled point};
        \node[anchor=east,text=greedygray] at (10.45,-1) {prefix cost};
    \end{tikzpicture}%
        }
    \caption{Illustration of the second case in Lemma~\ref{lemma: cover_certificate}. }
    \label{fig:greedy-cover-first-difference}
\end{figure}

{\emergencystretch=2em
\noindent Building on Lemma~\ref{lemma: cover_certificate}, we now bound the total number of failed rounds \mbox{for \textsc{PersonalPartition}$(i)$}.\par}

\begin{lemma}\label{lemma: cover_rounds}
The procedure \textsc{PersonalPartition}$(i)$ terminates after $O(n\log m)$ failed rounds.
\end{lemma}
\begin{proof}
We consider any failed round of Algorithm~\ref{alg: personal_contiguous}.
Because the verification condition fails, the greedy cover of $\mathcal{S}$ with length bound $c_i(T) = \lambda_i(\mathcal{S})$ does not cover all prefix points.
By Lemma~\ref{lemma: cover_certificate}, at least one prefix point $\operatorname{sum}_i(j^*)$ with $j^* \in \mathcal{B}^*$ is left uncovered by the intervals $[\operatorname{sum}_i(\ell_j), \operatorname{sum}_i(\ell_j) + c_i(T)]$.
Since every prefix point with index in $\bigcup_{j=1}^k \{\ell_j, \ldots, r_j\}$ is covered, and $\ell_1=0$ and $r_k=m$, any uncovered position $j^* \in \mathcal{B}^*$ must belong to the gap set $U = \bigcup_{j=1}^{k-1} \{r_j + 1, \ldots, \ell_{j+1} - 1\}$, and therefore its weight $w(j^*)$ is doubled in this round.

Suppose that $t$ failed rounds have elapsed.
Since each failed round doubles the weight of at least one element in $\mathcal{B}^*$ and all weights start at $1$, the product of the weights in $\mathcal{B}^*$ after $t$ failed rounds satisfies
$
    \prod_{j \in \mathcal{B}^*} w(j) \geq 2^t.
$
On the other hand, by~\eqref{eq:gap-weight-bound}, in each failed round the positions in $U$ account for at most a $\frac{1}{4n}$-fraction of the total weight $\sum_{j=0}^m w(j)$, so doubling their weights increases the total weight by a factor of at most $1 + \frac{1}{4n}$.
After $t$ failed rounds, the total weight is bounded above by
\begin{equation*}
    \sum_{j=0}^m w(j) \leq (m+1)\left(1 + \frac{1}{4n}\right)^t
    \leq (m+1)\exp\left(\frac{t}{4n}\right),
\end{equation*}
where the last inequality follows from $1+x \leq e^x$ with $x=1/(4n)$.
Since $|\mathcal{B}^*| \leq 2n$ and $w(j) \geq 1$ for all $j$, bounding each factor in $\prod_{j \in \mathcal{B}^*} w(j)$ by the total weight $\sum_{j=0}^m w(j)$~yields
\begin{equation*}
    2^t \leq \prod_{j \in \mathcal{B}^*} w(j) \leq \left(\sum_{j=0}^m w(j)\right)^{2n} \leq (m+1)^{2n}\exp\left(\frac{t}{2}\right).
\end{equation*}
The natural logarithms of both sides give
\begin{equation*}
    t\left(\ln 2 - \frac{1}{2}\right) \leq 2n\ln(m+1) \leq 4n\ln m.
\end{equation*}
Since $\ln 2 - 1/2$ is a constant, it follows that $t = O(n\log m)$, which ensures termination in $O(n\log m)$~rounds.
\end{proof}

Lemmas~\ref{lemma: sample_cover} and~\ref{lemma: cover_rounds} establish the following guarantee for \textsc{PersonalPartition}$(i)$.

\begin{lemma}\label{lemma: personal_prop1}
For every $i \in N$, \textsc{PersonalPartition}$(i)$ computes a contiguous PROP1 partition $\bP^i$ using $O(n^4\log m)$ comparison queries.
Furthermore, every bundle $B = \{e_l, \ldots, e_r\}$ in $\bP^i$ satisfies $c_i(B \setminus \{e_r\}) \leq \PROP_i$.
\end{lemma}
\begin{proof}
We first analyze the procedures under strictly positive chore costs with comparisons by cost.
When \textsc{PersonalPartition}$(i)$ terminates, the returned partition $(B_1, \ldots, B_n)$ satisfies $c_i(B_j \setminus \{e_{\ell_{j+1}}\}) \leq c_i(T)$ for every $j \in [k-1]$ and $c_i(B_k) \leq c_i(T)$, where $c_i(T) = \lambda_i(\mathcal{S}) \leq \mu_i^n(M) \leq \PROP_i$ by Lemma~\ref{lemma: prefix_cover}.
Thus, removing the rightmost chore from any nonempty block leaves cost at most $\PROP_i$.
For the query complexity, each call to \textsc{SampleCover}$_i(\mathcal{S}, n)$ uses $O(n|\mathcal{S}|) = O(n^3)$ comparison queries by Lemma~\ref{lemma: sample_cover}, since $|\mathcal{S}| = O(n^2)$. The verification step uses at most $n-1$ further comparisons.
By Lemma~\ref{lemma: cover_rounds}, there are $O(n\log m)$ rounds in total, resulting in $O(n^4\log m)$ comparison queries.

It remains to show that the lexicographic order $(c_i(B), |B|)$ extends this guarantee to instances with zero-cost chores.
For any two contiguous bundles $B$ and $B'$, the relative order of $c_i(B)$ and $c_i(B')$ (including whether $c_i(B) = c_i(B')$) can be determined using at most two oracle calls, $\Comp_i(B, B')$ and $\Comp_i(B', B)$, because $\Comp_i$ breaks cost ties in favor of its first argument.
Whenever $c_i(B) = c_i(B')$, the algorithm breaks the tie by comparing the known cardinalities $|B|$ and $|B'|$ (and resolves any remaining ties by a fixed deterministic rule).
Hence, simulating the lexicographic order increases the number of comparison queries by at most a factor of two.

For correctness, we consider an analytical perturbation of the instance in which each chore $e \in M$ is assigned the strictly positive cost $c_i^\varepsilon(e) = c_i(e) + \varepsilon$ for a small parameter $\varepsilon > 0$, so that $c_i^\varepsilon(B) = c_i(B) + \varepsilon |B|$ for every bundle $B \subseteq M$.
Because there are only finitely many pairs of contiguous bundles, choosing $\varepsilon > 0$ sufficiently small preserves all strict cost inequalities under $c_i$ while resolving all cost equalities according to bundle cardinality $|B|$, which coincides exactly with the lexicographic order $(c_i(B), |B|)$.
Consequently, the execution of \textsc{PersonalPartition}$(i)$ using lexicographic comparisons on the original instance is identical to that of the analyzed procedures using cost comparisons under $c_i^\varepsilon$ for all sufficiently small $\varepsilon > 0$.
Our positive-cost guarantee for $c_i^\varepsilon$ ensures that every nonempty block $B = \{e_l, \ldots, e_r\}$ in the returned partition satisfies
\begin{equation*}
    c_i(B \setminus \{e_r\}) + \varepsilon(|B| - 1)
    = c_i^\varepsilon(B \setminus \{e_r\})
    \leq \frac{c_i^\varepsilon(M)}{n}
    = \frac{c_i(M) + \varepsilon m}{n}.
\end{equation*}
As $\varepsilon \to 0^+$, this inequality yields $c_i(B \setminus \{e_r\}) \leq c_i(M)/n = \PROP_i$, completing the~proof.
\end{proof}

\begin{algorithm}[!ht]
\caption{Computation of Contiguous PROP1 Allocations}
\label{alg: contiguousPROP1}
\KwIn{Instance $\mathcal I=(N,M,\bc)$ with ordered chores $M=(e_1,\ldots,e_m)$}
\KwOut{A contiguous PROP1 allocation $\bX = (X_1, \ldots, X_n)$}
Initialize $L\gets 0$ and $X_i\gets\emptyset$ for all $i\in N$\;
\tcp{Phase 1: Construct a personal contiguous PROP1 partition for each agent}
\For{$i\in N$}{
    $\bP^i=(P^i_1,\ldots,P^i_n)\gets\textsc{PersonalPartition}(i)$\;
}
\tcp{Phase 2: Aggregate personal partitions via Rightmost Merge}
\While{$N\neq\emptyset$ \textbf{and} $L<m$}{
    $S_L\gets\{e_{L+1},\ldots,e_m\}$\;
    \For{$i\in N$}{
        Let $B_i$ be the leftmost block of $\bP^i$ intersecting $S_L$\;
        Let $r_i$ be the index of the right endpoint of $B_i$\;
    }
    $i^*\gets\argmax_{i\in N}\{r_i\}$ \tcp{Break ties arbitrarily}
    $X_{i^*}\gets\{e_{L+1},\ldots,e_{r_{i^*}}\}$\;
    $L\gets r_{i^*}$; $N\gets N\setminus\{i^*\}$\;
}
\Return $\bX = (X_1, \ldots, X_n)$\;
\end{algorithm}

\paragraph{Combining Personal Partitions via Rightmost Merge}
Phase 1 computes a personal contiguous PROP1 partition $\bP^i = (P^i_1, \dots, P^i_n)$ for each agent $i \in N$. We now describe how Phase 2 of Algorithm~\ref{alg: contiguousPROP1} aggregates these $n$ personal partitions into a single contiguous PROP1 allocation $\bX = (X_1, \dots, X_n)$.
Because different agents may place their partition cut points at different locations, we employ a greedy \emph{Rightmost Merge} procedure.
Specifically, we maintain a pointer $L \in \{0, \ldots, m\}$ marking the last allocated chore, the remaining unallocated suffix $S_L = \{e_{L+1}, \dots, e_m\}$, and the set of currently unassigned agents $N$.
Initially, $L = 0$ and all agents are unassigned.
In each round, for every unassigned agent $i \in N$, we identify the leftmost block $B_i$ in her personal partition $\bP^i$ that intersects the current suffix $S_L$, and let $r_i$ denote the index of the right endpoint of $B_i$.
The algorithm selects an unassigned agent $i^* \in \argmax_{i \in N}\{r_i\}$ whose intersecting block extends furthest to the right, assigns her the contiguous bundle $X_{i^*} = \{e_{L+1}, \dots, e_{r_{i^*}}\}$, updates $L \gets r_{i^*}$, and removes $i^*$ from~$N$.

We next prove that Rightmost Merge allocates all chores and preserves the PROP1 guarantee for every agent.

\begin{lemma}\label{lemma: rightmost_merge}
For every agent $i \in N$, either $X_i = \emptyset$ or $X_i = (e_l, \dots, e_r)$, where
\begin{equation*}
c_i(X_i \setminus \{e_r\}) \leq \PROP_i.
\end{equation*}
\end{lemma}
\begin{proof}
\begin{samepage}
We maintain the invariant that at the beginning of a round with $k$ unassigned agents, the current suffix $S_L$ intersects at most $k$ blocks of $\bP^i$ for every remaining~agent~$i$.
\par\end{samepage}

Initially, $L=0$ and $k=n$, so the initial suffix $S_0 = M$ intersects at most $n$ blocks of each personal partition, and the invariant trivially holds.
In a general round with $k$ unassigned agents, let $B_i$ be the leftmost block of $\bP^i$ intersecting $S_L$ for each remaining agent $i$, and let $r_i$ be its right endpoint.
Since $B_i$ is the leftmost intersecting block, we have $B_i \cap S_L = (e_{L+1}, \dots, e_{r_i})$.
The algorithm then selects an agent $i^* \in \argmax_{i \in N}\{r_i\}$ and allocates the bundle $X_{i^*} = (e_{L+1}, \dots, e_{r_{i^*}}) = B_{i^*} \cap S_L$.
Because $X_{i^*}$ is a suffix of the personal block $B_{i^*}$ and $B_{i^*}$ is a PROP1 bundle for agent $i^*$ by Lemma~\ref{lemma: personal_prop1}, it satisfies $c_{i^*}(B_{i^*} \setminus \{e_{r_{i^*}}\}) \le \PROP_{i^*}$.
By monotonicity, $X_{i^*} \setminus \{e_{r_{i^*}}\} \subseteq B_{i^*} \setminus \{e_{r_{i^*}}\}$ implies that
\begin{equation*}
c_{i^*}(X_{i^*} \setminus \{e_{r_{i^*}}\}) \le c_{i^*}(B_{i^*} \setminus \{e_{r_{i^*}}\}) \le \PROP_{i^*},
\end{equation*}
which means $X_{i^*}$ is PROP1 for $i^*$.
To show the invariant is preserved, consider any unassigned agent $j \neq i^*$.
Since $r_j \le r_{i^*}$, updating the pointer to $L \leftarrow r_{i^*}$ completely removes the remaining part of $B_j$ from the suffix.
This ensures that the new suffix intersects at least one fewer block of $\bP^j$, so it intersects at most $k-1$ blocks, successfully preserving the invariant.

Finally, when only one agent remains ($k=1$) and the suffix is nonempty, the invariant ensures that the suffix intersects at most one block of that agent's personal partition, so the final allocation completely exhausts all remaining chores.
If the suffix becomes empty earlier, all remaining agents receive empty bundles.
Thus, the merge phase returns a contiguous allocation where every bundle is either empty or satisfies PROP1 for its receiver.
\end{proof}

\Needspace{9\baselineskip}
\begin{samepage}
Together, these lemmas establish the following theorem.
\begin{theorem}\label{thm:contiguous_prop1}
For instances with additive cost functions, contiguous PROP1 allocations can be
computed using $O(n^5\log m)$ comparison queries.
\end{theorem}
\end{samepage}
\begin{proof}
We first apply Lemma~\ref{lemma: personal_prop1} to compute a personal contiguous PROP1 partition $\bP^i$ for each agent $i$, and then apply Rightmost Merge to the resulting partitions.
By Lemma~\ref{lemma: rightmost_merge}, this yields a complete contiguous allocation in which every agent satisfies PROP1.

Regarding the query complexity, computing each personal partition requires $O(n^4\log m)$ comparisons.
Repeating this construction for all $n$ agents therefore takes $O(n^5\log m)$ comparisons in total, while Rightmost Merge requires no additional queries.
Moreover, every comparison in the personal constructions is made between intervals of original chores (possibly including the empty interval).
For any fixed $n$, the overall number of comparisons is $O(\log m)$.
\end{proof}

\section{Query Complexity of Maximin Share}\label{sec:mms}
In this section, we turn our attention to the MMS guarantee and focus on designing an algorithm that computes approximate MMS allocations with low query complexity.
We first introduce a useful reduction and the corresponding lemma.
\begin{definition}[Identical-Order Preferences (IDO)~\cite{conf/sigecom/HuangL21}]
An instance exhibits identical-order preferences if there exists a global ordering of the chores $(e_1, e_2, \dots, e_m)$, such that for every agent $i \in N$, $c_i(e_j) \geq c_i(e_k)$ whenever $j \leq k$. That is, all agents universally agree on the ranking of chores from the most costly to the least~costly.
\end{definition}
Informally, IDO instances represent the worst-case scenarios for the maximin share guarantee.
This intuition is formalized by the following lemma, which mirrors well-known IDO reduction techniques established in the goods~setting.

\begin{lemma}[\cite{conf/sigecom/HuangL21}] \label{lemma: ido_reduction}
Any arbitrary instance $\mathcal{I}$ can be transformed into an IDO instance $\mathcal{I}'$ in polynomial time, such that any $\alpha$-MMS allocation for $\mathcal{I}'$ can be mapped back to an $\alpha$-MMS allocation for $\mathcal{I}$ in polynomial time.
\end{lemma}

Most prior works compute approximate MMS allocations by relying on IDO reduction.
However, applying this reduction directly necessitates a complete global ranking of all items, which incurs a query complexity of $O(m \log m)$ for each agent.

To bypass the use of IDO reduction, a natural approach is to construct a bipartite graph following Bu et al.~\cite{conf/wine/BuLLST24}, where the selected agents receive a bundle within their $\alpha$-MMS threshold while all other agents view the cost of this bundle as exceeding their own $\alpha$-MMS thresholds. We refer to this approach as the \emph{Residual Algorithm}. However, we show that this method suffers from a lower bound of $2n/(n+2)$ on the approximation ratio (see Appendix~\ref{app:mms-discussion}). This limitation motivates us to design a novel~technique.
To overcome this challenge, our strategy is to reduce the IDO to local instances, which capture the essential ordinal structure without fully sorting the items.
Specifically, we first compress the chores into a reduced set $M'$ of size $f(n)$, which is independent of $m$.
This phase, termed \emph{chore compression} (Section~\ref{subsection: compression}), defines a new instance $\mathcal{I'} = (N, M', \bc)$.
We show that any $\alpha$-MMS allocation in $\mathcal{I'}$ can be mapped back to the original instance $\mathcal{I}$ with a slight loss in the approximation ratio depending on the compression ratio $\beta$.
This new instance directly enables the application of the IDO reduction.
Additionally, by carefully adapting the algorithm of Huang and Segal-Halevi~\cite{conf/sigecom/HuangS23} to our setting (Section~\ref{subsection: 13/11MMS}), we successfully map the resulting allocation back to $\mathcal{I}$, achieving a final $(13/11 + \varepsilon)$-MMS allocation.

% Keep this algorithm and its following footnote on the same page.
\begin{algorithm}[H]
\caption{Computation of Meta-Chores via Chore Compression}
\label{alg: ChoreCompression}
\KwIn{Instance $\mathcal{I} = (N, M, \textbf{c})$ with a fixed order $M = (e_1, \dots, e_m)$, and an integer $k \geq n$}
\KwOut{The set of meta-chores $M'$}

\SetKwProg{Fn}{Function}{}{}
\Fn{\textsc{ComputeCutPoints}$(i, k, M)$}{
    Compute a PROP1 allocation $\bA = (A_1, \dots, A_k)$ with $k$ identical copies of agent $i$ by Algorithm~\ref{alg: PROP1}\;
    Initialize set counter $\ell \gets 1$\;

    \For{each $j \in [k]$}{
        Let $A_j = B_j \cup \{e_j\}$, where $B_j$ is assigned in Phase 1 and $e_j$ in Phase 2.

        \tcp{set $e_j = \bot$ if no such item exists\footnotemark}
        \If{$e_j \neq \bot$}{$C_{\ell} \gets \text{index of } e_j$; $\ell \gets \ell + 1$\;}
    }

    Let $\mathcal{S}$ be the collection of remaining contiguous segments in $M \setminus \{e_j \mid j \in [k], e_j \neq \bot\}$\;
    $B' \gets \argmax_{j \in [k]}\{c_i(B_j)\}$\;

    \For{each contiguous segment $S = (e_s, e_{s+1}, \dots, e_t) \in \mathcal{S}$}{\label{line:seg_loop}
        \While{$S \neq \emptyset$}{\label{line:while_loop}
            \If{$\Comp_i(S, B') = S$}{\label{line:term_cond}
                $C_{\ell} \gets t$; $\ell \gets \ell + 1$; \textbf{break}\;
            }
            \Else{
                Binary search for the smallest number $p$ such that $c_i(\Pre(S, p)) > c_i(B')$\; \label{line:bin_search}
                $C_{\ell} \gets s + p - 2$, $C_{\ell + 1} \gets s + p - 1$; $\ell \gets \ell + 2$\;
                Update $S \gets S \setminus \Pre(S, p)$\;
            }
        }
    }
    \Return Set of cut points $\textbf{C}_i \gets \{C_1, \dots, C_{\ell-1}\}$\;
}

\vspace{2mm}
$\mathbf{C} \gets \{0, m\} \cup \bigcup_{i \in N} \textsc{ComputeCutPoints}(i, k, M)$; \\
Sort and remove duplicates of $\mathbf{C}$ to obtain $(C_0, C_1, \dots, C_\ell)$; \\

$M' \gets \emptyset$\;
\For{$t \gets 1$ \KwTo $\ell$}{
    $M' \gets M' \cup \big\{ \{e_j \mid C_{t-1} < j \le C_t\} \big\}$\;
}
\Return Meta-chores set $M'$
\end{algorithm}
\footnotetext{Note that at least one of the $k$ agents receives no item in Phase $2$, since at most $k-1$ items~remain.}

% \begin{lemma}[\cite{conf/sigecom/FeigeH23}]
%     The allocation $\textbf{Y} = (Y_1, \dots, Y_n)$ returned by \textsf{PickingMetaChores} satisfies $c_i(Y_i) \leq \beta_n \cdot \mu_i$ for every agent $i \in N$, where
%     \begin{equation*}
%         \beta_n =
%         \begin{cases}
%             4/3, & \text{if } n = 2, \\
%             7/5, & \text{if } n = 3, \\
%             13/9, & \text{if } n = 4, \\
%             8/5, & \text{if } n \geq 5.
%         \end{cases}
%     \end{equation*}
% \end{lemma}

\subsection{Chore Compression} \label{subsection: compression}
In this subsection, we introduce our technique called Chore Compression.
Throughout this section, we use $\mu_i$ for agent $i$'s MMS cost in the original instance.
Let $k \geq n$ be an integer specifying the number of identical copies of each agent used in the compression procedure, and define the compression ratio $\beta = n/k \leq 1$.
In the following, we first provide an overview of Algorithm~\ref{alg: ChoreCompression}.

\paragraph{Algorithm Overview}
We fix an arbitrary linear order $M=(e_1,\dots,e_m)$, shared by all agents throughout the compression procedure\footnote{This order is independent of their preferences and requires no comparison queries.}.
Algorithm~\ref{alg: ChoreCompression} compresses the ordered chore set by constructing a common refinement of cut points generated separately for each agent.
For each agent $i$, we first call Algorithm~\ref{alg: PROP1} on $k$ identical copies of agent $i$.
The resulting $k$-partition of $M$ is not used directly as the compressed instance; instead, it serves as a guide for generating agent-specific cut points.
In particular, the items allocated in Phase 2 of Algorithm~\ref{alg: PROP1} are marked so that they become singleton blocks, while the remaining contiguous segments are further split by binary search according to the largest bundle of Phase $1$.
After computing these cut points for all agents, we take their union together with the two endpoints of the global order.
The intervals between consecutive cut points form the compressed items, which we call meta-chores.
Throughout the remaining section, we use $M'$ to denote the resulting set of meta-chores.
We show that $M'$ has bounded cardinality (Lemma~\ref{lemma: meta_count}).
Furthermore, every meta-chore is either a singleton or has a bounded cost for every agent (Lemma~\ref{lemma: meta_bound}).
\medskip

In the following, we first show that the number of meta-chores can be bounded.
\begin{lemma} \label{lemma: meta_count}
    We have $|M'| \leq n(4k-1)$.
\end{lemma}
\begin{proof}
    Fix any agent $i \in N$.
    We begin by bounding the total number of cut points returned by $\textsc{ComputeCutPoints}(i, k, M)$.
    By the framework of Algorithm~\ref{alg: PROP1}, Phase 2 terminates with fewer than $k$ remaining items.
    We use $R$ to denote those remaining items; then we have $|R| \leq k - 1$.
    Each non-empty item in $R$ is treated as a singleton and directly contributes exactly $1$ cut point.
    After removing these singleton items from the fixed global order $M$, the remaining items decompose into at most $|R| + 1 \leq k$ contiguous segments, forming the collection $\mathcal{S}$.
    We now analyze the \texttt{while} loop on Line~\ref{line:while_loop} for these segments.
    Recall that $B'$ is selected as the maximum bundle among $(B_1, \dots, B_k)$.
    Since its cost is at least proportional to $\mathcal{S}$, the algorithm can find at most $k$ disjoint prefixes with a cost of at least $c_i(B')$.
    Each binary search step (Line~\ref{line:bin_search}) identifies one such prefix and generates $2$ cut points, contributing at most $2k$ cut points across all segments.
    Conversely, when the remaining part of a segment costs less than $B'$, the terminal condition (Line~\ref{line:term_cond}) triggers to create $1$ cut point.
    This terminal branch is invoked at most once per segment, adding at most $k$ cut points in~total.

    Consequently, the total number of cut points generated within the loop over $\mathcal{S}$ (Line~\ref{line:seg_loop}) is bounded above by $2k + k = 3k$. Summing the singletons and the segment cuts, the total number of cut points generated by a single agent $i$ satisfies:
    \begin{equation*}
    |\textsc{ComputeCutPoints}(i, k, M)| \leq (k-1) + 3k = 4k-1.
    \end{equation*}

    Let $\mathbf{C}$ be the global set of unique cut points. Algorithm~\ref{alg: ChoreCompression} constructs $\mathbf{C}$ by taking the union of cut points from all agents alongside the boundaries $\{0, m\}$. Since the endpoint $m$ is inherently captured by $\textsc{ComputeCutPoints}(i, k, M)$ for every agent (either as a singleton index or as the rightmost boundary $t$ of a terminal segment), incorporating $\{0, m\}$ into the union introduces at most one genuinely new point (the origin $0$). The cardinality of $\mathbf{C}$ is thus bounded by:
    \begin{equation*}
    |\mathbf{C}| \leq 1 + \sum_{i \in N} (4k - 1) = n(4k - 1) + 1.
    \end{equation*}

    Finally, let the sorted sequence of cut points in $\mathbf{C}$ be $C_0 < C_1 < \cdots < C_\ell$.
    Since $\mathbf{C}$ contains both endpoints $0$ and $m$, we have $C_0 = 0$ and $C_\ell = m$.
    Algorithm~\ref{alg: ChoreCompression} defines one meta-chore for each pair of adjacent cut points,~namely
    \begin{equation*}
    \{e_j \mid C_{t-1} < j \leq C_t\}, \qquad t \in [\ell].
    \end{equation*}
    These intervals are pairwise disjoint, and their union is exactly $M$.
    Hence, there are exactly $\ell = |\mathbf{C}| - 1$~\mbox{meta-chores:}
    \begin{equation*}
    |M'| = \ell = |\mathbf{C}| - 1 \leq n(4k - 1),
    \end{equation*}
    which completes the proof.
\end{proof}

Next, we show that for every agent, each meta-chore is either a singleton or bounded in cost.
\begin{lemma} \label{lemma: meta_bound}
    For every agent $i \in N$ and every meta-chore $A \in M'$, either $A$ is a singleton or $c_i(A) \leq \beta \cdot \mu_i$.
\end{lemma}
\begin{proof}
    Fix any agent $i \in N$. It suffices to show that the cost of any non-singleton meta-chore $A$ is bounded above by $\beta \cdot \mu_i$.
    Recall that non-singleton meta-chores are constructed via the terminal branch on Line~\ref{line:term_cond} or the binary search step on Line~\ref{line:bin_search}. Specifically, the algorithm identifies a prefix of a remaining contiguous segment $S$ such that $c_i(\Pre(S, p)) > c_i(B')$ and $c_i(\Pre(S, p - 1)) \leq c_i(B')$. This prefix $\Pre(S, p - 1)$ is then isolated by the cut points to form a meta-chore $A$. Therefore, we immediately~have
    \begin{equation*}
    c_i(A) = c_i(\Pre(S, p-1)) \le c_i(B').
    \end{equation*}
    Note that $B'$ is selected as the maximum bundle from Phase 1.
    By the design of Algorithm~\ref{alg: PROP1}, Phase 1 outputs a partial PROP allocation, ensuring that $c_i(B') \leq c_i(M)/k = \beta\cdot c_i(M)/n \leq \beta\mu_i$.
    Therefore, combining this bound with the previous inequality yields $c_i(A) \le c_i(B') \leq \beta \cdot \mu_i$, thereby completing the~proof.
\end{proof}

Given Lemma~\ref{lemma: meta_count} and Lemma~\ref{lemma: meta_bound}, we can establish the following lemma, which shows that we can map the resulting allocation to the original instance with only a slight loss in fairness.
\begin{lemma} \label{lemma: mms_inflation}
    By treating the meta-chores $M'$ as indivisible items, for every agent $i \in N$, we have \mbox{$\MMS_i(M') \leq (1+\beta) \cdot \mu_i$}.
\end{lemma}

\begin{proof}
    Fix any agent $i \in N$. For the original instance with chore set $M$, let $\bP = (P_1, \dots, P_n)$ denote her MMS partition of $M$. By definition, we have $c_i(P_j) \leq \mu_i$ for all $j \in [n]$.
    We now construct an $n$-partition of the meta-chores $M'$.
    Based on $\bP$, we initialize $n$ empty bins.
    For each bundle $P_j \in \bP$, we place every singleton meta-chore $e \in P_j$ into bin $j$.
    Since the singleton meta-chores assigned to bin $j$ form a subset of $P_j$, their total cost is at most $c_i(P_j) \leq \mu_i$.

    Next, we allocate the non-singleton meta-chores one by one into a currently minimum-load bin.
    Before any non-singleton meta-chore is placed, the total cost of all items currently allocated across all bins is at most $\sum_{j=1}^n c_i(P_j) \leq n \cdot \mu_i$.
    Therefore, by the pigeonhole principle, there must exist a bin with a current load of at most $\mu_i$.
    By Lemma~\ref{lemma: meta_bound}, the non-singleton meta-chore being placed has a cost of at most $\beta \cdot \mu_i$.
    Consequently, after placing it into this minimum-load bin, the load of that bin becomes at most $(1+\beta) \cdot \mu_i$.
    By repeating this greedy placement, the maximum cost of any bin in the resulting $n$-partition of $M'$ is upper bounded by $(1+\beta)\cdot \mu_i$.
    By the definition of the MMS, this ensures that $\MMS_i(M') \leq (1+\beta)\cdot \mu_i$.
\end{proof}

\subsection{Achieving $13/11$-MMS by Comparison} \label{subsection: 13/11MMS}
In this subsection, we demonstrate that the state-of-the-art $13/11$-MMS allocation algorithm by Huang and Segal-Halevi~\cite{conf/sigecom/HuangS23} can be successfully adapted to our comparison model.
Before detailing the complete algorithm, we first establish several necessary technical foundations.

For ease of notation, we let $\mathcal{I} = (N, M, \bc)$ denote the initial instance.
In the following, we focus on the allocation under the IDO instance $\mathcal{I'} = (N, M', \bc)$, where $M'$ represents the set of meta-chores.
Recall that the cardinality of $M'$ is bounded by $|M'| \leq n(4k-1)$ according to Lemma~\ref{lemma: meta_count}. Since this bound depends solely on the parameters $n$ and $k$, we can apply Lemma~\ref{lemma: ido_reduction} to efficiently transform the instance over $M'$ into an equivalent IDO instance.

The primary challenge in adapting the framework of Huang and Segal-Halevi~\cite{conf/sigecom/HuangS23} to our setting lies in its heavy reliance on cardinal numerical values. To overcome this hurdle, we show that an exact MMS bundle can be computed for every agent $i \in N$ via Algorithm~\ref{alg: ExactMMS}. Specifically, by exhaustively enumerating all possible $n$-partitions of $M'$, we can exactly compute the MMS bundle for any agent using $f(n, k)$ comparison queries. Crucially, this query complexity remains completely independent of the original number of chores $m$.
\begin{algorithm}[!htbp]
\caption{Computation of the MMS Bundle}
\label{alg: ExactMMS}
\KwIn{Instance $\mathcal{I'} = (N, M')$ with cost function $c_i$}

$B \gets \emptyset$\;
\For{every $n$-partition $P = (P_1, \dots, P_n)$ of $M'$}{
    $B' \gets \argmax_{j \in [n]}\{c_i(P_j)\}$\;
    \If{$B = \emptyset$ \text{\textbf{or}} $c_i(B') < c_i(B)$}{
        $B \gets B'$\;
    }
}
\Return $B$\;
\KwOut{An MMS bundle $B$}
\end{algorithm}

\begin{lemma}
    For every agent $i \in N$, Algorithm~\ref{alg: ExactMMS} computes a bundle $B \subseteq M'$ satisfying $c_i(B) = \MMS_i(M')$ using $f(n, k)$ comparison queries, where $f$ is a function independent of $m$.
\end{lemma}

\begin{proof}
    Fix any agent $i \in N$. For each $n$-partition $\bP = (P_1, \dots, P_n)$ of $M'$, Algorithm~\ref{alg: ExactMMS} identifies a maximum-cost bundle $B'$ using $O(n)$ comparison queries.
    It then compares $B'$ with the running best bundle, retaining the one with the smaller cost. Consequently, after enumerating all possible $n$-partitions, the returned bundle $B$~satisfies:
    \begin{equation*}
    c_i(B) = \min_{\bP \in \prod_n(M')}\bigl\{\max_{j \in [n]} \{c_i(P_j)\}\bigr\} = \MMS_i(M').
    \end{equation*}
    Since the total number of $n$-partitions of $M'$ depends solely on the cardinality $|M'|$, which is bounded by $|M'| \leq n(4k-1)$ according to Lemma~\ref{lemma: meta_count}, the number of comparisons is strictly bounded by a function $f(n, k)$ that is independent of $m$.
    This completes the proof.
\end{proof}
With Algorithm~\ref{alg: ExactMMS} established, we are now ready to present our modified algorithm for computing $13/11$-MMS allocations under the comparison model.
We show that by carefully updating a sequence of benchmark threshold bundles with monotonically increasing costs, all numerical operations can be fully substituted by ordinal comparison queries.
To formalize this approach, we recall several key structural properties established by Huang and Segal-Halevi~\cite{conf/sigecom/HuangS23}. Our analysis centers on the perspective of the last assigned agent, denoted by $\omega$. The high-level intuition is that when an allocation attempt fails and leaves certain items unallocated, the instance possesses crucial structural invariants from the viewpoint of $\omega$.

\begin{lemma}[\cite{conf/sigecom/HuangS23}] \label{lemma: FFDfail}
    Consider an execution of the HFFD algorithm on an IDO instance $\mathcal{I}=(N,M',c)$ with a threshold vector $(h_1, \dots, h_n)$.
    % satisfying $h_i \geq \MMS_i(M')$ for all $i \in N$.
    If the execution terminates in failure, then we have $h_{\omega} < 13/11\cdot  \MMS_{\omega}(M')$.
\end{lemma}

\begin{proof}
    Following the reduction framework in~\cite{conf/sigecom/HuangS23}, any failed HFFD execution maps directly to a MultiFit counterexample from the perspective of the last assigned agent $\omega$. Because the worst-case approximation ratio of MultiFit is tightly bounded by $13/11$ for chore allocation, it follows immediately that a failed run cannot feature a threshold satisfying $h_{\omega} \geq 13/11 \cdot  \MMS_{\omega}(M')$.
\end{proof}

\begin{invariant} \label{lemma: HFFD_invariant}
    Throughout the execution of Algorithm~\ref{alg: ComparisonHFFD}, the invariant $c_i(H_i) \leq 13/11 \cdot \MMS_i(M')$ holds for every agent $i \in N$.
\end{invariant}

\begin{proof}
    Initially, for each agent $i \in N$, we set $H_i$ to be the bundle with the exact MMS cost computed by Algorithm~\ref{alg: ExactMMS}, thereby satisfying $c_i(H_i) = \MMS_i(M') \leq 13/11 \cdot  \MMS_i(M')$.
    Now, suppose this invariant holds at the start of an arbitrary iteration of the \texttt{while} loop on Line~\ref{line: FFDwhile}. If \textsc{CompareHFFD} returns success, the algorithm terminates without updating any thresholds.

    Otherwise, the procedure fails and returns the last assigned agent $\omega$. According to Lemma~\ref{lemma: FFDfail}, this scenario corresponds precisely to a failed HFFD execution with thresholds $h_i = c_i(H_i)$, which implies $h_\omega = c_\omega(H_\omega) < 13/11 \cdot \MMS_{\omega}(M')$.
    In this case, Algorithm~\ref{alg: ComparisonHFFD} updates $H_\omega$ to the next strictly larger cost class in $T_\omega$, denoted by $H_{\omega}^+$.

    We claim that $c_\omega(H_\omega^+) \leq 13/11 \cdot \MMS_{\omega}(M')$.
    Assume for contradiction that $c_\omega(H_\omega^+) > 13/11 \cdot \MMS_{\omega}(M')$.
    By the definition of $H_\omega^+$, there exists no subset $S \subseteq M'$~satisfying
    \begin{equation*}
    h_\omega < c_\omega(S) \leq \frac{13}{11} \cdot \MMS_{\omega}(M').
    \end{equation*}

    Since all bundles evaluated during \textsc{CompareHFFD} are subsets of $M'$, replacing the threshold $h_\omega$ with the numerical threshold $13/11 \cdot \MMS_{\omega}(M')$ does not alter any comparison outcome involving agent $\omega$.
    Specifically, the cost of every tested bundle is either at most $h_\omega$ or strictly greater than $13/11 \cdot \MMS_{\omega}(M')$.
    Given that the thresholds of all other agents remain unchanged, the entire execution of the algorithm proceeds identically and still fails with $\omega$ as the last assigned agent.
    This yields a contradiction to Lemma~\ref{lemma: FFDfail}, which shows that the last assigned agent must possess a threshold no more than $13/11 \cdot \MMS_{\omega}(M')$, thereby completing the~proof.
\end{proof}
\begin{lemma} \label{lemma: HFFD_terminate}
Algorithm~\ref{alg: ComparisonHFFD} terminates after $f(n,k)$ comparison queries and returns a $13/11$-MMS allocation for the IDO instance $\mathcal{I'}$.
\end{lemma}

\begin{proof}
    We first establish termination and query complexity.
    In each failed iteration of the main loop, the algorithm advances the threshold of exactly one agent to its next distinct cost class in $T_i$.
    By Lemma~\ref{lemma: HFFD_invariant}, no threshold ever exceeds $13/11 \cdot \MMS_i(M')$.
    Since the number of distinct cost classes for any agent is bounded by $|T_i| \leq 2^{|M'|} \leq 2^{n(4k-1)}$, each threshold can be updated at most $2^{n(4k-1)}$ times, bounding the total number of failed iterations solely by a function of $n$ and $k$.

    Furthermore, each execution of $\textsc{CompareHFFD}$ requires at most $O(n^2|M'|)$ comparison queries, as it scans at most $|M'|$ remaining meta-chores across $n$ rounds, testing at most $n$ agents per chore.
    Combined with the preprocessing required to compute exact MMS bundles and sort the threshold cost classes, the overall comparison query complexity is strictly bounded by a function $f(n,k)$ that is independent of $m$.

    For correctness, termination implies that $\textsc{CompareHFFD}$ successfully returns a full meta-chore allocation $\textbf{Y}=(Y_1,\dots,Y_n)$.
    By the design of the assignment rule and Lemma~\ref{lemma: HFFD_invariant}, every agent $i \in N$ satisfies $c_i(Y_i) \leq c_i(H_i) \leq 13/11 \cdot \MMS_i(M')$, confirming that $\textbf{Y}$ is a $13/11$-MMS allocation for the IDO instance~$\mathcal{I'}$.
\end{proof}

\begin{algorithm}[!ht]
\caption{Computation of $13/11$-MMS Allocations under Comparison}
\label{alg: ComparisonHFFD}
\KwIn{IDO instance $I'=(N,M',c)$, and an integer $k \geq n$}

\SetKwProg{Fn}{Function}{}{}
\Fn{\textsc{CompareHFFD}$(M', H_1, \dots, H_n)$}{
    $R \gets M'$, $N' \gets N$\;
    $Y_i \gets \emptyset$ for all $i \in [n]$\;
    \For{$i \gets 1$ \KwTo $n$}{
        $T \gets \emptyset$\;
        \For{each $e \in R$ in non-increasing order}{
            \If{there exists $j \in N'$ such that $\Comp_j(T \cup \{e\}, H_j) = T \cup \{e\}$}{
                $T \gets T \cup \{e\}$\;
            }
        }
        Choose any $j \in N'$ such that $\Comp_j(T, H_j) = T$ \tcp*{Break ties by~index}
        $Y_j \gets T$, $N' \gets N' \setminus \{j\}$, $R \gets R \setminus T$\;
    }
    \If{$R = \emptyset$}{
        \Return $(\textbf{true}, \mathbf{Y})$\;
    } \Else{
        Let $\omega$ be the unique agent assigned in round $n$\;
        \Return $(\textbf{false}, \omega)$\;
    }
}

\vspace{2mm}
\tcp{Main execution}
\For{each agent $i \in N$}{
    $T_i \gets \text{all subsets of } M' \text{ sorted by increasing } c_i$ \tcp*{Store distinct cost classes}
    $H_i \gets \text{MMS bundle computed by Algorithm~\ref{alg: ExactMMS}}$\;
}
Initialize $X_i \gets \emptyset$ for all $i \in N$\;

\While{\textbf{true}}{ \label{line: FFDwhile}
    $(\textit{flag}, \textit{res}) \gets \textsc{CompareHFFD}(M', H_1,\dots,H_n)$\;
    \If{$\textit{flag} = \textbf{true}$}{
        $\mathbf{Y} \gets \textit{res}$\;
        \For{each $i \in N$}{
            \tcp{Expand the meta-chores in $Y_i$ back to the original~chores}
            \For{each meta-chore $S \in Y_i$ with $S = (e_s, \dots, e_t)$}{
                $X_i \gets X_i \cup \{e_s, \dots, e_t\}$\;
            }
        }
        \Return $\bX$\;
    }
    \Else{
        $\omega \gets \textit{res}$\;
        Update $H_\omega$ to be the next strictly larger threshold bundle in $T_{\omega}$\;
    }
}
\KwOut{A $(13/11 + \varepsilon)$-MMS allocation $\mathbf{X}$ for $\mathcal{I}$}
\end{algorithm}

\begin{theorem} \label{theorem: 13_11_mms}
\emergencystretch=2em
For additive chore allocation instances, given any fixed $n$ and constant $\varepsilon>0$,  $\left(13/11+\varepsilon\right)$-MMS allocations can be computed with query \mbox{complexity}~$O(\log m)$.
\end{theorem}
\begin{proof}
    Fix $n$ and $\varepsilon>0$.
    Choose an integer $k \geq \max\left\{n,\left\lceil 13n/(11\varepsilon)\right\rceil\right\}$, and set $\beta = n /k$.
    Then, $\beta\leq 1$ and by the choice of $k$, we have $\beta \leq 11\varepsilon / 13$.
    Since $n$ and $\varepsilon$ are fixed, this choice of $k$ is a constant independent of $m$.

    Start with the original instance $\mathcal{I}=(N,M,c)$.
    We first run Algorithm~\ref{alg: ChoreCompression} to obtain the meta-chore set $M'$. By Lemma~\ref{lemma: meta_count}, $|M'| \leq n(4k - 1)$ and by Lemma~\ref{lemma: mms_inflation}, we have $\MMS_i(M') \leq (1+\beta) \cdot \MMS_i(M)$ for any $i \in N$. Then we apply the IDO reduction of Lemma~\ref{lemma: ido_reduction} to the compressed instance.
    Since IDO reduction preserves approximation guarantees when mapping allocations back, it suffices to compute a $13/11$-MMS allocation for the resulting IDO meta-instance.
    Finally, expand each allocated meta-chore into its corresponding contiguous block of original chores. Let $X_i$ be the expanded bundle assigned to agent $i$. By additivity, expansion preserves the cost of each meta-chore bundle, so we can guarantee
    \begin{equation*} c_i(X_i) = c_i(Y_i) \leq \frac{13}{11} \cdot \MMS_i(M') \leq \frac{13}{11} (1+\beta) \cdot \MMS_i(M) \leq (\frac{13}{11} +\varepsilon)\cdot \MMS_i(M).
    \end{equation*}
    Thus, we can conclude that $\bX$ is a $(13/11 + \varepsilon)$-MMS allocation for the original instance.
    
    It remains to analyze the query complexity. The compression step uses $O(\log m)$ comparison queries, while all subsequent enumeration, sorting of threshold classes, MMS-bundle computation, and calls to $\textsc{CompareHFFD}$ depend only on $|M'|$. Hence, the total comparison complexity is $O(\log m) + f(n, k)$, where $f(n, k)$ is independent of $m$.
    Since $n$ and $\varepsilon$ are fixed, this is $O_{n,\varepsilon}(\log m)$, which is logarithmic in $m$.
\end{proof}
\noindent\textit{Remark.} The dependence on $n$ and $\varepsilon$ in Theorem~\ref{theorem: 13_11_mms} can be made explicit. Choose $k=\max\left\{n,\left\lceil 13n/(11\varepsilon)\right\rceil\right\}$, and let $s=|M'|\leq n(4k-1)$. The compression step uses $O(nk^3\log m)$ comparison queries. For each agent, sorting all $2^s$ subsets of the IDO meta-instance into ordered cost classes uses $O(s2^s)$ comparison queries. Using these stored rankings, the MMS-bundle computation and all calls to \textsc{CompareHFFD} can be implemented without further comparison queries. Including the IDO reduction, the total comparison-query complexity is therefore $O\left(nk^3\log m+ns2^s\right)=O\left(nk^3\log m+2^{O(nk)}\right)$. In particular, for $0<\varepsilon\leq 1$, this yields
% \begin{equation*}
$O\left(\frac{n^4}{\varepsilon^3}\log m+2^{O(n^2/\varepsilon)}\right)$.
% \end{equation*}

\subsection{The Distortion of Comparison Model} \label{subsection: distortition}
In this subsection, we delve into the informational limits of the comparison model by evaluating its distortion under the MMS benchmark.
Having established that the state-of-the-art approximation ratio can be guaranteed via bounded comparison queries, we now dismiss both query and time complexity constraints to focus exclusively on the intrinsic performance boundaries dictated by the~model.

This shift in focus motivates a fundamental inquiry: what are the exact tight upper and lower bounds on the performance gap when transitioning from cardinal utilities to ordinal comparison rankings?
To precisely quantify and mathematically structure this performance degradation, we formalize the concept of distortion in what~follows.

Informally, distortion measures the worst-case performance loss caused by using ordinal rankings instead of underlying numerical costs.
Since multiple cardinal cost profiles can produce the same rankings, a ranking-only algorithm is forced to output a single allocation for all of them.
Distortion captures the maximum factor by which this blind allocation underperforms the full-information optimal allocation.

For a full bundle ranking profile $R$, $C(R)$ denotes the set of all additive cardinal cost profiles that induce $R$.
For any allocation $\bX$ and a cardinal instance $\mathcal{I} \in C(R)$, the approximation ratio of $\bX$ with respect to the maximin share is defined as
\begin{equation*}
\rho(\bX, \mathcal{I}) =
\max\left\{1,\max_{i \in N}\left\{\frac{c_i(X_i)}{\MMS_i^\mathcal{I}}\right\}\right\},
\end{equation*}
where $\MMS_i^\mathcal{I}$ is agent $i$'s MMS value in the cardinal instance $\mathcal{I}$.

The optimal approximation ratio under complete information is
\begin{equation*}
\alpha^*(\mathcal{I}) = \min_{\bX'}\{\rho(\bX', \mathcal{I})\}.
\end{equation*}

\begin{definition}[Distortion]\label{def:distortion}
The distortion of a deterministic ranking-only algorithm $\mathcal{M}$, which maps each ranking profile $R$ to an allocation $\mathcal{M}(R)$, is defined as
\begin{equation*}
\text{Distortion}(\mathcal{M}) = \sup_{R} \sup_{\mathcal{I} \in C(R)} \frac{\rho(\mathcal{M}(R), \mathcal{I})}{\alpha^*(\mathcal{I})}.
\end{equation*}
\end{definition}

Observe that our previous approximation result guarantees a bounded distortion across all instances, establishing an upper bound of $13/11$.
Indeed, if computation is unrestricted, the comparison-based $13/11$-MMS algorithm can be run on the original instance rather than on the compressed meta-instance.
\begin{corollary}
    The upper bound of distortion is at most $13/11$.
\end{corollary}

To complement this upper bound, we next provide a lower bound of $31/30$.
Before proving the lower bound, we note that a valid indistinguishability construction cannot merely hide whether the optimal ratio is $1$ or $31/30$.
Full bundle rankings determine exact MMS feasibility.
\begin{observation}\label{obs:exact-mms-ordinal}
Fix a full bundle ranking profile $R$.
If there exists an MMS allocation for at least one cardinal instance in $C(R)$, then this allocation remains valid for every cardinal instance in $C(R)$.
\end{observation}
\begin{proof}
For an agent $i$, write $S\preceq_i T$ if $S$ is weakly cheaper than $T$ according to $R_i$.
For each $n$-partition $\bP=(P_1,\ldots,P_n)$, its worst bundle for agent $i$ is the maximum element of $\{P_1,\ldots,P_n\}$ under $\preceq_i$.
The MMS threshold is the minimum, under $\preceq_i$, among these worst bundles over all partitions.
Thus, the family of bundles whose cost is at most agent $i$'s MMS value is determined by $R_i$ alone.
Consequently, the set of allocations $\bX$ satisfying $c_i(X_i)\leq \MMS_i$ for every agent $i$ is determined by $R$ alone.
\end{proof}

Our construction uses two boundary profiles with different bundle rankings, which we perturb into a common strict ranking profile.
\begin{lemma}\label{lemma:distortion-lower-bound}
\emergencystretch=2em
    For deterministic ranking-only algorithms, the distortion is at least~$31/30$.
\end{lemma}
\begin{proof}
We present an indistinguishability construction on $N=[4]$ and $M=\{e_1,\ldots,e_{11}\}$, obtained by perturbing the MMS lower-bound instance of Ezra and Garbuz~\cite{journals/corr/EzraG26}.
With $N$ and $M$ fixed, we identify each instance with its cost profile, which lists the agents' additive cost functions in order.
For a profile $J$, $c_i^J$ denotes the cost function of agent $i$.
The two rows below specify the cost functions $c_1$ and $c_2$ used in the construction:
\begin{equation*}
\begin{array}{c|rrrrrrrrrrr}
 &e_1&e_2&e_3&e_4&e_5&e_6&e_7&e_8&e_9&e_{10}&e_{11}\\\hline
 c_1&19&40&31&37&16&37&25&40&25&31&59\\
 c_2&7&15&9&13&5&13&8&15&8&9&18
\end{array}
\end{equation*}
The cost profile $(c_1,c_1,c_2,c_2)$ corresponds to the instance of Ezra and Garbuz with its last two rows scaled by $1/3$.
The four MMS values are $(90,90,30,30)$, as realized by the partitions $(\{e_1,e_2,e_3\},\allowbreak\{e_4,e_5,e_6\},\allowbreak\{e_7,e_8,e_9\},\allowbreak\{e_{10},e_{11}\})$ and $(\{e_4,e_7,e_{10}\},\allowbreak\{e_1,e_5,e_{11}\},\allowbreak\{e_2,e_8\},\allowbreak\{e_3,e_6,e_9\})$, whose four bundles have equal costs of $90$ under $c_1$ and $30$ under $c_2$, respectively.
By the result of Ezra and Garbuz, every allocation on $(c_1,c_1,c_2,c_2)$ has an approximation ratio of at least $31/30$.

The additive cost function $c_2'$ agrees with $c_2$ on all chores except $e_7$, where $c_2'(e_7) = c_2(e_7) + 1 = 9$.
We define two boundary cost profiles
\begin{equation*}
J_1 = (c_1,c_1,c_2,c_2'),
\qquad
J_2 = (c_1,c_1,c_2',c_2).
\end{equation*}
Since the total cost of all chores under $c_2'$ is $121$ and all item costs are integers, the MMS value under $c_2'$ is at least $\lceil 121/4 \rceil = 31$.
Moreover, the partition $(\{e_4,e_7,e_{10}\},\allowbreak\{e_1,e_5,e_{11}\},\allowbreak\{e_2,e_8\},\allowbreak\{e_3,e_6,e_9\})$ achieves a maximum bundle cost of $31$ under $c_2'$.
The four MMS values are therefore $(90,90,30,31)$ for $J_1$ and $(90,90,31,30)$ for $J_2$.
We consider the two allocations
\begin{equation*}
\begin{aligned}
\bA &= (\{e_1,e_8,e_{10}\},\{e_2,e_7,e_9\},\{e_4,e_5,e_6\},\{e_3,e_{11}\}),\\
\bB &= (\{e_1,e_8,e_{10}\},\{e_2,e_7,e_9\},\{e_3,e_{11}\},\{e_4,e_5,e_6\}).
\end{aligned}
\end{equation*}
Under $J_1$, allocation $\bB$ gives the four agents costs $(90,90,27,31)$, each at most the corresponding MMS value.
It is therefore an MMS allocation, whereas $\bA$ has approximation ratio $31/30$ because agent $3$ receives a bundle of cost $31$.
Under $J_2$, the roles are symmetric: $\bA$ is an MMS allocation with agent costs $(90,90,31,27)$, whereas $\bB$ has approximation ratio $31/30$ because agent $4$ receives a bundle of cost $31$.

Crucially, this obstruction extends beyond $\bA$ and $\bB$ to all possible allocations.
For any allocation $\bX$, some agent incurs a cost of at least $31/30$ times her MMS value in $(c_1,c_1,c_2,c_2)$.
If this agent is among the first three agents, both her cost function and her MMS value remain unchanged in $J_1$, whereas if she is the fourth agent, they remain unchanged in $J_2$.
Combined with the upper bound attained by $\bA$ and $\bB$, this~gives
\begin{equation*}
    \min_{\bX}\bigl\{\max\{\rho(\bX,J_1),\rho(\bX,J_2)\}\bigr\} = \frac{31}{30}.
\end{equation*}
Hence, every allocation incurs an approximation ratio of at least $31/30$ on at least one of the two boundary profiles.

It remains to show that $J_1$ and $J_2$ can be perturbed into genuinely indistinguishable instances sharing a single strict bundle-ranking profile.
Because all costs under $c_2$ are integers and increasing the cost of $e_7$ by $1$ alters any bundle-cost difference by at most $1$, the functions $c_2$ and $c_2'$ never strictly order a pair of bundles in opposite directions.
Consequently, for every agent $i$ and any pair of bundles $S,T \subseteq M$, we have
\begin{equation*}
\bigl(c_i^{J_1}(S)-c_i^{J_1}(T)\bigr) \bigl(c_i^{J_2}(S)-c_i^{J_2}(T)\bigr) \geq 0.
\end{equation*}
To break ties without reversing any strict comparison of either profile, we define an additive cost profile $P$ by assigning each agent $i\in[4]$ the following cost for chore $e_j$, where $j\in[11]$:
\begin{equation*}
c_i^P(e_j) = 2048\bigl(c_i^{J_1}(e_j)+c_i^{J_2}(e_j)\bigr) + 2^{j-1}.
\end{equation*}
For each agent, the sum of a bundle's costs under $J_1$ and $J_2$ is an integer.
A nonzero difference between two such sums therefore has magnitude at least $1$, or at least $2048$ after scaling.
The added costs $2^{j-1}$ have total $2047$ and give distinct sums for distinct bundles.
Thus, $P$ induces a strict full bundle-ranking profile $R$ that preserves every strict comparison in either $J_1$ or $J_2$.

For $\ell\in\{1,2\}$ and $\delta\in(0,1)$, we define $J_\ell^\delta$ by taking a weighted average of each item cost under $J_\ell$ and $P$:
\begin{equation*}
c_i^{J_\ell^\delta}(e_j)
    = (1-\delta)c_i^{J_\ell}(e_j) + \delta c_i^P(e_j),
    \qquad i\in[4],\ j\in[11].
\end{equation*}
For every $\delta\in(0,1)$, both perturbed profiles $J_1^\delta$ and $J_2^\delta$ induce $R$, as this averaging preserves strict comparisons while resolving ties according to $P$.
Moreover, because the sets of allocations and partitions are finite, both $\rho(\bX, \mathcal{I})$ and $\alpha^*(\mathcal{I})$ are continuous in the cost profile $\mathcal{I}$ within neighborhoods of the strictly positive profiles $J_1$ and $J_2$.
Combined with $\alpha^*(J_1) = \alpha^*(J_2) = 1$ (attained by $\bB$ and $\bA$, respectively), this continuity yields
\begin{equation*}
\lim_{\delta \to 0^+} \min_{\bX}\left\{\max_{\ell \in \{1,2\}}\left\{\frac{\rho(\bX, J_\ell^\delta)}{\alpha^*(J_\ell^\delta)}\right\}\right\} = \frac{31}{30}.
\end{equation*}
Given $R$, any deterministic ranking-only algorithm must return the same allocation across both $J_1^\delta$ and $J_2^\delta$ for all $\delta\in(0,1)$.
Consequently, its worst-case distortion is at least the minimum inside the limit above for every $\delta\in(0,1)$, and therefore at least $31/30$ as $\delta \to 0^+$.
\end{proof}

\section{Query Complexity of EF1}\label{sec:ef1}
In this section, we demonstrate that an EF1 allocation for three agents can be achieved using $O(\log m)$~queries.
We begin by analyzing the subproblem for identical instances, which serves as the foundation of our overall approach.
Specifically, we introduce the witness EF1 guarantee needed by our further construction.

\begin{definition}[Identical Witness EF1]
\emergencystretch=2em
    For an additive cost function $c$, a three-partition $(X_1,X_2,X_3)$ is an \emph{identical witness EF1} partition if every nonempty bundle $X_i$ has a witness item $e_i\in X_i$, s.t.
    \begin{equation*}
        c(X_i \setminus \{e_i\})\leq c(X_j)
        \qquad\text{for all } i,j\in [3].
    \end{equation*}
    Furthermore, we order each nonempty bundle internally by placing its witness item~first.
\end{definition}

Note that Yamagata and Sumita~\cite{journals/corr/YamagataS26} study the related problem of computing witness EF1 allocations for goods using comparison queries.
Here, we only require witnesses for an identical chore partition to support the subsequent reallocation~\mbox{procedure}.

A logarithmic-query subroutine for identical goods was recently established by Bu et al.~\cite{conf/wine/BuLLST24}.
Through a more careful analysis, we adapt their MinBundleKept and moving-knife framework to obtain the witness-based guarantee needed for chores.
We formalize this existential and computational guarantee in the following lemma and defer its full proof to the Appendix.

\begin{lemma} \label{lemma: first_EF1}
\emergencystretch=2em
    An identical witness EF1 partition can be computed using $O(\log m)$~queries.
\end{lemma}
In the following, we first introduce an overview of our algorithm.

\paragraph{Algorithm Overview}
Our algorithm begins by computing an identical witness EF1 partition $(A, B, C)$ for agent $1$ via Lemma~\ref{lemma: first_EF1}.
If agents $2$ and $3$ have different minimum-cost bundles, they can simply choose their preferred bundles, leaving the remaining one for agent $1$ (which is guaranteed to be EF1 for her).
Otherwise, agents $2$ and $3$ must share the same preference for the minimum bundle (without loss of generality, we assume this bundle is $C$).
We then consider transferring items from $A$ and $B$ to $C$.
Intuitively, our goal is to identify a transition point where either agent $2$ or agent $3$ no longer strictly prefers $C$.
To efficiently transfer items while preserving the fairness guarantee for agent $1$ (ensuring that bundles $A$ and $B$ always remain EF1-feasible for her\footnote{A bundle is called EF1-feasible if, under a given partition, receiving this bundle ensures the EF1 property with respect to the rest of the bundles.}), we construct a monotone staircase to characterize the process, where each state represents a specific transfer configuration.
Specifically, we define the notion of a safe state and establish its fairness guarantee (Lemma~\ref{lemma:ef1-safe-state}).
Thus, our objective reduces to finding a specific transition point among the safe states along this staircase.
We further prove that a path consisting of safe states always exists on the monotone staircase (Lemma~\ref{lemma:ef1-safe-staircase}) and that the desired transition point can be located via a carefully designed binary search using $O(\log m)$ comparisons (Lemma~\ref{lemma:ef1-transition-search} and Lemma~\ref{lemma:ef1-query}).
Finally, by carefully analyzing this transition point, we can achieve an EF1 allocation with only $O(1)$ additional comparison queries (Lemma~\ref{lemma:ef1-transition-allocation}).
\medskip

For the remainder of this section, we focus exclusively on the hard case where agents $2$ and $3$ share the same preference for the minimum bundle, and we adopt the following notation.
For each nonempty $X\in\{A,B,C\}$, let $e$ denote the witness chore of $X$ in the identical witness EF1 partition for agent $1$.
Then we have
\begin{equation*}\label{eq:agent1-initial-ef1}
    c_1(X\setminus\{e\})\leq c_1(Y)
    \qquad
    \text{for all nonempty }X\in\{A,B,C\}\text{ and all }Y\in\{A,B,C\}.
\end{equation*}
After relabeling the common unique minimum-cost bundle of agents $2$ and $3$ as $C$, we~have
\begin{equation*}\label{eq:common-minimum-C}
    c_i(C)<c_i(A)
    \quad\text{and}\quad
    c_i(C)<c_i(B)
    \qquad\text{for }i\in\{2,3\}.
\end{equation*}

\subsection{Monotone Staircase Construction}
In this subsection, we demonstrate how to construct a staircase where each state represents a specific transfer of items from $A$ and $B$ to $C$.
Furthermore, we characterize the safe states along this staircase, which enables us to establish the existence of a monotone path of safe states throughout the entire transfer process.

To formalize the staircase construction, order the chores in bundles $A$ and $B$ as
\begin{equation*}
    A=(a_1,\ldots,a_p),\qquad B=(b_1,\ldots,b_q),
\end{equation*}
where $a_1=e_A$ and $b_1=e_B$ are the witness items.
For any $0\leq t\leq p$ and $0\leq s\leq q$, the state $(t,s)$ formally represents the partition
$(G_{t,s},A_t,B_s)$,
where
\begin{equation*}
    A_t=\Suf(A,t),\qquad B_s=\Suf(B,s),\qquad
    G_{t,s}=C\cup\Pre(A,t)\cup\Pre(B,s).
\end{equation*}
We also set $A_{p+1}=\emptyset$ and $B_{q+1}=\emptyset$.
Intuitively, this state corresponds to the configuration where the first $t$ items from $A$ and the first $s$ items from $B$ have been systematically transferred into bundle $C$.
Next, we provide a formal definition of a safe state within the state space.
\begin{definition} [Safe State]
    We call a state $(t, s)$ safe if
    \begin{equation*}
    c_1(A_{t+1}) \leq c_1(B_s) \text{  and  }
    c_1(B_{s+1}) \leq c_1(A_t).
    \end{equation*}
\end{definition}
The following lemma establishes that a safe state provides a sufficient condition for agent $1$ to achieve the EF1~guarantee.
\begin{lemma}\label{lemma:ef1-safe-state}
    At any safe state $(t, s)$, if agent $1$ receives either $A_t$ or $B_s$, then agent $1$ is EF1 with respect to the partition $(G_{t,s},A_t,B_s)$.
\end{lemma}
\begin{proof}
    We assume that agent $1$ receives $A_t$ and the proof is symmetric if she receives $B_s$.
    If $A_t=\emptyset$, agent $1$ is trivially EF1.
    Otherwise, removing the first remaining item $a_{t+1}$ leaves $A_{t+1}$.
    By definition of safe state, we have
    \begin{equation*}
    c_1(A_{t+1})\leq c_1(B_s),
    \end{equation*}
    so agent $1$ does not envy $B_s$ after the removal.
    To compare with the growing bundle, recall that $a_1$ is the witness for the initial bundle $A$, which implies $c_1(A \setminus \{a_1\}) \le c_1(C)$. Since $A_{t+1}\subseteq A\setminus\{a_1\}$ and $C\subseteq G_{t,s}$, monotonicity~gives
    \begin{equation*}
    c_1(A_{t+1}) \le c_1(A \setminus \{a_1\}) \le c_1(C) \le c_1(G_{t,s}).
    \end{equation*}
    Thus, after removing one item from her own bundle, agent $1$ envies neither of the other two bundles, satisfying~EF1.
\end{proof}

To specify the safe staircase used below, we define the boundary function $g$ as follows.
For any $0 \le t < p$, let $g(t)$ be the smallest $s \in \{0,\ldots,q\}$ such that either $s=q$ or $c_1(A_{t+1})>c_1(B_{s+1})$.
We also set $g(-1)=0$ and $g(p)=q$.
Equivalently, for $t<p$ and $s<q$, we have $s<g(t)$ if and only if $c_1(A_{t+1})\leq c_1(B_{s+1})$.

\begin{lemma}\label{lemma:ef1-safe-staircase}
    There is a monotone path of safe states from $(0,0)$ to $(p,q)$.
    Each step moves exactly one chore from one of the residual bundles into the growing bundle~$G$.
\end{lemma}

\begin{proof}
    Consider the staircase induced by the boundary function $g$ defined above.
    Because the residual bundle $A_{t+1}$ only shrinks as $t$ increases, its associated cost under $c_1$ is monotonically decreasing, which guarantees that the resulting sequence $g(0), g(1), \dots, g(p)$ is non-decreasing.
    Within each individual column $t$, the staircase moves vertically by increasing $s$, and it then proceeds horizontally from $(t, g(t))$ to $(t+1, g(t))$.
    Through this sequence of steps, the path starts seamlessly at $(0,0)$ and terminates at $(p,q)$.

    It suffices to show that this construction is always safe.
    The initial state $(0,0)$ is safe by Lemma~\ref{lemma: first_EF1}.
    Now, consider a vertical step moving from $(t,s)$ to $(t,s+1)$.
    Since this vertical transition occurs only when the condition $s < g(t)$ is met, we naturally have $c_1(A_{t+1}) \le c_1(B_{s+1})$, which establishes the first safety inequality for the new state.
    To verify the second safety inequality, we observe that $B_{s+2} \subseteq B_{s+1}$, which directly implies $c_1(B_{s+2}) \le c_1(B_{s+1})$.
    Furthermore, since the preceding state was already assumed to be safe, the inequality $c_1(B_{s+1}) \le c_1(A_t)$ must hold.
    Combining these two relations via transitivity yields $c_1(B_{s+2}) \le c_1(A_t)$, thereby confirming that the newly reached state remains entirely safe.
    Next, let us consider a horizontal step shifting from $(t,g(t))$ to $(t+1,g(t))$, and let $s = g(t)$ for convenience.
    The first safety inequality at this new state follows directly from the containment $A_{t+2} \subseteq A_{t+1}$.
    Because the safety inequality of the previous state guarantees that $c_1(A_{t+1}) \le c_1(B_s)$, it immediately follows that $c_1(A_{t+2}) \le c_1(B_s)$.
    It remains to show that $c_1(B_{s+1}) \leq c_1(A_t)$.
    If $s < q$, the definition of $g(t)$ provides the strict inequality $c_1(B_{s+1}) < c_1(A_{t+1})$.
    Alternatively, if $s = q$, the residual bundle becomes empty such that $B_{s+1}=\emptyset$, making the inequality trivially satisfied.
    Consequently, by mathematical induction on the step transitions, every state along the constructed staircase path is proven to be~safe.
\end{proof}

\subsection{The Computation of the Transition Point}
In this subsection, we focus on the computation of the transition point along the monotone path guaranteed by Lemma~\ref{lemma:ef1-safe-staircase}.
Specifically, for any staircase state $P_{t,s}=(G_{t,s},A_t,B_s)$, we define $H(P_{t,s})=H(t,s)$, where $H(t,s)=0$ if both agents $2$ and $3$ strictly prefer the growing~bundle:
\begin{equation*}
    c_i(G_{t,s})<c_i(A_t)
    \quad\text{and}\quad
    c_i(G_{t,s})<c_i(B_s)
    \qquad\text{for each }i\in\{2,3\},
\end{equation*}
and $H(t,s)=1$ otherwise.

\begin{figure}[H]
    \centering
    \begin{subfigure}{0.49\linewidth}
        \centering
        \begin{tikzpicture}[
            x=0.72cm,
            y=0.72cm,
            >=Stealth,
            every node/.style={font=\footnotesize},
            line cap=round,
            line join=round
        ]
            \draw[step=1,gray!20,very thin] (0,0) grid (5,5);
            \draw[->] (-0.15,0) -- (5.35,0) node[right] {$t$};
            \draw[->] (0,-0.15) -- (0,5.35) node[above] {$s$};
            \node[below] at (3,0) {$t^*$};
            \node[left] at (0,3) {$g(t^*-1)$};

            \draw[blue!75!black,very thick]
                (0,0) -- (0,1) -- (1,1) -- (1,3) -- (2,3) -- (3,3) -- (3,4) -- (5,4);
            \foreach \x/\y in {0/0,0/1,1/1,1/3,2/3,3/3,3/4,5/4} {
                \fill[blue!75!black] (\x,\y) circle (2.2pt);
            }

            \draw[red!70!black,dashed,thick]
                (0,5) -- (1,4) -- (2,4) -- (3,3) -- (4,2) -- (5,2);

            \draw[orange!90!black,line width=1.5pt,->] (2,3) -- (3,3);
            \fill[orange!90!black] (2,3) circle (2.5pt);
            \fill[orange!90!black] (3,3) circle (2.5pt);
        \end{tikzpicture}
        \caption{Horizontal transition: $h(t^*)\le g(t^*-1)$.}
        \label{fig:ef1-transition-horizontal}
    \end{subfigure}
    \hfill
    \begin{subfigure}{0.49\linewidth}
        \centering
        \begin{tikzpicture}[
            x=0.72cm,
            y=0.72cm,
            >=Stealth,
            every node/.style={font=\footnotesize},
            line cap=round,
            line join=round
        ]
            \draw[step=1,gray!20,very thin] (0,0) grid (5,5);
            \draw[->] (-0.15,0) -- (5.35,0) node[right] {$t$};
            \draw[->] (0,-0.15) -- (0,5.35) node[above] {$s$};
            \node[below] at (2,0) {$t^*$};
            \node[left] at (0,1) {$g(t^*-1)$};
            \node[left] at (0,3) {$h(t^*)$};

            \draw[blue!75!black,very thick]
                (0,0) -- (0,1) -- (1,1) -- (2,1) -- (2,2) -- (2,3) -- (5,3);
            \foreach \x/\y in {0/0,0/1,1/1,2/1,2/2,2/3,5/3} {
                \fill[blue!75!black] (\x,\y) circle (2.2pt);
            }

            \draw[red!70!black,dashed,thick]
                (0,5) -- (1,4) -- (2,3) -- (3,2) -- (5,2);

            \draw[orange!90!black,line width=1.5pt,->] (2,2) -- (2,3);
            \fill[orange!90!black] (2,2) circle (2.5pt);
            \fill[orange!90!black] (2,3) circle (2.5pt);
        \end{tikzpicture}
        \caption{Vertical transition: $g(t^*-1)<h(t^*)\le g(t^*)$.}
        \label{fig:ef1-transition-vertical}
    \end{subfigure}
    \medskip

    \begin{tikzpicture}[
        x=1cm,
        y=1cm,
        >=Stealth,
        every node/.style={font=\footnotesize},
        line cap=round
    ]
        \draw[blue!75!black,very thick] (0,0) -- (0.65,0);
        \node[anchor=west] at (0.75,0) {safe staircase $g$};
        \draw[red!70!black,dashed,thick] (3.25,0) -- (3.9,0);
        \node[anchor=west] at (4,0) {boundary $h$};
        \draw[orange!90!black,line width=1.5pt,->] (5.95,0) -- (6.6,0);
        \node[anchor=west] at (6.7,0) {first transition};
    \end{tikzpicture}
    \caption{Two possible ways in which the safe staircase first reaches a state $(t,s)$ with $H(t,s)=1$.}
    \label{fig:ef1-transition-cases}
\end{figure}

\begin{lemma} \label{lemma:ef1-transition-search}
There exist two consecutive safe states, $P = (G, U, V)$ and $P' = (G \cup \{e\}, U \setminus \{e\}, V)$ with $e \in U$, such that $H(P)=0$ and $H(P')=1$\footnote{Here, the transition is assumed to always proceed from $U$, potentially by swapping the roles of $U$ and $V$.}.
\end{lemma}

\begin{proof}
Since $C$ is uniquely the common minimum-cost bundle at the initialization, it holds that $H(0,0) = 0$.
Conversely, since both residual bundles are empty at $(p,q)$, we have $H(p,q) = 1$.
Moreover, $H$ is monotonically nondecreasing in both coordinates, as increasing either coordinate only adds chores to the growing bundle and removes them from a residual bundle.

For each $t$, define $h(t)=\min\{s:H(t,s)=1\}$.
Intuitively, given the first $t$ chores from $A$, the value $h(t)$ denotes the number of chores that must be transferred from $B$ to $C$ to reach the transition point.
This is well-defined since
$B_q=\emptyset$, and $h(t)$ is nonincreasing in $t$.  Let
\begin{equation*}
    t^*=\min\{t:g(t)\geq h(t)\}.
\end{equation*}
The value $t^*$ is the first column in which the safe staircase reaches a state
$(t,s)$ with $H(t,s)=1$; by the convention $g(-1)=0$, the height $g(t^*-1)$ is the entry height of this column.

If $h(t^*)\leq g(t^*-1)$, then the first state on the staircase satisfying $H(t,s)=1$ is reached by
the horizontal edge, as illustrated in Figure~\ref{fig:ef1-transition-cases}(a),
\begin{equation*}
    (t^*-1,g(t^*-1))\to(t^*,g(t^*-1)).
\end{equation*}
Otherwise, it is reached by the vertical edge, as illustrated in Figure~\ref{fig:ef1-transition-cases}(b),
\begin{equation*}
    (t^*,h(t^*)-1)\to(t^*,h(t^*)).
\end{equation*}
In both cases, if the tail and head are denoted by $P$ and $P'$, respectively, then $H(P)=0$ and $H(P')=1$, and the two states differ by moving exactly one chore from a residual bundle to the growing bundle.
This proves the existence of the desired transition.
\end{proof}

Therefore, the problem reduces to efficiently locating the consecutive safe states.
We address this by presenting a two-dimensional binary search framework that achieves the goal within $O(\log m)$ queries.
To illustrate our algorithm, we further define the staircase predicate $\Phi$ by $\Phi(t,s)=1$ if and only if $s<g(t)$.
The search procedure is described in Algorithm~\ref{alg:ef1-transition-search}.

\begin{algorithm}[!htb]
\caption{Search for the First Safe Staircase Transition}
\label{alg:ef1-transition-search}
\KwIn{Bundles $A=(a_1,\ldots,a_p)$, $B=(b_1,\ldots,b_q)$, $C$, and predicates $\Phi,H$}

$t_\ell\gets 0$, $t_r\gets p$, $s_\ell\gets 0$, $s_r\gets q$\;
\While{$t_\ell<t_r$ \textbf{and} $s_\ell<s_r$}{
    $t_m\gets\lfloor(t_\ell+t_r)/2\rfloor$, $s_m\gets\lfloor(s_\ell+s_r)/2\rfloor$\;
    $(\Phi_m,H_m)\gets(\Phi(t_m,s_m),H(t_m,s_m))$\;
    \lIf{$(\Phi_m,H_m)=(1,1)$}{$t_r\gets t_m$}
    \lElseIf{$(\Phi_m,H_m)=(0,0)$}{$t_\ell\gets t_m+1$}
    \lElseIf{$(\Phi_m,H_m)=(1,0)$}{$s_\ell\gets s_m+1$}
    \lElse{$s_r\gets s_m$}
}
\If{$t_\ell=t_r$}{
    $t^*\gets t_\ell$\;
    Find $h(t^*)$ via $H$ and $g(t^*-1)$ via $\Phi$ by binary~search\;
    \Return the corresponding horizontal or vertical transition\;
}
\Else{
    $s^*\gets s_\ell$\;
    Find $\tau_g=\min\{t:g(t)\geq s^*\}$ via $\Phi$ and $\tau_h=\min\{t:h(t)\leq s^*\}$ by binary~search\;
    $t^*\gets\max\{\tau_g,\tau_h\}$\;
    \Return the corresponding horizontal or vertical transition\;
}
\KwOut{Consecutive safe states $P,P'$ with $H(P)=0$ and $H(P')=1$}
\end{algorithm}

\begin{lemma} \label{lemma:ef1-query}
\emergencystretch=2em
The transition point described in Lemma~\ref{lemma:ef1-transition-search} can be found via $O(\log m)$~comparison queries by~Algorithm~\ref{alg:ef1-transition-search}.
\end{lemma}
\begin{proof}
It suffices to justify the correctness and query complexity of Algorithm~\ref{alg:ef1-transition-search}.
For $t<p$ and $s<q$, the condition $\Phi(t,s)=1$ is equivalent to the comparison
$c_1(A_{t+1})\leq c_1(B_{s+1})$,
and thus costs one comparison to agent $1$.  On the boundary, $\Phi(p,s)=1$ for
$s<q$ and $\Phi(t,q)=0$ by the conventions $g(p)=q$ and $s<g(t)$, so no query is
needed there.  A value of $H(t,s)$ can be evaluated with the four comparisons
$\Comp_i(A_t,G_{t,s})$ and $\Comp_i(B_s,G_{t,s})$ for $i\in\{2,3\}$.
For each agent, both comparisons return $G_{t,s}$ exactly when she strictly prefers the growing bundle.

Let $t^*=\min\{t:g(t)\geq h(t)\}$ be the first column in which the safe staircase reaches a state $(t,s)$ with $H(t,s)=1$, and let $s^*=\max\{h(t^*),g(t^*-1)\}$.
We maintain the intervals $[t_\ell,t_r]\subseteq[0,p]$ and
$[s_\ell,s_r]\subseteq[0,q]$ such that $t^*\in[t_\ell,t_r]$ and
$s^*\in[s_\ell,s_r]$.
Initially, the claim is trivial.
While both intervals have positive lengths, let
$t_m=\lfloor(t_\ell+t_r)/2\rfloor$ and
$s_m=\lfloor(s_\ell+s_r)/2\rfloor$, and query $\Phi(t_m,s_m)$ and $H(t_m,s_m)$.
The interval updates in Algorithm~\ref{alg:ef1-transition-search}, illustrated in
Figure~\ref{fig:ef1-binary-search-cases}, preserve the rectangle invariant as follows.
\begin{itemize}
    \item If $(\Phi,H)=(1,1)$, then
$s_m<g(t_m)$ and $s_m\geq h(t_m)$, so $g(t_m)>h(t_m)$ and hence
$t^*\leq t_m$.
    \item   If $(\Phi,H)=(0,0)$, then $g(t_m)\leq s_m<h(t_m)$, so $g(t_m)<h(t_m)$ and $t^*>t_m$.
    \item  If $(\Phi,H)=(1,0)$, then $s_m<g(t_m)$ and $s_m<h(t_m)$.
    When $t^*\leq t_m$, monotonicity gives $h(t^*)\geq h(t_m)>s_m$; when $t^*>t_m$, we have $g(t^*-1)\geq g(t_m)>s_m$.  Hence, $s^*>s_m$.
    \item Finally, if $(\Phi,H)=(0,1)$, then $s_m\geq g(t_m)$ and $s_m\geq h(t_m)$.
    If $t^*\leq t_m$, then $s^*\leq g(t^*)\leq g(t_m)\leq s_m$; if $t^*>t_m$, then by the minimality of $t^*$ and monotonicity of $h$, both $g(t^*-1)\leq s_m$ and $h(t^*)\leq h(t_m)\leq s_m$.  Thus, $s^*\leq s_m$.

\end{itemize}

Furthermore, each step halves one active interval, so there are $O(\log(p+1)+\log(q+1))=
O(\log m)$ iterations.
When one interval collapses, the exact transition is found by a final one-dimensional
binary search.  If $t_\ell=t_r=t^*$, first find $h(t^*)$ by binary search on the
monotone predicate $H(t^*,s)$.  If $t^*>0$, find $g(t^*-1)$ as the first $s$ for
which $\Phi(t^*-1,s)=0$, using the convention $\Phi(t,q)=0$; if $t^*=0$, use
$g(-1)=0$.  The two values determine whether the transition is horizontal or
vertical as in Lemma~\ref{lemma:ef1-transition-search}.

If instead $s_\ell=s_r=s^*$, find
\begin{equation*}
    \tau_g=\min\{t:g(t)\geq s^*\},\qquad
    \tau_h=\min\{t:h(t)\leq s^*\}
\end{equation*}
by binary search, using $\Phi(t,s^*-1)$ for the first predicate when $s^*>0$ and
setting $\tau_g=0$ when $s^*=0$.  Then $t^*=\max\{\tau_g,\tau_h\}$; if
$\tau_g=t^*$ the transition is vertical
$(t^*,s^*-1)\to(t^*,s^*)$, and otherwise it is horizontal
$(t^*-1,s^*)\to(t^*,s^*)$.
Every predicate evaluation uses only $O(1)$ comparison queries, so the total
number of queries is $O(\log m)$.
\end{proof}

\begin{figure}[!t]
    \centering
    \begin{subfigure}{0.47\linewidth}
        \centering
        \begin{tikzpicture}[
            x=0.58cm,
            y=0.58cm,
            >=Stealth,
            every node/.style={font=\scriptsize},
            line cap=round,
            line join=round
        ]
            \draw[step=1,gray!20,very thin] (0,0) grid (4,4);
            \draw[->] (-0.12,0) -- (4.35,0) node[right] {$t$};
            \draw[->] (0,-0.12) -- (0,4.35) node[above] {$s$};
            \draw[blue!75!black,very thick] (0,0) -- (0,1) -- (1,1) -- (1,3) -- (4,3);
            \draw[red!70!black,dashed,thick] (0,4) -- (1,3) -- (2,1) -- (4,1);
            \fill[black] (2,2) circle (2.2pt);
            \draw[orange!90!black,line width=1.2pt,->] (2.5,3.65) -- (1.25,3.65);
        \end{tikzpicture}
        \caption{$(\Phi,H)=(1,1):\,t_r\gets t_m$.}
    \end{subfigure}
    \hfill
    \begin{subfigure}{0.47\linewidth}
        \centering
        \begin{tikzpicture}[
            x=0.58cm,
            y=0.58cm,
            >=Stealth,
            every node/.style={font=\scriptsize},
            line cap=round,
            line join=round
        ]
            \draw[step=1,gray!20,very thin] (0,0) grid (4,4);
            \draw[->] (-0.12,0) -- (4.35,0) node[right] {$t$};
            \draw[->] (0,-0.12) -- (0,4.35) node[above] {$s$};
            \draw[blue!75!black,very thick] (0,0) -- (0,1) -- (4,1);
            \draw[red!70!black,dashed,thick] (0,4) -- (1,3) -- (2,3) -- (4,1);
            \fill[black] (2,2) circle (2.2pt);
            \draw[orange!90!black,line width=1.2pt,->] (1.5,3.65) -- (2.75,3.65);
        \end{tikzpicture}
        \caption{$(\Phi,H)=(0,0):\,t_\ell\gets t_m+1$.}
    \end{subfigure}
    \par\medskip
    \begin{subfigure}{0.47\linewidth}
        \centering
        \begin{tikzpicture}[
            x=0.58cm,
            y=0.58cm,
            >=Stealth,
            every node/.style={font=\scriptsize},
            line cap=round,
            line join=round
        ]
            \draw[step=1,gray!20,very thin] (0,0) grid (4,4);
            \draw[->] (-0.12,0) -- (4.35,0) node[right] {$t$};
            \draw[->] (0,-0.12) -- (0,4.35) node[above] {$s$};
            \draw[blue!75!black,very thick] (0,0) -- (0,2) -- (1,2) -- (1,3) -- (4,3);
            \draw[red!70!black,dashed,thick] (0,4) -- (1,3) -- (2,2) -- (4,2);
            \fill[black] (2,1) circle (2.2pt);
            \draw[orange!90!black,line width=1.2pt,->] (3.55,0.65) -- (3.55,1.9);
        \end{tikzpicture}
        \caption{$(\Phi,H)=(1,0):\,s_\ell\gets s_m+1$.}
    \end{subfigure}
    \hfill
    \begin{subfigure}{0.47\linewidth}
        \centering
        \begin{tikzpicture}[
            x=0.58cm,
            y=0.58cm,
            >=Stealth,
            every node/.style={font=\scriptsize},
            line cap=round,
            line join=round
        ]
            \draw[step=1,gray!20,very thin] (0,0) grid (4,4);
            \draw[->] (-0.12,0) -- (4.35,0) node[right] {$t$};
            \draw[->] (0,-0.12) -- (0,4.35) node[above] {$s$};
            \draw[blue!75!black,very thick] (0,0) -- (0,1) -- (4,1);
            \draw[red!70!black,dashed,thick] (0,3) -- (1,2) -- (2,2) -- (3,1) -- (4,0);
            \fill[black] (2,3) circle (2.2pt);
            \draw[orange!90!black,line width=1.2pt,->] (3.55,3.35) -- (3.55,2.1);
        \end{tikzpicture}
        \caption{$(\Phi,H)=(0,1):\,s_r\gets s_m$.}
    \end{subfigure}
    \medskip

    \begin{tikzpicture}[
        x=1cm,
        y=1cm,
        >=Stealth,
        every node/.style={font=\footnotesize},
        line cap=round
    ]
        \path[use as bounding box] (-0.6,-0.18) rectangle (11.4,0.18);
        \draw[blue!75!black,very thick] (-0.2,0) -- (0.45,0);
        \node[anchor=west] at (0.6,0) {safe staircase $g$};
        \draw[red!70!black,dashed,thick] (3.35,0) -- (4,0);
        \node[anchor=west] at (4.15,0) {boundary $h$};
        \fill[black] (6.55,0) circle (2.2pt);
        \node[anchor=west] at (6.75,0) {midpoint};
        \draw[orange!90!black,line width=1.5pt,->] (9,0) -- (9.75,0);
        \node[anchor=west] at (9.9,0) {update};
    \end{tikzpicture}
    \caption{The four possible midpoint configurations in Algorithm~\ref{alg:ef1-transition-search}.}
    \label{fig:ef1-binary-search-cases}
\end{figure}

\Needspace{7\baselineskip}
\subsection{The Construction of EF1 Allocations}
In this subsection, we establish the final EF1 allocation based on the transition point computed in~Lemma~\ref{lemma:ef1-transition-search}.
\begin{lemma} \label{lemma:ef1-transition-allocation}
Given the transition returned by Lemma~\ref{lemma:ef1-query}, one can compute an EF1 allocation using $O(1)$ additional comparisons.
\end{lemma}

\begin{proof}
Let $P=(G,U,V)$ and $P'=(G\cup\{e\},U\setminus\{e\},V)$ be the partition transition from Lemma~\ref{lemma:ef1-query}.
Under partition $P$, both agents 2 and 3 strictly prefer $G$ over $U$ and $V$. Furthermore, both $P$ and $P'$ are safe for agent 1.
We consider two cases based on how agents 2 and 3 update their preferred bundles under $P'$.

\begin{itemize}
    \item \textbf{Case 1: One of agents 2 and 3 switches under $P'$. }
    Without loss of generality, assume agent 2 switches while agent 3 does not.
    We allocate $G\cup\{e\}$ to agent 3, who remains envy-free as she keeps her strict preference for this bundle.
    We then assign agent 2 her preferred bundle from $\{U\setminus\{e\},V\}$.
    Since $G\cup\{e\}$ is no longer her favorite bundle under $P'$, her choice costs at most as much as $G\cup\{e\}$, which ensures her envy-freeness.
    Finally, we give the remaining bundle to agent 1, who is EF1 by the safety of $P'$.
    \item \textbf{Case 2: Both agents 2 and 3 switch at $P'$. }
    In this case, we query $\Comp_2(G,U\setminus\{e\})$ and branch as~follows:
    \begin{itemize}
        \item [(i)] $c_2(G) \leq c_2(U\setminus\{e\})$. We use the partition $P'$. Assign $G\cup\{e\}$ to agent 2, let agent 3 choose the cheaper bundle from $\{U\setminus\{e\},V\}$, and give the rest to agent 1. Agent 3 is envy-free because she switches. Agent 2 is EF1 because removing $e$ leaves her with a cost of $c_2(G)$, which is at most $c_2(U\setminus\{e\})$ (by the query) and strictly less than $c_2(V)$ (as she preferred $G$ at $P$). Agent 1 is EF1 by the safety of $P'$.

        \item [(ii)] $c_2(G) > c_2(U\setminus\{e\})$. We use the partition $P$. Assign $U$ to agent 2, $G$ to agent 3, and $V$ to agent 1. Agent 3 receives her strictly preferred bundle $G$ and is envy-free. Agent 2 is EF1 because removing $e$ from $U$ leaves her with $U\setminus\{e\}$, and her costs satisfy $c_2(U\setminus\{e\}) < c_2(G) < c_2(V)$. Agent 1 is EF1 by the safety of $P$.
        \qedhere
    \end{itemize}
\end{itemize}
\end{proof}

\begin{theorem} \label{theorem: EF1}
    For instances with additive cost functions, EF1 allocations for three agents can be computed using $O(\log m)$~queries.
\end{theorem}
\begin{proof}
We prove the correctness and query complexity of the procedure described above.
By Lemma~\ref{lemma: first_EF1}, using $O(\log m)$ comparisons to agent $1$, the
algorithm first computes an EF1 partition $(A,B,C)$.
Since there are only three bundles, the sets of minimum-cost bundles for agents $2$ and $3$ can be computed using only $O(1)$ comparisons.

If agents $2$ and $3$ can be assigned distinct minimum-cost bundles, the algorithm gives those bundles to them and gives the remaining bundle to agent $1$.
Agents $2$ and $3$ are envy-free, because each receives a minimum-cost bundle among $\{A,B,C\}$ while agent $1$ is EF1.
Hence, the returned allocation is EF1 in this case.

It remains to consider the hard case in which no distinct minimum-cost bundles can be assigned to agents $2$ and $3$.  Then their two sets of
minimum-cost bundles must be the same singleton.  Relabel this common unique
minimum-cost bundle as $C$ and relabel the other two bundles as $A$ and $B$.
By Lemma~\ref{lemma:ef1-safe-staircase}, there is a safe staircase from $(0,0)$ to $(p,q)$. By Lemma~\ref{lemma:ef1-transition-search} and Lemma~\ref{lemma:ef1-query}, the first preference transition along
this staircase can be found using $O(\log m)$ comparison queries.  Finally,
Lemma~\ref{lemma:ef1-transition-allocation} converts the two neighboring states at
this transition into an EF1 allocation using only $O(1)$ additional comparisons.
Therefore, the total number of comparison queries is $O(\log m)$, and the allocation returned by the procedure described above is EF1.
This completes the proof.
\end{proof}

\section{Conclusion and Open Problems}\label{sec:conclusion}
In this paper, we studied fair allocation of indivisible chores under the comparison-query model, where algorithms can only compare the costs of two bundles and do not have access to cardinal cost values.
We developed logarithmic-query algorithms for several fairness guarantees.
For PROP1, we gave a comparison-based algorithm with $O(n^3\log m)$ queries and also proved an $O(n^5\log m)$ bound for the contiguous setting, achieving logarithmic dependence on $m$ for fixed $n$ in both settings.
For the MMS guarantee, we introduced a compression framework that reduces the original instance to a bounded-size meta-instance while preserving the MMS benchmark up to a controlled loss, yielding a $(13/11 +\varepsilon)$-MMS allocation and only logarithmic dependence on $m$.
We also investigated the distortion of the comparison model for MMS, obtaining a $13/11$ upper bound and a $31/30$ lower bound via indistinguishable cardinal instances.
Finally, for EF1, we showed that when there are three agents, EF1 allocations can be computed using $O(\log m)$ comparisons by combining a contiguous identical-cost EF1 partition with a staircase search.

Our conference version left open whether EF1 allocations beyond three agents can be computed using a number of comparison queries that is polynomial in $n$ and logarithmic in $m$.
This question has since been resolved for additive preferences by Lin et al.~\cite{journals/corr/abs-2611-04753}, who obtain an $O(n^8\log m)$ query bound for both goods and chores.
For general monotone cost functions, however, it remains open whether EF1 allocations can be computed using $O(\log m)$ comparison queries for every fixed number of agents.

Another open problem concerns the distortion of the comparison model for MMS.
Our bounds leave a gap between the $31/30$ lower bound and the $13/11$ upper bound.
% The second direction is a sharper understanding of MMS distortion.
% The current bounds leave a gap between the $31/30$ lower bound and the $13/11$
% upper bound.
Closing this gap would clarify how much cardinal information is lost when only bundle comparisons are available.

%\newpage

\section*{Declaration of AI Usage}
The mathematical results in the preliminary conference version were developed independently by the authors. In this journal version, the improved algorithm for contiguous PROP1 allocations and its analysis in Section~\ref{sec:contiguous_prop1} were developed through extensive interaction with GPT-6 Astra. AI tools were also used to assist with manuscript revision and figure preparation. All mathematical statements, proofs, calculations, and citations were reviewed and verified by the authors, who take full responsibility for the content of the paper.

\bibliographystyle{abbrv}
\bibliography{ref}

@article{Suksompong21,
	author = {Suksompong, Warut},
	title = {Constraints in Fair Division},
	journal = {ACM SIGecom Exchanges},
	volume = {19},
	number = {2},
	year = {2021},
	pages = {46--61},
}

@article{liu2024mixed,
  title={Mixed fair division: A survey},
  author={Liu, Shengxin and Lu, Xinhang and Suzuki, Mashbat and Walsh, Toby},
  journal={Journal of Artificial Intelligence Research},
  volume={80},
  pages={1373--1406},
  year={2024}
}

@inproceedings{conf/wine/BuLLST24,
  author       = {Xiaolin Bu and
                  Zihao Li and
                  Shengxin Liu and
                  Jiaxin Song and
                  Biaoshuai Tao},
  title        = {Logarithmic Comparison-Based Query Complexity for Fair Division of
                  Indivisible Goods},
  booktitle    = {{WINE}},
  series       = {Lecture Notes in Computer Science},
  pages        = {348--365},
  publisher    = {Springer},
  year         = {2024}
}

@article{journals/ai/AmanatidisABFLMVW23,
  author       = {Georgios Amanatidis and
                  Haris Aziz and
                  Georgios Birmpas and
                  Aris Filos{-}Ratsikas and
                  Bo Li and
                  Herv{\'{e}} Moulin and
                  Alexandros A. Voudouris and
                  Xiaowei Wu},
  title        = {Fair division of indivisible goods: Recent progress and open questions},
  journal      = {Artif. Intell.},
  volume       = {322},
  pages        = {103965},
  year         = {2023}
}

@inproceedings{conf/sigecom/LiptonMMS04,
  author    = {Richard J. Lipton and
               Evangelos Markakis and
               Elchanan Mossel and
               Amin Saberi},
  title     = {On approximately fair allocations of indivisible goods},
  booktitle = {{EC}},
  pages     = {125--131},
  publisher = {{ACM}},
  year      = {2004}
}

@article{journals/teco/CaragiannisKMPS19,
  author    = {Ioannis Caragiannis and
               David Kurokawa and
               Herv{\'{e}} Moulin and
               Ariel D. Procaccia and
               Nisarg Shah and
               Junxing Wang},
  title     = {The Unreasonable Fairness of Maximum {Nash} Welfare},
  journal   = {{ACM} Trans. Economics and Comput.},
  volume    = {7},
  number    = {3},
  pages     = {12:1--12:32},
  year      = {2019}
}

@article{journals/siamdm/PlautR20,
  author       = {Benjamin Plaut and
                  Tim Roughgarden},
  title        = {Almost Envy-Freeness with General Valuations},
  journal      = {{SIAM} J. Discret. Math.},
  volume       = {34},
  number       = {2},
  pages        = {1039--1068},
  year         = {2020}
}

@article{conf/bqgt/Budish10,
  title={The combinatorial assignment problem: Approximate competitive equilibrium from equal incomes},
  author={Budish, Eric},
  journal={Journal of Political Economy},
  volume={119},
  number={6},
  pages={1061--1103},
  year={2011},
  publisher={University of Chicago Press Chicago, IL}
}

@inproceedings{conf/wine/FeigeST21,
  author       = {Uriel Feige and
                  Ariel Sapir and
                  Laliv Tauber},
  title        = {A Tight Negative Example for {MMS} Fair Allocations},
  booktitle    = {{WINE}},
  series       = {Lecture Notes in Computer Science},
  volume       = {13112},
  pages        = {355--372},
  publisher    = {Springer},
  year         = {2021}
}

@inproceedings{conf/aaai/AzizRSW17,
  author       = {Haris Aziz and
                  Gerhard Rauchecker and
                  Guido Schryen and
                  Toby Walsh},
  title        = {Algorithms for Max-Min Share Fair Allocation of Indivisible Chores},
  booktitle    = {{AAAI}},
  pages        = {335--341},
  publisher    = {{AAAI} Press},
  year         = {2017}
}

@article{journals/mp/AzizLW24,
  author       = {Haris Aziz and
                  Bo Li and
                  Xiaowei Wu},
  title        = {Approximate and strategyproof maximin share allocation of chores with
                  ordinal preferences},
  journal      = {Math. Program.},
  volume       = {203},
  number       = {1},
  pages        = {319--345},
  year         = {2024}
}

@inproceedings{conf/sigecom/FeigeH23,
  author       = {Uriel Feige and
                  Xin Huang},
  title        = {On picking sequences for chores},
  booktitle    = {{EC}},
  pages        = {626--655},
  publisher    = {{ACM}},
  year         = {2023}
}

@article{journals/teco/BarmanK20,
  author       = {Siddharth Barman and
                  Sanath Kumar Krishnamurthy},
  title        = {Approximation Algorithms for Maximin Fair Division},
  journal      = {{ACM} Trans. Economics and Comput.},
  volume       = {8},
  number       = {1},
  pages        = {5:1--5:28},
  year         = {2020}
}

@inproceedings{conf/sigecom/HuangL21,
  author       = {Xin Huang and
                  Pinyan Lu},
  title        = {An Algorithmic Framework for Approximating Maximin Share Allocation
                  of Chores},
  booktitle    = {{EC}},
  pages        = {630--631},
  publisher    = {{ACM}},
  year         = {2021}
}

@inproceedings{conf/sigecom/HuangS23,
  author       = {Xin Huang and
                  Erel Segal{-}Halevi},
  title        = {A Reduction from Chores Allocation to Job Scheduling},
  booktitle    = {{EC}},
  pages        = {908},
  publisher    = {{ACM}},
  year         = {2023}
}

@article{journals/siamdm/OhPS21,
  author       = {Hoon Oh and
                  Ariel D. Procaccia and
                  Warut Suksompong},
  title        = {Fairly Allocating Many Goods with Few Queries},
  journal      = {{SIAM} J. Discret. Math.},
  volume       = {35},
  number       = {2},
  pages        = {788--813},
  year         = {2021}
}

@article{steinhaus1948problem,
  title={The problem of fair division},
  author={Steinhaus, Hugo},
  journal={Econometrica},
  volume={16},
  pages={101--104},
  year={1948}
}

@book{foley1966resource,
  title={Resource allocation and the public sector},
  author={Foley, Duncan Karl},
  year={1966},
  publisher={Yale University}
}

@inproceedings{conitzer2017fair,
  title={Fair public decision making},
  author={Conitzer, Vincent and Freeman, Rupert and Shah, Nisarg},
  booktitle={Proceedings of the 2017 ACM Conference on Economics and Computation},
  pages={629--646},
  year={2017}
}

@article{journals/siamcomp/PlautR20,
  author       = {Benjamin Plaut and
                  Tim Roughgarden},
  title        = {Communication Complexity of Discrete Fair Division},
  journal      = {{SIAM} J. Comput.},
  volume       = {49},
  number       = {1},
  pages        = {206--243},
  year         = {2020}
}

@inproceedings{conf/sigecom/Feige25,
  author       = {Uriel Feige},
  title        = {Low communication protocols for fair allocation of indivisible goods},
  booktitle    = {{EC}},
  pages        = {358--382},
  publisher    = {{ACM}},
  year         = {2025}
}

@inproceedings{conf/aaai/Igarashi23,
  author       = {Ayumi Igarashi},
  title        = {How to Cut a Discrete Cake Fairly},
  booktitle    = {{AAAI}},
  pages        = {5681--5688},
  publisher    = {{AAAI} Press},
  year         = {2023}
}

@article{journals/geb/BiloCFIMPVZ22,
  author       = {Vittorio Bil{\`{o}} and
                  Ioannis Caragiannis and
                  Michele Flammini and
                  Ayumi Igarashi and
                  Gianpiero Monaco and
                  Dominik Peters and
                  Cosimo Vinci and
                  William S. Zwicker},
  title        = {Almost envy-free allocations with connected bundles},
  journal      = {Games Econ. Behav.},
  volume       = {131},
  pages        = {197--221},
  year         = {2022}
}

@inproceedings{conf/aaai/GolzIMS26,
  author       = {Paul G{\"{o}}lz and
                  Ayumi Igarashi and
                  Pasin Manurangsi and
                  Warut Suksompong},
  title        = {Fair Allocation of Indivisible Goods with Variable Groups},
  booktitle    = {{AAAI}},
  pages        = {16954--16962},
  publisher    = {{AAAI} Press},
  year         = {2026}
}

@inproceedings{conf/ec/BranzeiN19,
  author       = {Simina Br{\^{a}}nzei and
                  Noam Nisan},
  title        = {Communication Complexity of Cake Cutting},
  booktitle    = {{EC}},
  pages        = {525},
  publisher    = {{ACM}},
  year         = {2019}
}

@inproceedings{conf/ijcai/0002021,
  author       = {Daniel Halpern and
                  Nisarg Shah},
  title        = {Fair and Efficient Resource Allocation with Partial Information},
  booktitle    = {{IJCAI}},
  pages        = {224--230},
  publisher    = {ijcai.org},
  year         = {2021}
}

@inproceedings{conf/atal/0001LXZ23,
  author       = {Haris Aziz and
                  Bo Li and
                  Shiji Xing and
                  Yu Zhou},
  title        = {Possible Fairness for Allocating Indivisible Resources},
  booktitle    = {{AAMAS}},
  pages        = {197--205},
  publisher    = {{ACM}},
  year         = {2023}
}

@article{journals/ior/BenadeHP25,
  author       = {Gerdus Benad{\`{e}} and
                  Daniel Halpern and
                  Alex Psomas},
  title        = {Dynamic Fair Division with Partial Information},
  journal      = {Oper. Res.},
  volume       = {73},
  number       = {4},
  pages        = {1876--1896},
  year         = {2025}
}

@article{journals/iandc/LiMSS25,
  author       = {Zihan Li and
                  Pasin Manurangsi and
                  Jonathan Scarlett and
                  Warut Suksompong},
  title        = {Complexity of round-robin allocation with potentially noisy queries},
  journal      = {Inf. Comput.},
  volume       = {306},
  pages        = {105332},
  year         = {2025}
}

@inproceedings{conf/esa/BarmanBKS20,
  author       = {Siddharth Barman and
                  Umang Bhaskar and
                  Anand Krishna and
                  Ranjani G. Sundaram},
  title        = {Tight Approximation Algorithms for p-Mean Welfare Under Subadditive
                  Valuations},
  booktitle    = {{ESA}},
  series       = {LIPIcs},
  volume       = {173},
  pages        = {11:1--11:17},
  publisher    = {Schloss Dagstuhl - Leibniz-Zentrum f{\"{u}}r Informatik},
  year         = {2020}
}

@inproceedings{conf/aaai/Arunachaleswaran19,
  author       = {Eshwar Ram Arunachaleswaran and
                  Siddharth Barman and
                  Nidhi Rathi},
  title        = {Fair Division with a Secretive Agent},
  booktitle    = {{AAAI}},
  pages        = {1732--1739},
  publisher    = {{AAAI} Press},
  year         = {2019}
}

@inproceedings{conf/wine/BarmanV21,
  author       = {Siddharth Barman and
                  Paritosh Verma},
  title        = {Approximating {Nash} Social Welfare Under Binary {XOS} and Binary Subadditive
                  Valuations},
  booktitle    = {{WINE}},
  series       = {Lecture Notes in Computer Science},
  volume       = {13112},
  pages        = {373--390},
  publisher    = {Springer},
  year         = {2021}
}

@inproceedings{conf/atal/BarmanNV23,
  author       = {Siddharth Barman and
                  Vishnu V. Narayan and
                  Paritosh Verma},
  title        = {Fair Chore Division under Binary Supermodular Costs},
  booktitle    = {{AAMAS}},
  pages        = {2863--2865},
  publisher    = {{ACM}},
  year         = {2023}
}

@inproceedings{yao1979some,
  author       = {Andrew Chi{-}Chih Yao},
  title        = {Some Complexity Questions Related to Distributive Computing (Preliminary
                  Report)},
  booktitle    = {{STOC}},
  pages        = {209--213},
  publisher    = {{ACM}},
  year         = {1979}
}

@article{journals/tit/BravermanR14,
  author       = {Mark Braverman and
                  Anup Rao},
  title        = {Information Equals Amortized Communication},
  journal      = {{IEEE} Trans. Inf. Theory},
  volume       = {60},
  number       = {10},
  pages        = {6058--6069},
  year         = {2014}
}

@article{journals/siamcomp/GoosPW20,
  author       = {Mika G{\"{o}}{\"{o}}s and
                  Toniann Pitassi and
                  Thomas Watson},
  title        = {Query-to-Communication Lifting for {BPP}},
  journal      = {{SIAM} J. Comput.},
  volume       = {49},
  number       = {4},
  year         = {2020}
}

@article{journals/jacm/ChattopadhyayMS20,
  author       = {Arkadev Chattopadhyay and
                  Nikhil S. Mande and
                  Suhail Sherif},
  title        = {The Log-Approximate-Rank Conjecture Is False},
  journal      = {J. {ACM}},
  volume       = {67},
  number       = {4},
  pages        = {23:1--23:28},
  year         = {2020}
}

@inproceedings{journals/corr/abs-2602-06361,
  author       = {Zihan Li and
                  Yan Hao Ling and
                  Jonathan Scarlett and
                  Warut Suksompong},
  title        = {Envy-Free Allocation of Indivisible Goods via Noisy Queries},
  booktitle      = {{ICML}},
  year         = {2026}
}

@article{journals/amm/DubinsS61,
  author       = {Lester E. Dubins and
                  Edwin H. Spanier},
  title        = {How to Cut a Cake Fairly},
  journal      = {Amer. Math. Monthly},
  volume       = {68},
  number       = {1},
  pages        = {1--17},
  year         = {1961}
}

@book{books/daglib/RobertsonW98,
  author       = {Jack Robertson and
                  William A. Webb},
  title        = {Cake-Cutting Algorithms: Be Fair if You Can},
  publisher    = {A K Peters},
  year         = {1998}
}

@article{journals/talg/EdmondsP11,
  author       = {Jeff Edmonds and
                  Kirk Pruhs},
  title        = {Cake Cutting Really is Not a Piece of Cake},
  journal      = {{ACM} Trans. Algorithms},
  volume       = {7},
  number       = {4},
  pages        = {51:1--51:12},
  year         = {2011}
}

@inproceedings{conf/nips/BranzeiN22,
  author       = {Simina Br{\^{a}}nzei and
                  Noam Nisan},
  title        = {The Query Complexity of Cake Cutting},
  booktitle    = {Advances in Neural Information Processing Systems},
  volume       = {35},
  year         = {2022}
}

@article{journals/dam/EvenP84,
  author       = {Shimon Even and
                  Azaria Paz},
  title        = {A Note on Cake Cutting},
  journal      = {Discrete Applied Mathematics},
  volume       = {7},
  number       = {3},
  pages        = {285--296},
  year         = {1984}
}

@inproceedings{conf/ijcai/Procaccia09,
  author       = {Ariel D. Procaccia},
  title        = {Thou Shalt Covet Thy Neighbor's Cake},
  booktitle    = {{IJCAI}},
  pages        = {239--244},
  year         = {2009}
}

@article{journals/combinatorics/Stromquist08,
  author       = {Walter Stromquist},
  title        = {Envy-Free Cake Divisions Cannot be Found by Finite Protocols},
  journal      = {Electron. J. Comb.},
  volume       = {15},
  number       = {1},
  pages        = {R11},
  year         = {2008}
}

@inproceedings{conf/stoc/AzizM16,
  author       = {Haris Aziz and
                  Simon Mackenzie},
  title        = {A Discrete and Bounded Envy-Free Cake Cutting Protocol for Any Number of Agents},
  booktitle    = {{STOC}},
  pages        = {416--427},
  publisher    = {{ACM}},
  year         = {2016}
}

@article{journals/mathmag/PetersonS02,
  author       = {Elisha Peterson and
                  Francis Edward Su},
  title        = {Four-Person Envy-Free Chore Division},
  journal      = {Math. Mag.},
  volume       = {75},
  number       = {2},
  pages        = {117--122},
  year         = {2002}
}

@article{journals/tcs/HeydrichS15,
  author       = {Sandy Heydrich and
                  Rob van Stee},
  title        = {Dividing Connected Chores Fairly},
  journal      = {Theor. Comput. Sci.},
  volume       = {593},
  pages        = {51--61},
  year         = {2015}
}

@article{journals/corr/YamagataS26,
  author       = {Tatsuhito Yamagata and Hanna Sumita},
  title        = {Witness-Certified Fair Division with Comparison Queries},
  journal      = {arXiv preprint arXiv:2608.16109},
  year         = {2026},
  doi          = {10.48550/arXiv.2608.16109},
  url          = {https://arxiv.org/abs/2608.16109}
}

@article{journals/corr/abs-2611-04753,
  author       = {Zehan Lin and Shengxin Liu and Biaoshuai Tao and Shengwei Zhou},
  title        = {Logarithmic Comparison-Query Complexity for {EF1} Allocations},
  journal      = {arXiv preprint arXiv:2611.04753},
  year         = {2026},
  url          = {https://arxiv.org/abs/2611.04753}
}

@inproceedings{conf/approx/BhaskarSV21,
  author       = {Umang Bhaskar and
                  A. R. Sricharan and
                  Rohit Vaish},
  title        = {On Approximate Envy-Freeness for Indivisible Chores and Mixed Resources},
  booktitle    = {{APPROX-RANDOM}},
  series       = {LIPIcs},
  volume       = {207},
  pages        = {1:1--1:23},
  publisher    = {Schloss Dagstuhl - Leibniz-Zentrum f{\"{u}}r Informatik},
  year         = {2021}
}

@article{goldman2015spliddit,
  title={Spliddit: Unleashing fair division algorithms},
  author={Goldman, Jonathan and Procaccia, Ariel D},
  journal={ACM SIGecom Exchanges},
  volume={13},
  number={2},
  pages={41--46},
  year={2015},
  publisher={ACM New York, NY, USA}
}

@article{BU2023103904,
  author       = {Xiaolin Bu and
                  Jiaxin Song and
                  Biaoshuai Tao},
  title        = {On existence of truthful fair cake cutting mechanisms},
  journal      = {Artif. Intell.},
  volume       = {319},
  pages        = {103904},
  year         = {2023}
}

@inproceedings{journals/corr/abs-2606-31457,
  author       = {Xinyu Liu and
                  David Kempe and
                  Evi Micha},
  title        = {Learning Fair Allocation of Indivisible Items from Limited Feedback},
  booktitle    = {{EC}},
  publisher    = {{ACM}},
  year         = {2026}
}

@article{journals/corr/EzraG26,
  author       = {Tomer Ezra and Tamar Garbuz},
  title        = {Improved Impossibility Bounds for Maximin Share Allocations},
  journal      = {arXiv preprint arXiv:2609.15085},
  year         = {2026},
  doi          = {10.48550/arXiv.2609.15085},
  url          = {https://arxiv.org/abs/2609.15085}
}

@article{journals/jacm/Clarkson95,
  author       = {Kenneth L. Clarkson},
  title        = {{Las Vegas} Algorithms for Linear and Integer Programming when the
                  Dimension is Small},
  journal      = {J. ACM},
  volume       = {42},
  number       = {2},
  pages        = {488--499},
  year         = {1995},
  doi          = {10.1145/201019.201036},
  bibsource    = {dblp computer science bibliography, https://dblp.org}
}

@article{journals/dcg/BronnimannG95,
  author       = {Herv{\'{e}} Br{\"{o}}nnimann and
                  Michael T. Goodrich},
  title        = {Almost Optimal Set Covers in Finite {VC}-Dimension},
  journal      = {Discret. Comput. Geom.},
  volume       = {14},
  number       = {4},
  pages        = {463--479},
  year         = {1995},
  doi          = {10.1007/BF02570718},
  bibsource    = {dblp computer science bibliography, https://dblp.org}
}
% \newpage
\appendix
\section{Limitation of Comparison-Based Proportionality Oracles}\label{app:prop-oracles-chores}
In this appendix, we construct a counterexample showing that comparison queries are insufficient to implement PROP, PROP1, or PROP1-nonPROP oracles in the chores~setting.

\begin{example}[A unified impossibility example]
Consider an instance with four agents and four chores $e_1,e_2,e_3,e_4$ with identical additive costs.
For a parameter $x\in\{4, 10\}$, let
\begin{equation*}
    c(e_1)=c(e_2)=c(e_3)=1,\qquad c(e_4)=x .
\end{equation*}
Let $B=\{e_1,e_2,e_3\}$.
The cost of $e_4$ is invisible to comparison queries:
for $x=4$ and $x=10$, every comparison between two bundles returns the same
answer.
Indeed, in both cases, any bundle containing $e_4$ is more costly than
any bundle not containing $e_4$, and comparisons among bundles on the same side
depend only on the number of chores from $\{e_1,e_2,e_3\}$.
Thus, any comparison-based oracle must behave identically in the two instances.
However, the proportional threshold depends on the hidden value of $x$.
When $x=4$, after removing any chore from $B$, the remaining cost is
$2>c(M)/4=7/4$, so $B$ is not PROP1 (hence not PROP).
When $x=10$, we have $c(B)=3\leq c(M)/4=13/4$, so $B$ is PROP (hence PROP1).
Consequently, none of these oracles can be implemented using only comparison queries.
\end{example}

\section{A Common Comparison Lower Bound}\label{app:common-lower-bound}
In this appendix, we justify the lower-bound statement made in the preliminary section.

\begin{proposition}\label{prop:common-comparison-lower-bound}
For additive chore instances, any deterministic comparison-based algorithm that is guaranteed to compute a PROP1 allocation, an EF1 allocation, or an $\alpha$-MMS
allocation for a fixed nontrivial approximation factor $1\leq \alpha<n$, requires
$\Omega(\log(m/n))$ comparisons in the worst case.
In particular, when $n$ is constant, the lower bound is $\Omega(\log m)$.
\end{proposition}

\begin{proof}
We give one adversary argument that applies to all three guarantees. Fix an
integer $k\geq 2$, chosen later. The adversary hides a set $H\subseteq M$ of
$k$ costly chores. All agents have the same additive cost function: for every
agent $i\in N$,
\begin{equation*}
    c_i(e)=
    \begin{cases}
    1, & e\in H,\\
    0, & e\notin H .
    \end{cases}
\end{equation*}
The algorithm does not know $H$.
The adversary maintains a candidate set $G\subseteq M$ with the invariant that
every $k$-subset of $G$ is still a possible choice for $H$. Initially $G=M$.
Consider a comparison query $\mathrm{Comp}_i(X,Y)$. Partition $G$ into four
classes:
\begin{equation*}
    G\setminus(X\cup Y),\qquad
    G\cap X\cap Y,\qquad
    G\cap(X\setminus Y),\qquad
    G\cap(Y\setminus X).
\end{equation*}
Let $G'$ be a largest class, so $|G'|\geq |G|/4$. The adversary replaces $G$ by
$G'$ and answers consistently with the hypothesis that all costly chores lie in
$G'$. Namely, if $G'\subseteq G\setminus(X\cup Y)$ or $G'\subseteq X\cap Y$,
then $X$ and $Y$ have equal cost, so the adversary answers according to the
tie-breaking rule. If $G'\subseteq X\setminus Y$, then $X$ is more costly than
$Y$, so the adversary answers $Y$. If $G'\subseteq Y\setminus X$, then the
adversary answers $X$.
Thus, after $q$ comparisons, $|G|\geq  m/4^q$.
If
\begin{equation*}
    q<\log_4\!\left(\frac{m}{nk}\right),
\end{equation*}
then $|G|>nk$.
When the algorithm outputs an allocation $(X_1,\ldots,X_n)$, some bundle $X_i$ contains at least $k$ chores from $G$.
Choose $H\subseteq X_i\cap G$ with $|H|=k$. By the invariant, this choice of $H$ is consistent with the entire comparison transcript. Since the algorithm is deterministic, it outputs the same allocation on this instance, and all $k$ costly chores are assigned to agent $i$.
We now choose $k$ according to the fairness~notion.

For EF1, take $k=2$. Then agent $i$ receives both costly chores, while some other agent receives cost $0$. Even after removing one chore from $X_i$, agent $i$'s cost is at least $1$, so the allocation is not~EF1.

For PROP1, take $k=3$. The proportional share is $3/n$. Agent $i$ receives all
three costly chores, and after removing any one chore her remaining cost is at
least $2$.
Since $3/n < 2$ for $n \geq 2$, the allocation is not PROP1.

For $\alpha$-MMS, take
$k=\lfloor \alpha\rfloor+1$.
Since $1\leq \alpha<n$, we have $k\leq n$ and $k>\alpha$. In the constructed
instance, the MMS value of every agent is $1$: the $k$ costly chores can be placed in distinct bundles, and some bundle must contain a costly chore.
But agent $i$ receives cost $k>\alpha\cdot 1$, so the allocation is not $\alpha$-MMS.
Therefore, any deterministic comparison-based algorithm that is guaranteed to compute one of these allocations must use at least $\log_4\!\left({m}/(nk)\right)$
comparisons in the worst case.
For EF1 and PROP1, $k$ is an absolute constant.
For $\alpha$-MMS, $k=\lfloor\alpha\rfloor+1$, which is constant for every fixed
approximation factor $\alpha$. Hence, the lower bound is $\Omega(\log(m/n))$.
When $n$ is constant, this becomes $\Omega(\log m)$.
\end{proof}

\section{Limitations of the Residual Algorithm}\label{app:mms-discussion}
In this appendix, we provide the obstruction behind the residual algorithm in the following lemma.

\begin{lemma}\label{lemma:hall-limit}
For any $n \geq 3$ and $1 \leq \alpha < \frac{2n}{n+2}$, there exists an instance with a bundle $S$ satisfying $c_1(S) \leq \alpha \cdot \MMS_1$ and $c_i(S) > \alpha \cdot \MMS_i$ for all $i \neq 1$, such that assigning $S$ to agent $1$ fails to yield an $\alpha$-MMS allocation.
\end{lemma}

\begin{proof}
Choose $x$ such that
\begin{equation*}
    \frac{\alpha}{2}<x<1-\frac{\alpha}{n},
\end{equation*}
which is possible because $\alpha<2n/(n+2)$.  Let
$L=\{\ell_1,\ldots,\ell_n\}$ be the large chores and
$S=\{s_1,\ldots,s_n\}$ be the small chores.  The instance is summarized in
Table~\ref{table:residual-hall-instance}.

\begin{table}[!ht]
\caption{Hard instance for the residual Hall approach.}
\label{table:residual-hall-instance}
\centering
\begin{tabular}{@{}ccccccc@{}}
\toprule
\makecell{Agent} & $\ell_1$ & $\cdots$ & $\ell_n$ & $s_1$ & $\cdots$ & $s_n$ \\
\midrule
$1$ & $x$ & $\cdots$ & $x$ & $0$ & $\cdots$ & $0$ \\
$2$ & $x$ & $\cdots$ & $x$ & $1-x$ & $\cdots$ & $1-x$ \\
$\vdots$ & $\vdots$ & & $\vdots$ & $\vdots$ & & $\vdots$ \\
$n$ & $x$ & $\cdots$ & $x$ & $1-x$ & $\cdots$ & $1-x$ \\
\bottomrule
\end{tabular}
\end{table}

It is easy to verify that $\MMS_1 = x$ and $\MMS_i = 1$ for $i \geq 2$.
Now consider the bundle $S$ of all small chores.
For agent $1$, $c_1(S)=0\leq \alpha x=\alpha\cdot \MMS_1$,
whereas for every agent $i\geq 2$,
$
    c_i(S)=n(1-x)>\alpha=\alpha\cdot \MMS_i,
$
so the residual algorithm may allocate $S$ to agent $1$ and remove her from the instance.

After this residual step, all $n$ large chores remain and must be allocated among only $n-1$ agents.
Then by the pigeonhole principle, at least one remaining agent must receive at least two large chores.
If this agent is some $i \geq 2$, her cost becomes at least $2x > \alpha = \alpha \cdot \MMS_i$ by the choice of $x$.
Consequently, the resulting allocation violates the $\alpha$-MMS requirement.
This demonstrates that the residual algorithm faces an intrinsic barrier of $2n/(n+2)$, even with access to \mbox{cardinal information}.
\end{proof}

\section{The Computation of Witness EF1 Allocation}\label{app:one-witness-ef1}
In this appendix, we prove Lemma~\ref{lemma: first_EF1}.
Fix an additive cost function $c$ and let $\Comp$ be its comparison oracle.  Since the subproblem has three identical copies of the same cost function, it suffices to compute a three-partition $(X_1,X_2,X_3)$ and a witness item $e_i\in X_i$ for every
nonempty bundle such~that
\begin{equation*}
    c(X_i \setminus \{e_i\})\leq \min_{j\in[3]}\{c(X_j)\}
    \qquad\text{for every nonempty }X_i.
\end{equation*}
Afterwards, each nonempty bundle is ordered internally by placing its witness
first.

We first give the partial allocation subroutine used by the construction.
In each round, the remaining chores are split into three balanced partitions and the minimum-cost tentative bundle is~kept.

\begin{algorithm}[!ht]
\caption{MinBundleKept}
\KwIn{Instance $\mathcal{I} = ([3], M, c)$}
Initialize $P_i\gets\emptyset$ for all $i\in\{1,2,3\}$, and $R\gets M$\;
\While{$|R|\geq 3$}{
    Compute a partition $(R_1,R_2,R_3)$ of $R$ where
    $\big| |R_i|-|R_j| \big|\leq 1$ for all $i,j \in [3]$\;
    \For{$i\in\{1,2,3\}$}{
        $T_i\gets P_i\cup R_i$\;
    }
    $j\gets \argmin_{i \in [3]}\{c(T_i)\}$ \;
    $P_j\gets T_j$\;
    $R\gets R\setminus R_j$\;
}
\Return $(P_1,P_2,P_3),R$\;
\KwOut{Partial allocation ($P_1$, $P_2$, $P_3$) and remaining items $R$}
\end{algorithm}

We next describe the remaining construction in steps.

\paragraph{Moving-knife Procedure}
{\emergencystretch=2em
Given an ordered sequence $S$ and the benchmark bundle $B$, define
$\textsc{MovingKnife}(S,B)$ as follows.
It binary searches for the \mbox{smallest $\ell$ satisfying}
\begin{equation*}
    \Comp(B,\Pre(S,\ell))=B.
\end{equation*}
It outputs $C=\Pre(S,\ell)$, updates the remaining sequence to $\Suf(S,\ell)$,
and records the last item of $C$ as its witness.
The predicate is monotone because prefix costs only~increase.\par}

\paragraph{Step 1: \textsc{MinBundleKept}}
We first run $\textsc{MinBundleKept}$ and let the output be partial allocation $(P_1,P_2,P_3)$ and remaining set $R$.
If $R=\emptyset$, output $(P_1,P_2,P_3)$ and use any item of each nonempty bundle as its witness.

\paragraph{Step 2: Handle the Zero Benchmark}
Relabel the partial bundles so that $c(P_1)\leq c(P_2)\leq c(P_3)$.
If $c(P_3)=0$, allocate the items in $R$ one by one to a current minimum-cost bundle, and use the newly inserted item as the witness of the updated bundle.

\paragraph{Step 3: Roll Back to the Tentative Partition}
Assume $c(P_3)>0$.  If $|R|=1$, let $r_2$ be the unique item in
$R$; if $|R|=2$, write $R=\{r_1,r_2\}$ with
$c(r_1)\leq c(r_2)$.  Consider the last round of $\textsc{MinBundleKept}(M)$
in which the final bundle $P_3$ was selected.  Let $(B_1,B_2,P_3)$ be the tentative
partition formed in that round, where $P_3$ is the selected bundle.  Relabel
$B_1,B_2$ so that
$r_2\in B_2$, order $B_2$ internally so that $r_2$ is its last item, and order
$B_1$ arbitrarily.  Let $S$ be the ordered sequence of $B_1\cup B_2$ in which
all items of $B_1$ appear before all items of $B_2$.

\paragraph{Step 4: Cut Two Witnessed Bundles}
Run $\textsc{MovingKnife}(S,P_3)$ twice.  This creates bundles $C_1,C_2$ with
witness items $e_1,e_2$, respectively, and leaves a residual sequence still denoted by $S$.
Choose any item of $P_3$ as the witness of $P_3$.  If $S=\emptyset$, output
$(C_1,C_2,P_3)$.

\paragraph{Step 5: Cleanup Exchange}
Assume $S\neq\emptyset$.  If $|R|=2$ and $r_1\notin S$, let $C_i$ be the
bundle among $C_1,C_2$ that contains $r_1$, and let $e_i$ be its current
witness item.  If $c(C_i\setminus \{r_1\})\geq c(P_3)$, replace $C_i$ by
$C_i\setminus \{r_1\}$, keep $e_i$ as its witness, and append $r_1$ to $S$.  Otherwise,
find the shortest prefix $I=\Pre(S,\ell)$ such that
$c((C_i\setminus \{r_1\})\cup I)\geq c(P_3)$, replace $C_i$ by $(C_i\setminus \{r_1\})\cup I$, use the
last item of $I$ as its witness, update
$S=\Suf(S,\ell)$, and append $r_1$ to $S$.

\paragraph{Step 6: Allocate the Remaining Chores}
Let $S_0=S\setminus R$.  Add all chores in $S_0$ to $P_3$.  Then let
$Q=S\cap R$.  Allocate each item in $Q$ to a current minimum-cost bundle,
using the inserted item as the new witness.  Finally, order every nonempty
output bundle internally so that its witness is first.

\begin{proof}[Proof of Lemma~\ref{lemma: first_EF1}]
We prove the correctness and query complexity of the above~construction.

First, \textsc{MinBundleKept} leaves at most two unallocated chores.
Indeed, in every round the selected balanced block has size at least a constant
fraction of the current remaining set.  Hence, after $O(\log m)$ rounds, the
remaining set $R$ has size at most $2$.  Each round uses only $O(1)$
comparisons, since finding the minimum of three tentative bundles requires a
constant number of comparisons.

If $R=\emptyset$, the returned partial bundles already form a
witness EF1 partition.  To see this, observe that every time
\textsc{MinBundleKept} updates a bundle, it keeps a minimum-cost bundle among
three tentative bundles that partition $M$.  Thus, every final partial bundle has
cost at most $c(M)/3$.  Since the three partial bundles now partition all of
$M$, their costs must all be exactly $c(M)/3$.  Therefore, for any nonempty
bundle $P_i$ and any item~$e\in P_i$,
\begin{equation*}
    c(P_i \setminus \{e\})\leq c(P_i)=c(P_j)
    \qquad\text{for every }j\in\{1,2,3\}.
\end{equation*}

Next suppose $R\neq\emptyset$.  After relabeling, if $c(P_3)=0$, then all three partial bundles have
zero cost.  Since at most two chores remain, we allocate them one by one to a
current minimum-cost bundle and use the newly added chore as the witness of the
updated bundle.  If a chore $e$ is added to a bundle $X$ of minimum cost, then
\begin{equation*}
    c((X\cup \{e\})\setminus \{e\})=c(X),
\end{equation*}
and the minimum bundle cost cannot decrease after adding a chore.  Hence, the
witness condition is preserved at each insertion.

It remains to consider the case $c(P_3)>0$.  Consider the last round in which the
final bundle $P_3$ was selected by \textsc{MinBundleKept}.  Let
$(B_1,B_2,P_3)$ be the tentative partition formed in that round.  Since $P_3$ was the
selected minimum-cost tentative bundle, we have
\begin{equation*}
    c(P_3)\leq c(B_1)
    \qquad\text{and}\qquad
    c(P_3)\leq c(B_2).
\end{equation*}
Moreover, $(B_1,B_2,P_3)$ is a partition of $M$, and the final remaining set
$R$ is contained in $B_1\cup B_2$.

Now order $B_1$ and $B_2$ as in the algorithm, and let $S$ be the resulting
ordered sequence of $B_1\cup B_2$ with all items of $B_1$ before all items of
$B_2$.  The first call to \textsc{MovingKnife} finds a prefix $C_1$ whose cost
reaches
$c(P_3)$ for the first time.  Since $c(B_1)\geq c(P_3)$, this cut exists and ends
within $B_1$.  If $e_1$ is the last item of $C_1$, then by the minimality of the
prefix,
\begin{equation*}
    c(C_1)\geq c(P_3)
    \qquad\text{and}\qquad
    c(C_1\setminus \{e_1\})<c(P_3).
\end{equation*}
After removing this prefix, the remaining sequence still contains all of $B_2$.
Since $c(B_2)\geq c(P_3)$, the second call to \textsc{MovingKnife} produces
$C_2$ and witness item $e_2$~with
\begin{equation*}
    c(C_2)\geq c(P_3)
    \qquad\text{and}\qquad
    c(C_2\setminus \{e_2\})<c(P_3).
\end{equation*}
Thus, $(C_1,C_2,P_3)$ already has the desired witness structure on the allocated chores.
If the remaining sequence $S$ is empty, the partition is witness EF1.

Assume from now on that $S\neq\emptyset$.  Because $r_2$ was placed as the last
item of $B_2$, and because the two cuts are prefixes of the initial ordered
sequence, the
remaining sequence contains $r_2$.  If $|R|=2$ and $r_1\notin S$, then
$r_1$ lies in one of $C_1,C_2$, say $C_i$.  If
$c(C_i\setminus \{r_1\})\geq c(P_3)$, we remove $r_1$ from $C_i$ and put it into the remaining
sequence.  The old witness of $C_i$ is still valid: it cannot be $r_1$, because
otherwise $c(C_i\setminus \{r_1\})<c(P_3)$, contradicting~$c(C_i\setminus \{r_1\})\geq c(P_3)$.

If instead $c(C_i\setminus \{r_1\})<c(P_3)$, then the ordered sequence $S$ contains $r_2$,
and $c(r_1)\leq c(r_2)$.  Since $c(C_i)\geq c(P_3)$, adding $r_2$ to
$C_i\setminus \{r_1\}$ also reaches cost $c(P_3)$.  Therefore, some prefix
$I$ of $S$ makes $(C_i\setminus \{r_1\})\cup I$ reach cost $c(P_3)$.
The algorithm chooses the shortest such prefix.
If $e_I$ is the last item of $I$, then
\begin{equation*}
    c((C_i\setminus \{r_1\})\cup I)\geq c(P_3)
    \qquad\text{and}\qquad
    c((C_i\setminus \{r_1\})\cup (I\setminus \{e_I\}))<c(P_3).
\end{equation*}
Thus, replacing $C_i$ by $(C_i\setminus \{r_1\})\cup I$ preserves the witness property, with
$e_I$ as the new witness.  After this cleanup step, either $S\setminus R=\emptyset$
or $R\subseteq S$.

We claim that all remaining chores in $S_0=S\setminus R$ have zero cost.
If $S_0=\emptyset$, this is immediate.  Otherwise, by the conclusion of the
cleanup step, $R\subseteq S$.  Recall that the final output of
\textsc{MinBundleKept} before the cleanup is $(P_1,P_2,P_3)$ with remaining set
$R$, where
\begin{equation*}
    c(P_1)\leq c(P_2)\leq c(P_3).
\end{equation*}
Hence, we have
\begin{equation*}
    c(R)
    =c(M)-c(P_1)-c(P_2)-c(P_3)
    \geq c(M)-3\cdot c(P_3).
\end{equation*}
On the other hand, after the moving-knife and cleanup steps, the bundles
$C_1,C_2,P_3$ together with the remaining sequence $S=R\cup S_0$ partition
$M$.  Since $c(C_1)\geq c(P_3)$ and $c(C_2)\geq c(P_3)$, we obtain
\begin{equation*}
    c(M)
    =c(C_1)+c(C_2)+c(P_3)+c(R)+c(S_0)
    \geq 3 \cdot c(P_3) + c(R)+c(S_0)
    \geq c(M)+c(S_0).
\end{equation*}
Therefore, $c(S_0)=0$.  Adding $S_0$ to $P_3$ does not change the cost of $P_3$ or
any witness~inequality.

Finally, the only unallocated chores are the items in
$Q=S\cap R$, so $|Q|\leq 2$.  We allocate them one by one to a current
minimum-cost bundle.  Suppose an item $e\in Q$ is added to a current
minimum-cost bundle $X$.  Using $e$ as the new witness gives
\begin{equation*}
    c((X\cup \{e\})\setminus \{e\})=c(X).
\end{equation*}
After adding a chore, the minimum bundle cost cannot decrease, so this remains
at most the new minimum bundle cost.  All other bundles keep their old witnesses.
Thus, the witness EF1 invariant is preserved until all chores are~allocated.

Finally, we show that the overall query complexity of the construction is $O(\log m)$.
Specifically, the \textsc{MinBundleKept} phase requires $O(\log m)$ comparisons, following a query complexity analysis similar to that of Lemma~\ref{lemma: phase1}.
Since each execution of \textsc{MovingKnife} performs a binary search on a monotone prefix-cost predicate, it also consumes $O(\log m)$ comparisons.
Furthermore, the cleanup exchange needs at most one additional binary search, while all remaining operations deal with only a constant number of bundles or at most two chores.
Consequently, the algorithm computes a witness EF1 partition using $O(\log m)$ comparison queries in total.
\end{proof}
\end{document}